\documentclass[12pt]{article}
\usepackage{amsmath}
\usepackage[dvipsnames]{xcolor}
\usepackage{amsfonts}
\usepackage[T1]{fontenc}
\usepackage[latin9]{inputenc}
\usepackage[margin=1.2in,bottom=1.2in,top=1.3in]{geometry}
\usepackage{lscape}
\usepackage{setspace}
\usepackage{footmisc}
\usepackage{amsmath}
\usepackage{amsthm}
\usepackage{bm}
\usepackage{amssymb}
\usepackage{stmaryrd}
\usepackage{graphicx}
\usepackage{float}
\usepackage{setspace}
\usepackage{rotating}
\usepackage{subcaption}
\usepackage{ragged2e}
\usepackage{esint}
\usepackage{ulem}
\usepackage{babel}
\usepackage{longtable}
\usepackage{booktabs}
\usepackage{caption,booktabs,array}
\usepackage[round]{natbib}
\usepackage[hypertexnames=false]{hyperref}
\usepackage{colortbl}
\usepackage[toc,page]{appendix}
\usepackage[sc]{mathpazo}
\usepackage{colortbl}
\usepackage{dsfont}
\usepackage{pifont}
\usepackage{pgfplots}
\usepackage{tikz}
\usepackage{enumerate}
\usepackage{titling}
\usepackage{sw20elba}
\usepackage{multibib}
\usepackage{url}

\newcites{supp}{Supplementary References}
\renewcommand{\theequation}{\textcolor{red}{\arabic{equation}}}
\pgfplotsset{compat=1.17}
\definecolor{darkblue}{rgb}{0,0,.6}
\hypersetup{
colorlinks=true,
linkcolor=darkblue,
filecolor=darkblue,
urlcolor=darkblue,
citecolor=darkblue,
}

\newtheorem{assumption}{Assumption}
\newtheorem{definition}{Definition}
\newtheorem{remark}{Remark}

\newtheorem{theorem}{Theorem}

\newtheorem{lemma}{Lemma}
\newtheorem{corollary}{Corollary}
\newtheorem{proposition}{Proposition}
\renewenvironment{proof}[1][Proof]{\noindent \textbf{#1.} }{\  \rule{0.5em}{0.5em}}
\usetikzlibrary{arrows}
\usetikzlibrary{patterns}
  \pgfdeclarepatternformonly{north east lines wide}
        {\pgfqpoint{-1pt}{-1pt}}
        {\pgfqpoint{10pt}{10pt}}
        {\pgfqpoint{9pt}{9pt}}
        {
            \pgfsetlinewidth{0.3pt}
            \pgfpathmoveto{\pgfqpoint{0pt}{0pt}}
            \pgfpathlineto{\pgfqpoint{9.0pt}{9.0pt}}
            \pgfusepath{stroke}
        }
        \pgfdeclarepatternformonly{north west lines wide}
        {\pgfqpoint{-1pt}{1pt}}
        {\pgfqpoint{10pt}{-10pt}}
        {\pgfqpoint{9pt}{-9pt}}
        {
            \pgfsetlinewidth{0.3pt}
            \pgfpathmoveto{\pgfqpoint{0pt}{0pt}}
            \pgfpathlineto{\pgfqpoint{9.0pt}{-9.0pt}}
            \pgfusepath{stroke}
        }
\definecolor{green}{rgb}{0.0,0.7,0.0}
\input{tcilatex}
\begin{document}

\begin{singlespace} 
\title{Identification and Estimation of Unconditional Policy Effects of an
Endogenous Binary Treatment: An Unconditional MTE Approach\thanks{For their constructive feedback, we thank Xiaohong Chen (coeditor), an associate editor, and two anonymous referees. We also thank Marinho Bertanha, Michael Leung, Jessie Li, Xinwei Ma, Augusto Nieto-Barthaburu, Vitor Possebom, Christoph Rothe, Kaspar Wuthrich, and seminar participants at UC Berkeley, UC Davis, UC San Diego, University of Notre Dame, LACEA 2021, NAMES 2022, SEA 2022, the 2023 CEME Conference for Young Econometricians, and the 2\textsuperscript{nd} NorCal Junior Econometricians' Conference for their very helpful comments. This research was conducted with restricted access to
Bureau of Labor Statistics (BLS) data. The views expressed here do not
necessarily reflect the views of the BLS.}}
\author{Julian Martinez-Iriarte\thanks{
Email: jmart425@ucsc.edu. } \\
Department of Economics\\
UC Santa Cruz \and Yixiao Sun\thanks{
Email: yisun@ucsd.edu.} \\
Department of Economics\\
UC San Diego}
\date{\textbf{August 6, 2024}}
\maketitle

\begin{abstract}
This paper studies the identification and estimation of
policy effects when treatment status is binary and endogenous.
We introduce a new class of marginal treatment effects (MTEs)
based on the influence function of the functional underlying the policy
target. We show that an unconditional policy effect can be represented as a
weighted average of the newly defined MTEs over the
individuals who are indifferent about their treatment status. We provide
conditions for point identification of the unconditional policy effects.
When a quantile is the functional of interest, we introduce the UNconditional Instrumental Quantile
Estimator (UNIQUE) and establish its consistency and asymptotic distribution. In the
empirical application, we estimate the effect of changing college enrollment
status, induced by higher tuition subsidy, on the quantiles of the wage
distribution.
\end{abstract}

\thispagestyle{empty} 


\end{singlespace}

\normalem

\clearpage\pagenumbering{arabic}

\section{Introduction}

An unconditional policy effect is the effect of a change in a target
covariate on the unconditional distribution of an outcome variable of
interest.\footnote{%
There are several groups of variables in this framework. The target
covariate is the variable a policy maker aims to change. We often refer to
the target covariate as the treatment or treatment variable. The outcome
variable is the variable that a policy maker ultimately cares about. A
policy maker hopes to change the target covariate in order to achieve a
desired effect on the distribution of the outcome variable. Sometimes, a
policy maker can not change the treatment variable directly and has to
change it by\ intervening some other covariates, which may be referred to as
the policy covariates. There may also be covariates that will not be
intervened.} When the target covariate has a continuous distribution, we may
be interested in shifting its location and evaluating the effect of such a
shift on the distribution of the outcome variable. For example, we may
consider increasing the number of years of education for \emph{every} worker
in order to improve the median of the wage distribution. When the change in
the covariate distribution is small, such an effect may be referred to as
the marginal unconditional policy effect.

In this paper, we consider a binary target covariate that indicates the
treatment status. In this case, a location shift is not possible, and the
only way to change its distribution is to change the proportion of treated
individuals. We analyze the impact of such a marginal change on a general
functional of the distribution of the outcome. For example, when the
functional of interest is the mean of the outcome, this corresponds to the
marginal policy-relevant treatment effect (MPRTE) of \cite{Carneiro2010,
Carneiro2011}. For the case of quantiles, we obtain an unconditional
quantile effect (UQE). Previously, in a seminal contribution, \cite%
{Firpo2009} proposed using\ an unconditional quantile regression (UQR) to
estimate the UQE.\footnote{\cite{mukhin2019} generalizes \cite{Firpo2009} to
allow for non-marginal changes in continuous covariates. Other recent
studies in the case of continuous covariates include \cite{SasakiUraZhang20}
who allow for a high-dimensional setting, \cite{InoueLiXu21} who tackle a
two-sample problem, \cite{MontesRojas} who analyze location-scale and
compensated shifts, and \cite{alejo2022} who propose an alternative
estimation method based on the slopes of (conditional) quantile regression.}
However, we show that their identification strategy can break down under
endogeneity. An extensive analysis of the resulting asymptotic bias of the
UQR estimator is provided.

The first contribution of this paper is to introduce a new class of
unconditional marginal treatment effects (MTEs) and show that the
corresponding unconditional policy effect can be represented as a weighted
average of these unconditional MTEs. The novel MTEs are derived from the
influence function of the functional of the outcome distribution that we
care about. This framework allows us to show that the MPRTE and UQE belong
to the same family of parameters. To the best of our knowledge, this was not
previously recognized in either the literature on MTEs or the literature on
unconditional policy effects.

To illustrate the usefulness of this general approach, we provide an
extensive analysis of the unconditional quantile effects. This is
empirically important since the UQR estimator proposed by \cite{Firpo2009}
is consistent for the UQE only if a certain distributional invariance
assumption holds. Such an assumption is unlikely to hold when the treatment
status is endogenous. We note that treatment endogeneity is the rule rather
than the exception in economic applications.

The second contribution of this paper is to provide a closed-form expression
for the asymptotic bias of the UQR estimator when endogeneity is overlooked.
We show that when selection into treatment follows a threshold-crossing
model, the UQR estimator can be inconsistent, even if the treatment status
is exogenous. This intriguing result underscores the need for caution in
using the UQR without careful consideration. The asymptotic bias can be
traced back to two sources. First, the subpopulation of the individuals at
the margin of indifference might have different characteristics than the
whole population. We refer to this source of bias as the \emph{marginal
heterogeneity bias. }Second, the treatment effect for a marginal individual
might be different from an apparent effect obtained by comparing the
treatment group with the control group. It is the marginal subpopulation,
not the whole population or any other subpopulation, that contributes to the
UQE. We refer to the second source of bias as the \emph{marginal selection
bias. }

The third contribution of this paper is to show that if assumptions similar
to instrument validity are imposed on the policy variables under
intervention, the resulting UQE can be point identified using the local
instrumental variable approach as in \cite{Carneiro2009}. Building on this,
we introduce the UNconditional Instrumental QUantile Estimator (UNIQUE) and
develop methods for statistical inference based on the UNIQUE when the
binary treatment is endogeneous. We take a nonparametric approach but allow
for the propensity score function to be either parametric or nonparametric.
We establish the asymptotic distribution of the UNIQUE. This is a formidable
task, as the UNIQUE is a four-step estimator, and we have to pin down
estimation errors from each step.

\textbf{Related Literature.} This paper is related to the literature on
marginal treatment effects. Introduced by \cite{Bjorklund1987}, MTEs can be
used as a building block for many different causal parameters of interest as
shown in \cite{Heckman2001}. For an excellent review, the reader is referred
to \cite{mogstad2018}. As a main departure from this literature, this paper
introduces unconditional MTEs targeted at studying unconditional policy
effects. As mentioned above, if we focus on the mean, the unconditional
marginal policy effect we study corresponds to the MPRTE of \cite%
{Carneiro2010}. In this case, the unconditional MTE is the same as the usual
MTE. For quantiles, the unconditional marginal policy effect is studied by 
\cite{Firpo2009}, but in a setting that does not allow for endogeneity.%
\footnote{%
For the case of continuous endogenous covariates, \cite{Rothe2010b} shows
that the control function approach of \cite{Imbens2009} can be used to
achieve identification. This paper considers a binary endogenous covariate.}
The unconditional MTE is novel in this case, as well as in the case with a
more general nonlinear functional of interest. In our setting that allows
some covariates to enter both the outcome equation and the selection
equation, conditioning on the propensity score is not enough, but we show
that the propensity score plays a key role in averaging the unconditional
MTEs to obtain the unconditional policy effect. This is in contrast with 
\cite{Zhou2019} where the marginal treatment effect parameter is defined
based on conditioning on the propensity score. Among other contributions, 
\cite{torgo2020} provide an extensive list of applications for MTEs, many of
which can also be applied to the unconditional MTEs introduced in this paper.

Our general treatment of the problem using functionals is closely related to
that of \cite{Rothe2012}. \cite{Rothe2012} analyzes the effect of an
arbitrary change in the distribution of a target covariate, either
continuous or discrete, on some feature of the distribution of the outcome
variable. By assuming a form of conditional independence, for the case of
continuous target covariates, \cite{Rothe2012} generalizes the approach of 
\cite{Firpo2009}. However, for the case of a discrete treatment, instead of
point identifying the effect as we do here, bounds are obtained by assuming
that either the highest-ranked or lowest-ranked individuals enter the
program under the new policy.

We are not the first to consider unconditional quantile regressions under
endogeneity.\footnote{\cite{pereda2023} studies the unconditional quantile
treatment effect of a binary endogenous regressor. However, the effect there
differs from the unconditional quantile effect considered in this paper.
Here, unconditional refers to the marginal distribution of the observed $Y$,
while in \cite{pereda2023} unconditional means that the distribution of the
potential outcomes after the covariates are integrated out.} \cite{Kasy2016}
focuses on ranking counterfactual policies and, for the case of discrete
regressors, allows for endogeneity. However, a key difference from our
approach is that the counterfactual policies analyzed in \cite{Kasy2016} are
randomly assigned conditional on a covariate vector. In our setting,
selection into treatment follows a threshold-crossing model, where we use
the exogenous variation of an instrument to obtain different counterfactual
scenarios. \cite{martinez2020} introduces the quantile breakdown frontier in
order to perform a sensitivity analysis on departures from the
distributional invariance assumption employed by \cite{Firpo2009}. 

Estimation of the MTE and parameters derived from the MTE curve is discussed
in \cite{urzua2006}, \cite{Carneiro2009}, \cite{Carneiro2010, Carneiro2011},
and \cite{ura2021}. All of these studies make a linear-in-parameters
assumption regarding the conditional means of the potential outcomes, which
yields a tractable partially linear model for the MTE curve. This strategy
is not very helpful in our case because our newly defined unconditional MTE
might involve a nonlinear function of the potential outcomes. For example,
in the case of quantiles, there is an indicator function involved. Using our
expression for the weights, we can write the UQE as a quotient of two
average derivatives. One of them, however, involves as a regressor the
estimated propensity score as in the setting of \cite{hahn2013}. We provide
conditions, different from those in \cite{hahn2013}, under which the error
from estimating the propensity score function, either parametrically or
nonparametrically, does not affect the asymptotic variance of the UNIQUE.
This may be of independent interest.

\textbf{Outline.} Section \ref{general_functional} introduces the new MTE
curve and shows how it relates to the unconditional policy effect. Section %
\ref{section_uqr} presents a model for studying the UQE under endogeneity.
Section \ref{sect:IV} considers intervening an instrumental variable in
order to change the treatment status and establishes the identification of
the corresponding UQE. Section \ref{estimation} introduces and studies the
UNIQUE under a parametric specification of the propensity score. Section \ref%
{simulation} provides simulation evidence. In Section \ref{empirical} we
revisit the empirical application of \cite{Carneiro2011} and focus on the
unconditional quantile effect. Section \ref{conclusion} concludes. An
appendix contains the proof of the main results. A supplementary appendix
provides the technical conditions for two lemmas, the proof of all lemmas
and a proposition, the estimation of the asymptotic variance for the UNIQUE,
and its asymptotic properties under a nonparametric specification of the
propensity score.

\textbf{Notation.} For any generic random variable $W_{1}$, we denote its
CDF and pdf by $F_{W_{1}}\left( \cdot \right) $ and $f_{W_{1}}\left( \cdot
\right) ,$ respectively. We denote its conditional CDF and pdf conditional
on a second random variable $W_{2}$ by $F_{W_{1}|W_{2}}\left( \cdot |\cdot
\right) $ and $f_{W_{1}|W_{2}}\left( \cdot |\cdot \right) ,$ respectively.

\section{Unconditional Policy Effects under Endogeneity}

\label{general_functional}

\subsection{Policy Intervention}

We employ the potential outcomes framework. For each individual, there are
two potential outcomes: $Y(0)$ and $Y(1)$, where $Y(0)$ is the outcome had
she received no treatment and $Y(1)$ is the outcome had she received
treatment. We assume that the potential outcomes are given by 
\begin{equation*}
Y(0)=r_{0}(X,U_{0}),\text{ and }Y(1)=r_{1}(X,U_{1}),
\end{equation*}%
for a pair of unknown functions $r_{0}$ and $r_{1}$. The vector $X\in 
\mathbb{R}^{d_{X}}$ consists of observables and $U:=\left( U_{0}^{\prime
},U_{1}^{\prime }\right) ^{\prime }$ consists of unobservables. Depending on
the individual's actual choice of treatment, denoted by $D,$ we observe
either $Y(0)$ or $Y(1)$, but we can never observe both. The observed outcome
is denoted by $Y$:%
\begin{equation}
Y=\left( 1-D\right) Y(0)+DY(1)=\left( 1-D\right)
r_{0}(X,U_{0})+Dr_{1}(X,U_{1}).  \label{model1Y}
\end{equation}

Following \cite{Heckman1999, Heckman2001, Heckman2005}, we assume that
selection into treatment is determined by a threshold-crossing equation 
\begin{equation}
D=\mathds{1}\left\{ V\leq \mu \left( W\right) \right\} ,  \label{model1D}
\end{equation}%
where $W:=(Z,X)$ and $Z\in \mathbb{R}^{d_{Z}}$ consists of covariates that
do not affect the potential outcomes directly. In the above, the unknown
function $\mu \left( W\right) $ can be regarded as the benefit from the
treatment and $V$ as the cost of the treatment. Individuals decide to take
up the treatment if and only if its benefit outweighs its cost. 
While we observe $\left( D,W,Y\right) $, we observe neither $U$ nor $V.$
Also, we do not restrict the dependence among $U,W,$ and $V$. Hence, they
can be mutually dependent and $D$ could be endogenous. Moreover, the
dimension of $U$ is left unrestricted, while $V$ is a real-valued random
variable.

The propensity score is $P(w):=\Pr \left[ D=1|W=w\right] $. In view of %
\eqref{model1D}, we can represent it as%
\begin{equation}
P(w)=\Pr \left[ V\leq \mu \left( W\right) |W=w\right] =F_{V|W}(\mu (w)|w).
\label{PS_equ}
\end{equation}%
If the conditional CDF $F_{V|W}(\cdot |w)$ is a strictly increasing function
for all $w\in \mathcal{W}$, the support of $W$, we have%
\begin{equation*}
D=\mathds{1}\left\{ V\leq \mu \left( W\right) \right\} =\mathds{1}\left\{
F_{V|W}(V|W)\leq F_{V|W}(\mu \left( W\right) |W)\right\} =\mathds{1}\left\{
U_{D}\leq P(W)\right\} ,
\end{equation*}%
where $U_{D}:=F_{V|W}(V|W)$ measures an individual's relative resistance to
the treatment, and it can be shown that $U_{D}$ is uniform on $[0,1]$ and is
independent of $W.$

To change the treatment take-up rate, we manipulate $Z,$ a subvector of $W$.%
\footnote{%
When we induce the covariate $Z$ to change, the distribution of this
covariate will change. However, we do not specify the new distribution \emph{%
a priori}. Instead, we specify the policy rule that dictates how the value
of the covariate will change for each individual in the population. Our
intervention may then be regarded as a \emph{value} intervention. This is in
contrast to a \emph{distribution} intervention that stipulates a new
covariate distribution directly. An advantage of our policy rule is that it
is directly implementable in practice while a hypothetical distribution
intervention is not. The latter intervention may still have to be
implemented via a value intervention, which is our focus here. For a similar
comment in the continuous case, see Section 3 in \cite{MontesRojas}.} More
specifically, we consider a policy intervention that changes $Z$ into $%
Z_{\delta }=\mathcal{G}\left( W,\delta \right) $ for a vector of smooth
functions $\mathcal{G}\left( \cdot ,\cdot \right) \in \mathbb{R}^{d_{Z}}$.
We assume that $\mathcal{G}(W,0)=Z$ so that the \emph{status quo} policy
corresponds to $\delta =0.$\footnote{%
For notational convenience, when $\delta =0$, we drop the subscript and
denote $Y_{0}$ and $D_{0}$ as $Y$ and $D,$ respectively.} With the induced
change in $Z,$ the selection equation becomes 
\begin{equation}
D_{\delta }=\mathds{1}\left\{ V\leq \mu \left( Z_{\delta },X\right) \right\}
=\mathds{1}\left\{ V\leq \mu \left( \mathcal{G}(W,\delta ),X\right) \right\}
,  \label{model2D}
\end{equation}%
which can be written as $D_{\delta }=\mathds{1}\left\{ U_{D}\leq P_{\delta
}(W)\right\} $, for $P_{\delta }(W)=F_{V|W}(\mu \left( \mathcal{G}(W,\delta
),X\right) |W)$.\footnote{%
Note that $U_{D}$ is still defined as $F_{V|W}(V|W)$, and so it does not
change under the counterfactual policy regime. This is to say that, relative
to others, an individual's resistance to the treatment is preserved across
the two policy regimes. In particular, $U_{D}$ is still uniform on $[0,1]$
and independent of $W.$} The outcome equation, in turn, is now 
\begin{equation}
Y_{\delta }=\left( 1-D_{\delta }\right) Y\left( 0\right) +D_{\delta }Y\left(
1\right) =\left( 1-D_{\delta }\right) r_{0}(X,U_{0})+D_{\delta
}r_{1}(X,U_{1}).  \label{model2Y}
\end{equation}

Equations (\ref{model2D}) and (\ref{model2Y}) are the same as the \emph{%
status quo} equations; the only exception is that $Z$ has been replaced by $%
Z_{\delta }.$ We have maintained the structural forms of the outcome
equation and the treatment selection equation. Importantly, we have also
maintained the stochastic dependence among $U,W,$ and $V,$ which is
manifested through the use of the same notation $U,W,$ and $V$ in equations (%
\ref{model2D}) and (\ref{model2Y}) as in equations (\ref{model1Y}) and (\ref%
{model1D}). Our policy intervention has a \emph{ceteris paribus}
interpretation at the population level: we apply the same form of
intervention on $Z$ for all individuals in the population but hold all else,
including the causal mechanism and the stochastic dependence among the \emph{%
status quo} variables, constant. In particular, the conditional distribution
of $\left( U,V\right) $ given $W$ is invariant to the value of $\delta .$

We allow $\mathcal{G}(\cdot ,\delta )$ to take a general form but some
examples may be helpful. Consider the case that $\mu \left( W\right) =Z$ for
a univariate $Z.$ We may take $\mathcal{G}\left( W,\delta \right) =Z+\delta $%
. Under such a policy, we change the value of $Z$ by $\delta $ for each
individual in the population so that the location of the distribution of $Z$
is shifted by $\delta .$ We may also take $\mathcal{G}\left( W,\delta
\right) =Z\left( 1+\delta \right) $, in which case, there is a proportional
or scale change in the value of $Z$ for all individuals. Both cases have
been studied by \cite{Carneiro2010}; see also table 2 in \cite{mogstad2018}.
More general location and scale changes, such as those given in \cite%
{MontesRojas}, are allowed. In fact, any policy function that satisfies
Assumption \ref{Assumption_regularity} in the next subsection is permitted.

\subsection{Unconditional Policy Effects}

To define an unconditional policy effect, we first describe the space of
distributions and the functional of interest on this space. Let $\mathcal{F}%
^{\ast }$ be the space of finite signed measures $\nu $ on $\mathcal{Y}%
\subseteq \mathbb{R}$ with distribution function $F_{\nu }\left( y\right)
=\nu (-\infty ,y]$ for $y\in \mathcal{Y}$. We endow $\mathcal{F}^{\ast }$
with the usual supremum norm: for two distribution functions $F_{\nu _{1}}$
and $F_{\nu _{2}}$ associated with the respective signed measures $\nu _{1}$
and $\nu _{2}$ on $\mathcal{Y}$, we define $\left\Vert F_{\nu _{1}}-F_{\nu
_{2}}\right\Vert _{\infty }:=\sup_{y\in \mathcal{Y}}\left\vert F_{\nu
_{1}}\left( y\right) -F_{\nu _{2}}\left( y\right) \right\vert .$ With some
abuse of notation, denote $F_{Y}$ as the distribution of $Y:$ $F_{Y}\left(
y\right) =\nu _{Y}(-\infty ,y]$ where $\nu _{Y}$ is the measure induced by
the distribution of $Y.$ Define $F_{Y_{\delta }}$ similarly. Clearly, both $%
F_{Y}$ and $F_{Y_{\delta }}$ belong to $\mathcal{F}^{\ast }.$ We consider a
general functional $\rho :\mathcal{F}^{\ast }\rightarrow \mathbb{R}$ and
study the general unconditional policy effect.


\begin{definition}
\label{gue}\textbf{General Unconditional Policy Effect}

\item The general unconditional policy effect for the functional $\rho $ is
defined as 
\begin{equation*}
\Pi _{\rho }:=\lim_{\delta \rightarrow 0}\frac{\rho [F_{Y_{\delta }}]- \rho[%
F_{Y}]}{E[D_{\delta }]-E[D]},
\end{equation*}%
whenever this limit exists.
\end{definition}

In this paper, we consider a Hadamard differentiable functional and its
associated influence function.\footnote{%
An earlier working paper \cite{sun2021} considers the mean functional under
different assumptions since the mean functional is not Hadamard
differentiable.} For completeness, we provide the definitions of Hadamard
differentiability and the influence function below.

\begin{definition}
\label{Def: Hadamard}$\rho :\mathcal{F}^{\ast }\rightarrow \mathbb{R}$ is
Hadamard differentiable at $F\in \mathcal{F}^{\ast }$ if there exists a
linear and continuous functional $\dot{\rho}_{F}:\mathcal{F}^{\ast
}\rightarrow \mathbb{R}$ such that for any $G\in \mathcal{F}^{\ast }$ and $%
G_{\delta }\in \mathcal{F}^{\ast }$ with $\lim_{\delta \rightarrow
0}\left\Vert G_{\delta }-G\right\Vert _{\infty }=0,$ we have 
\begin{equation*}
\lim_{\delta \rightarrow 0}\frac{\rho \lbrack F+\delta G_{\delta }]-\rho
\lbrack F]}{\delta }=\dot{\rho}_{F}\left[ G\right] .
\end{equation*}
\end{definition}

\begin{definition}
\label{Def: IF}The influence function of $\rho :\mathcal{F}^{\ast
}\rightarrow \mathbb{R}$ at $F\in \mathcal{F}^{\ast }$ is given by 
\begin{equation*}
\psi (y,\rho ,F):=\lim_{\epsilon \rightarrow 0+}\frac{\rho \left[ \left(
1-\epsilon \right) F+\epsilon \Delta _{y}\right] -\rho \left[ F\right] }{%
\epsilon },
\end{equation*}%
where $\Delta _{y}$ is the distribution function that assigns all
probability mass to the single point $\left\{ y\right\} ,$ that is, $\Delta
_{y}\left( x\right) =1\left\{ x\geq y\right\} .$
\end{definition}

To see how we can use the Hadamard differentiability to obtain $\Pi _{\rho }$%
, we write 
\begin{equation*}
\frac{\rho \lbrack F_{Y_{\delta }}]-\rho \lbrack F_{Y}]}{\delta }=\frac{\rho %
\left[ F_{Y}+\delta G_{\delta }\right] -\rho \lbrack F_{Y}]}{\delta },
\end{equation*}%
for 
\begin{equation}
G_{\delta }=\frac{F_{Y_{\delta }}-F_{Y}}{\delta }.  \label{G_Delta}
\end{equation}%
As long as we can show that $\lim_{\delta \rightarrow 0}\left\Vert G_{\delta
}-G\right\Vert _{\infty }=0$ for some $G,$ then, we obtain 
\begin{equation*}
\Pi _{\rho }=\lim_{\delta \rightarrow 0}\frac{\rho \lbrack F_{Y_{\delta
}}]-\rho \lbrack F_{Y}]}{\delta }\left( \frac{E[D_{\delta }]-E[D]}{\delta }%
\right) ^{-1}=\dot{\rho}_{F_{Y}}\left[ G\right] \left( \lim_{\delta
\rightarrow 0}\frac{E[D_{\delta }]-E[D]}{\delta }\right) ^{-1}.
\end{equation*}%
In the proof of Theorem \ref{Theorem general} below, we show that we can use
the influence function to represent $\dot{\rho}_{F_{Y}}\left[ G\right] $ as $%
\dot{\rho}_{F_{Y}}\left[ G\right] =\int_{\mathcal{Y}}\psi (y,\rho
,F_{Y})dG\left( y\right) $. 

Next, we provide sufficient conditions for $\lim_{\delta \rightarrow
0}\left\Vert G_{\delta }-G\right\Vert _{\infty }=0$. We first formalize two
primary assumptions.\footnote{%
In Assumptions \ref{Assumption_primary} and \ref{Assumption_regularity},
\textquotedblleft for all $w\in \mathcal{W}$\textquotedblright\ can be
replaced by \textquotedblleft for almost all $w\in \mathcal{W}$%
\textquotedblright , and the supremum over $w\in \mathcal{W} $ can be
replaced by the essential supremum over $w\in \mathcal{W}$.}

\begin{assumption}
\textbf{Primary Assumptions} \label{Assumption_primary}

\begin{enumerate}[(a)]%

\item \label{uniformity}$U_{D}$ is independent of $W$ and is uniformly
distributed on $\left[ 0,1\right] .$

\item \label{feasibility}For all $w=(z^{\prime },x^{\prime })^{\prime }\in 
\mathcal{W}$, $\mathcal{G}(w,0)=z,$ and for sufficiently small $\delta $, $%
\mathcal{G}(w,\delta )\in \mathcal{Z}$, the support of $Z.$

\end{enumerate}%
\end{assumption}

As discussed earlier, Assumption \ref{Assumption_primary}(\ref{uniformity})
holds if the conditional CDF $F_{V|W}(\cdot |w)$ is a strictly increasing
function for all $w\in \mathcal{W}$. Assumption \ref{Assumption_primary}(\ref%
{feasibility}) requires that the policy function $\mathcal{G}(w,\delta )$ be
feasible. This assumption, along with Assumption \ref{Assumption_regularity}(%
\ref{p_delta}.ii) below, requires that the subvector of $Z$ undergoing
nontrivial intervention consists of continuous random variables. The
variables in $X$ and the part of $Z$ not subject to intervention do not need
to be continuous random variables.



\begin{assumption}
\textbf{Regularity Conditions} \label{Assumption_regularity}

\begin{enumerate}[(a)]%

\item \label{regularity_x_abs}For $d=0,1,$ the conditional distribution of $%
(Y(d),U_{D})$ conditional on $W=w\in \mathcal{W}$ is absolutely continuous
with conditional density function given by $f_{Y(d),U_{D}|W}(y,u|w).$

\item \label{f_y_u_x} (i) For $d=0,1,$ $u\mapsto f_{Y(d)|U_{D},W}(y|u,w)$ is
continuous for all $y\in \mathcal{Y}\left( d\right) $ and all $w\in \mathcal{%
W}.$ (ii) For $d=0,1,$ $\sup_{y\in \mathcal{Y}\left( d\right) }\sup_{w\in 
\mathcal{W}}\sup_{\delta \in N_{\varepsilon }}f_{Y(d)|U_{D},W}(y|P_{\delta
}(w),w)<\infty $ where $N_{\varepsilon }:=\left\{ \delta :\left\vert \delta
\right\vert \leq \varepsilon \right\} $ for some $\varepsilon >0.$

\item \label{p_delta}(i) For all $w\in \mathcal{W}$, $P(w)\in (0,1).$ (ii)
For all $w\in \mathcal{W}$, the map $\delta \mapsto P_{\delta }(w)$ is
continuously differentiable on $N_{\varepsilon }.$ (iii) $\sup_{w\in 
\mathcal{W}}\sup_{\delta \in N_{\varepsilon }}\left\vert \frac{\partial
P_{\delta }(w)}{\partial \delta }\right\vert <\infty .$

\end{enumerate}%
\end{assumption}

\begin{assumption}
\textbf{Domination Conditions} \label{Assumption_domination} For $d=0,1,$ 
\begin{equation*}
\int_{\mathcal{Y}(d)}\sup_{\delta \in N_{\varepsilon }}f_{Y(d)|D_{\delta
}}(y|d)dy<\infty \text{ and }\int_{\mathcal{Y}(d)}\sup_{\delta \in
N_{\varepsilon }}\left\vert \frac{\partial f_{Y(d)|D_{\delta }}(y|d)}{%
\partial \delta }\right\vert dy<\infty .
\end{equation*}
\end{assumption}

Under Assumption \ref{Assumption_regularity}(\ref{f_y_u_x}, \ref{p_delta}),
we have $E\left[ P_{\delta }(W)\right] =E\left[ P(W)\right] +E\left[ \dot{P}%
\left( W\right) \right] \delta +o(\delta ),$ as $\delta \rightarrow 0$,
where 
\begin{equation*}
\dot{P}\left( w\right) =\left. \frac{\partial P_{\delta }\left( w\right) }{%
\partial \delta }\right\vert _{\delta =0}.
\end{equation*}%
So, under the new policy regime $\mathcal{G}(\cdot ,\delta )$, the
participation rate in the treatment is increased by {approximately }$E\left[ 
\dot{P}\left( W\right) \right] \delta .$ Locally at $\delta =0$, the policy
function $\mathcal{G}(\cdot ,\delta )$ affects the participation rate via $%
\dot{P}\left( \cdot \right) $, which is the rate of change in the propensity
score at $\delta =0$. The function $\dot{P}\left( \cdot \right) $ depends on
the policy function $\mathcal{G}(\cdot ,\delta )$ used, and a more
cumbersome notation for $\dot{P}\left( \cdot \right) $ is $\dot{P}_{\mathcal{%
G}}\left( \cdot \right) .$ For notational simplicity, we suppress such
dependence. Section \ref{section_dot_p} provides further analysis on $\dot{P}%
\left( \cdot \right) $.

\begin{theorem}
\label{Theorem general}Let Assumptions \ref{Assumption_primary}--\ref%
{Assumption_domination} hold. Assume further that $\rho :\mathcal{F}^{\ast
}\rightarrow \mathbb{R}$ is Hadamard differentiable with influence function $%
\psi $, and $E\left[ \dot{P}\left( W\right) \right] \neq 0$. Then, 
\begin{equation*}
\Pi _{\rho }=\int_{\mathcal{W}}\mathrm{MTE}_{\rho }\left( P(w),w\right) 
\mathcal{\dot{P}}\left( w\right) dF_{W}(w),
\end{equation*}%
where 
\begin{equation}
\mathrm{MTE}_{\rho }\left( u_{D},w\right) :=E\left[ \psi (Y\left( 1\right)
,\rho ,F_{Y})-\psi (Y\left( 0\right) ,\rho ,F_{Y})|U_{D}=u_{D},W=w\right] ,
\label{MTE_rho}
\end{equation}%
is the unconditional marginal treatment effect for the $\rho $ functional,
and 
\begin{equation*}
\mathcal{\dot{P}}\left( w\right) =\frac{\dot{P}\left( w\right) }{E\left[ 
\dot{P}\left( W\right) \right] }.
\end{equation*}%
%
%
%
%
%
%
%
%
%
%
%
%
%
%
%
%
%
%
%
%
%
%
%
%
%
%
%
%
%
%
%
%
\end{theorem}



Theorem \ref{Theorem general} reveals that the general unconditional policy
effect $\Pi _{\rho }$ is composed of two key elements: the unconditional $%
\mathrm{MTE}_{\rho }\left( u_{D},w\right) $ evaluated at $u_{D}=P(w)$ and
the weighting function $\mathcal{\dot{P}}\left( w\right) .$ To understand
the first element, consider the group of individuals with the same value $w$
of $W.$ Within this group, those for whom $u_{D}=P\left( w\right) $ are
indifferent between participating and not participating. A small incentive
will induce a change in the treatment status for and only for this subgroup
of individuals. It is the change in their treatment status, and hence the
change in the composition of $Y(1)$ and $Y(0)$ in the observed outcome $Y,$
that changes its unconditional characteristics, such as the quantiles. We
refer to the individuals for whom\ $u_{D}=P\left( w\right) $ as the marginal
subpopulation. As for the second element, we defer the discussion to Section %
\ref{section_dot_p}.

\subsection{Unconditional quantile effect}

Throughout the rest of this paper, we consider the case that $\rho $ is a
quantile functional at the quantile level $\tau \in \left( 0,1\right) ,$
that is, $\rho _{\tau }[F_{Y}]=F_{Y}^{-1}(\tau ):=\inf_{y}\left\{ y\in 
\mathcal{Y}:F_{Y}(y)\geq \tau \right\} .$ Here we have added a subscript $%
\tau $ to $\rho $ to signify the quantile level under consideration. We are
interested in how an improvement in the treatment take-up rate affects the $%
\tau $-quantile of the (unconditional) outcome distribution.

\begin{definition}
\textbf{Unconditional Quantile Effect}

\item The unconditional quantile effect (UQE) is defined as 
\begin{equation*}
\Pi _{\tau }:=\lim_{\delta \rightarrow 0}\frac{\rho _{\tau }[F_{Y_{\delta
}}]-\rho _{\tau }[F_{Y}]}{E[D_{\delta }]-E[D]}
\end{equation*}%
whenever this limit exists.
\end{definition}

Let $y_{\tau }\equiv \rho _{\tau }[F_{Y}]$ be the $\tau $-quantile of $Y.$
If $f_{Y}(y_{\tau })>0,$ then, under Assumption \ref{Assumption_regularity}(%
\ref{regularity_x_abs}), the influence function of the $\tau $-quantile
functional is\footnote{%
Strictly speaking, $\psi \left( y,\rho _{\tau },F_{Y}\right) =0$ for $%
y=y_{\tau }$. However, redefining $\psi \left( y,\rho _{\tau },F_{Y}\right) $
at one point has no consequence on our results, as $F_{Y}\left( \cdot
\right) $ is absolutely continuous under Assumption \ref%
{Assumption_regularity} (\ref{regularity_x_abs}).} 
\begin{equation*}
\psi \left( y,\rho _{\tau },F_{Y}\right) =\frac{1}{f_{Y}\left( y_{\tau
}\right) }\left[ \tau -\mathds{1}\left\{ y\leq y_{\tau }\right\} \right] .
\end{equation*}%
Plugging this influence function into (\ref{MTE_rho}), we obtain the
unconditional marginal treatment effect for the $\tau $-quantile.

\begin{definition}
\label{MTE_tau} The unconditional marginal treatment effect for the $\tau $%
-quantile is 
\begin{equation*}
\mathrm{MTE}_{\tau }\left( u,w\right) =\frac{1}{f_{Y}\left( y_{\tau }\right) 
}E\left[ \mathds{1}\left\{ Y(0)\leq y_{\tau }\right\} -\mathds{1}\left\{
Y(1)\leq y_{\tau }\right\} \mid U_{D}=u,W=w\right] .
\end{equation*}
\end{definition}

The $\mathrm{MTE}_{\tau }$ defined above is a basic building block for the
unconditional quantile effect. It is different from the quantile analogue of
the marginal treatment effect of \cite{Carneiro2009} and \cite{ping2014},
which is defined as $F_{Y(1)|U_{D},W}^{-1}(\tau
|u,w)-F_{Y(0)|U_{D},W}^{-1}(\tau |u,w)$. An unconditional quantile effect
can not be represented as an integrated version of the latter. The next
corollary follows from applying Theorem \ref{Theorem general} to $\rho_\tau$.

\begin{corollary}
\label{Corollary_UQTE0} Let Assumptions \ref{Assumption_primary}--\ref%
{Assumption_domination} hold. Assume further that $f_{Y}(y_{\tau })>0$.
Then, 
\begin{eqnarray}
\Pi _{\tau } &=&\int_{\mathcal{W}}\mathrm{MTE}_{\tau }\left( P(w),w\right) 
\mathcal{\dot{P}}\left( w\right) dF_{W}(w)  \notag \\
&=&\frac{1}{f_{Y}(y_{\tau })}\int_{\mathcal{W}}E\left[ \mathds{1}\left\{
Y(0)\leq y_{\tau }\right\} |U_{D}=P\left( w\right) ,W=w\right] \mathcal{\dot{%
P}}\left( w\right) dF_{W}\left( w\right)  \notag \\
&-&\frac{1}{f_{Y}(y_{\tau })}\int_{\mathcal{W}}E\left[ \mathds{1}\left\{
Y(1)\leq y_{\tau }\right\} |U_{D}=P\left( w\right) ,W=w\right] \mathcal{\dot{%
P}}\left( w\right) dF_{W}\left( w\right) .  \label{UQTE_endogeneity}
\end{eqnarray}
\end{corollary}

\subsection{Understanding $\mathrm{MTE}_{\protect\tau }\left( u,w\right)$}

To understand $\mathrm{MTE}_{\tau }\left( u,w\right) ,$ we define $\Delta
(y_{\tau }):=\left( \mathds{1}\left\{ Y(0)\leq y_{\tau }\right\} -\mathds{1}%
\left\{ Y(1)\leq y_{\tau }\right\} \right) /f_{Y}(y_{\tau }),$ which
underlies the above definition of $\mathrm{MTE}_{\tau }\left( u,w\right) $.
The random variable $\Delta (y_{\tau })$ can take three values: 
\begin{equation*}
\Delta (y_{\tau })=%
\begin{cases}
\begin{aligned} &f_Y(y_\tau)^{-1}&\text{ if }&Y(0)\leq y_{\tau }\text{ and
}Y(1)>y_{\tau } \\ &0&\text{if }&\bigg[Y(0)>y_{\tau }\text{ and
}Y(1)>y_{\tau }\bigg]\text{ or }\bigg[Y(0)\leq y_{\tau }\text{ and }Y(1)\leq
y_{\tau }\bigg] \\ &-f_Y(y_\tau)^{-1}&\text{ if }&Y(0)>y_{\tau }\text{ and
}Y(1)\leq y_{\tau }\end{aligned}%
\end{cases}%
\end{equation*}%
For a given individual, $\Delta (y_{\tau })=1/f_{Y}(y_{\tau })$ when the
treatment induces the individual to \textquotedblleft
cross\textquotedblright\ the $\tau $-quantile $y_{\tau }$ of $Y$ from below,
and $\Delta (y_{\tau })=-1/f_{Y}(y_{\tau })$ when the treatment induces the
individual to \textquotedblleft cross\textquotedblright\ the $\tau $%
-quantile $y_{\tau }$ of $Y$ from above. In the first case, the individual
benefits from the treatment, while in the second case, the treatment harms
her. The intermediate case, $\Delta (y_{\tau })=0,$ occurs when the
treatment induces no quantile crossing of any type. Thus, the unconditional
expected value $f_{Y}(y_{\tau })\times E[\Delta (y_{\tau })]$ equals the
difference between the proportion of individuals who benefit from the
treatment and the proportion of individuals who are harmed by it. For the
UQE, whether the treatment is beneficial or harmful is measured in terms of
quantile crossing. Among the individuals with characteristics $U_{D}=u$ and $%
W=w$, $\mathrm{MTE}_{\tau }\left( u,w\right) $ is then equal to the rescaled
(by $1/f_{Y}(y_{\tau })$) difference between the proportion of individuals
who benefit from the treatment and the proportion of individuals who are
harmed by it. Thus, $\mathrm{MTE}_{\tau }\left( u,w\right) $ is positive if
the treatment leads to a greater number of individuals improving their
outcome above the threshold $y_{\tau }$, compared to those whose outcome
falls below $y_{\tau }$. Conversely, $\mathrm{MTE}_{\tau }\left( u,w\right) $
is negative if the treatment leads to a greater number of individuals
experiencing a decline in their outcome, falling below $y_{\tau }$, compared
to those experiencing an increase above $y_{\tau }.$

\subsection{Understanding $\mathcal{\dot{P}}$}

\label{section_dot_p}

To understand the weighting function $\mathcal{\dot{P}}\left( w\right) ,$
consider the case $\delta >0$, $P_{\delta }(w)\geq P(w)$ for all $w\in W$,
and $F_{W}\left( \cdot \right) $ is absolutely continuous with density $%
f_{W}\left( \cdot \right) .$ Let $\epsilon $ be a small positive number.
Then, $f_{W}\left( w\right) \epsilon $ measures the proportion of
individuals for whom $W$ is in $\left[ w-\epsilon /2,w+\epsilon /2\right] .$
Note that for $W\in \left[ w-\epsilon /2,w+\epsilon /2\right] ,$ the
propensity scores under $D$ and $D_{\delta }$ are approximately $P(w)$ and $%
P_{\delta }\left( w\right) $. The proportion of the individuals for whom $%
W\in \left[ w-\epsilon /2,w+\epsilon /2\right] $ and who have switched their
treatment status from $0$ to $1$ is then equal to $\left[ P_{\delta }(w)-P(w)%
\right] f_{W}\left( w\right) \cdot \epsilon .$ Scaling this by $E[D_{\delta
}]-E[D]$, which is the overall proportion of the individuals who have
switched the treatment status, we obtain 
\begin{equation*}
\frac{\left[ P_{\delta }(w)-P(w)\right] f_{W}\left( w\right) }{E[D_{\delta
}]-E[D]}\cdot \epsilon .
\end{equation*}%
Thus, we can regard $\left[ P_{\delta }(w)-P(w)\right] f_{W}\left( w\right) /%
\left[ E[D_{\delta }]-E[D]\right] $ as the density function of $W$ among
those who have switched their treatment status from $0$ to $1$ as a result
of the policy intervention. On the one hand, when $\delta \rightarrow 0,$
the set of individuals who change their treatment status are the individuals
on the margin, namely, those for whom $u_{D}=P(w).$ On the other hand, when $%
\delta \rightarrow 0,$ the density function approaches $\mathcal{\dot{P}}%
\left( w\right) f_{W}\left( w\right) .$ So $\mathcal{\dot{P}}\left( w\right)
f_{W}\left( w\right) $ is the probability density function (with respect to
the Lebesgue measure) of the distribution of $W$ over the marginal
subpopulation. Also, note that by construction,\ $\int_{\mathcal{W}}\mathcal{%
\dot{P}}\left( w\right) dF(w)=1$, and thus $\mathcal{\dot{P}}\left( w\right) 
$ can be interpreted as the density of the distribution of $W$ for the
marginal subpopulation with respect to the distribution of $W$ for the
entire population.

It is now clear that the general unconditional policy effect is equal to the
average of the unconditional $\mathrm{MTE}$ over the marginal subpopulation.
Such an interpretation is still valid even if $\mathcal{\dot{P}}\left(
w\right) $ is not positive for all $w\in \mathcal{W}$. In this case, we only
need to view the distribution with density $\mathcal{\dot{P}}\left( w\right) 
$ (with respect to the distribution of $W$ for the entire population) as a
signed measure.

\subsection{Understanding the Unconditional MTE}

To gain a deeper understanding of the unconditional MTE, both in the general
case and the special quantile case, we will explore another perspective
here. Note that this subsection will only provide heuristics, as the formal
developments have already been covered in the previous subsections.

\cite{heckman_prte, Heckman2005} focus on the mean functional and consider
the \emph{policy-relevant treatment effect} defined as 
\begin{equation}
\text{\textrm{PRTE}}_{\delta }=\frac{E[Y_{\delta }]-E[Y]}{E[D_{\delta }]-E[D]%
}.  \label{prte}
\end{equation}%
Taking the limit $\delta \rightarrow 0$ yields the \emph{marginal
policy-relevant treatment effect} (\textrm{MPRTE}) of \cite{Carneiro2010}: $%
\mathrm{MPRTE}=\lim_{\delta \rightarrow 0}\text{\textrm{PRTE}}_{\delta }.$ 
\cite{Carneiro2010, Carneiro2011} show that $\mathrm{MPRTE}$ can be
represented in terms of the conventional marginal treatment effect defined
by $\mathrm{MTE}(u):=E\left[ Y(1)-Y(0)|U_{D}=u\right] $. These results are
applicable to the mean functional only.

For a general functional $\rho $ that is Hadamard differentiable, we have 
\begin{equation*}
\lim_{\delta \rightarrow 0}\frac{\rho \lbrack F_{Y_{\delta }}]-\rho \lbrack
F_{Y}]}{E[D_{\delta }]-E[D]}=\dot{\rho}_{F_{Y}}\left( \lim_{\delta
\rightarrow 0}\frac{F_{Y_{\delta }}-F_{Y}}{E[D_{\delta }]-E[D]}\right) ,
\end{equation*}%
where $\dot{\rho}_{F_{Y}}$ is the derivative of $\rho $ at $F_{Y}$. One way
to represent the limit on the right-hand side is to replace $E\left(
Y_{\delta }\right) $ and $E\left( Y\right) $ in (\ref{prte}) by $E[\mathds%
1\left\{ Y_{\delta }\leq y\right\} ]$ and $E[\mathds1\left\{ Y\leq y\right\}
],$ respectively. So, for a given $y\in \mathcal{Y}$, 
\begin{equation}
\frac{F_{Y_{\delta }}\left( y\right) -F_{Y}\left( y\right) }{E[D_{\delta
}]-E[D]}=\frac{E[\mathds1\left\{ Y_{\delta }\leq y\right\} -E[\mathds%
1\left\{ Y\leq y\right\} ]}{E[D_{\delta }]-E[D]}.
\end{equation}%
In the spirit of \cite{heckman_prte, Heckman2005}, we may then regard the
above as a policy-relevant treatment effect: it is the treatment effect of
the policy on the percentage of individuals whose value of $Y$ is less than $%
y.$ The effect is tied to a particular value $y$, and we obtain a continuum
of policy-relevant effects indexed by $y\in \mathcal{Y}$ if $y$ is allowed
to vary over $\mathcal{Y}$. The limit 
\begin{equation*}
\lim_{\delta \rightarrow 0}\frac{F_{Y_{\delta }}\left( y\right) -F_{Y}\left(
y\right) }{E[D_{\delta }]-E[D]}
\end{equation*}%
can then be regarded as a continuum of marginal policy-relevant treatment
effects indexed by $y\in \mathcal{Y}$. By the results of \cite{Carneiro2009}%
, for each $y\in \mathcal{Y}$, the marginal policy-relevant treatment effect
can be represented as a weighted integral of the following \emph{%
policy-relevant \textquotedblleft distributional\textquotedblright }\ MTE: 
\begin{eqnarray*}
\mathrm{MTE}_{d}(u,w;y) &=&E\left[ \mathds1\left\{ Y(1)\leq y\right\} -%
\mathds1\left\{ Y(0)\leq y\right\} |U_{D}=u,W=w\right] \\
&=&F_{Y\left( 1\right) |U_{D},W}\left( y|u,w\right) -F_{Y\left( 0\right)
|U_{D},W}\left( y|u,w\right) .
\end{eqnarray*}%
The composition $\dot{\rho}_{F_{Y}}\circ \mathrm{MTE}_{d}$, defined as $%
\int_{\mathcal{Y}}\psi (y,\rho ,F_{Y})\mathrm{MTE}_{d}(u,w;dy),$ is then 
\begin{eqnarray*}
\left( \dot{\rho}_{F_{Y}}\circ \mathrm{MTE}_{d}\right) \left( u,w\right)
&=&\int_{\mathcal{Y}}\psi (y,\rho ,F_{Y})\left[ F_{Y\left( 1\right)
|U_{D},W}\left( dy|u,w\right) -F_{Y\left( 0\right) |U_{D},W}\left(
dy|u,w\right) \right] \\
&=&\int_{\mathcal{Y}}\psi (y,\rho ,F_{Y})\left[ f_{Y\left( 1\right)
|U_{D},W}\left( y|u,w\right) -f_{Y\left( 0\right) |U_{D},W}\left(
y|u,w\right) \right] dy \\
&=&E\left[ \psi (Y\left( 1\right) ,\rho ,F_{Y})-\psi (Y\left( 0\right) ,\rho
,F_{Y})|U_{D}=u,W=w\right] .
\end{eqnarray*}%
This shows that $\dot{\rho}_{F_{Y}}\circ \mathrm{MTE}_{d}$ is exactly the
unconditional MTE defined in (\ref{MTE_rho}). The unconditional MTE is
therefore a composition of the underlying influence function with the \emph{%
policy-relevant distributional} MTE.

\section{UQR with a Threshold-crossing Model}

\label{section_uqr}

In this section, we study whether the UQR proposed by \cite{Firpo2009} can
provide a consistent estimator of UQE. When it is inconsistent, we
investigate the sources of asymptotic bias in the UQR estimator and show
that it is asymptotically biased, even when $D$ is exogenous. As a result,
the UQR may not be suitable for estimating unconditional policy effects when
the treatment variable follows a binary threshold-crossing model. In Section %
\ref{estimation}, we will present a consistent estimator of the
unconditional policy effect for such a model.

\subsection{UQR with a Binary Regressor}

We provide a quick review of the UQR. As before, let $y_{\tau }$ be the $%
\tau $-quantile of $Y,$ and let $y_{\tau ,\delta }$ be the $\tau $-quantile
of $Y_{\delta }$. That is, $\Pr [Y\leq y_{\tau }]=\Pr [Y_{\delta }\leq
y_{\tau ,\delta }]=\tau .$ By definition, we have%
\begin{equation*}
\Pr \left[ Y_{\delta }\leq y\right] =\Pr \left[ D_{\delta }=1\right] \Pr %
\left[ Y_{\delta }\leq y|D_{\delta }=1\right] +\Pr \left[ D_{\delta }=0%
\right] \Pr \left[ Y_{\delta }\leq y|D_{\delta }=0\right] .
\end{equation*}%
When $W$ is not present, Corollary 3 in the working paper \cite{Firpo2007}
makes the following assumption to achieve identification: 
\begin{equation}
\Pr \left[ Y_{\delta }\leq y|D_{\delta }=d\right] =\Pr \left[ Y\leq y|D=d%
\right] \text{ for }d=0\text{ and }1.  \label{key_assumption}
\end{equation}%
We refer to this assumption as \emph{distributional invariance}, and it
readily identifies the counterfactual distribution: 
\begin{equation*}
\Pr \left[ Y_{\delta }\leq y\right] =\Pr \left[ D_{\delta }=1\right] \Pr %
\left[ Y\leq y|D=1\right] +\Pr \left[ D_{\delta }=0\right] \Pr \left[ Y\leq
y|D=0\right] .
\end{equation*}%
%
%
%
%
%
%
%
%
%
%
%
%
%
%
%
%
%
%
%
%
%
%
%
%
%
%
%
%
%
%
%
%
%
%
%
%
%
%
%
%
%
%
%
%
%
%
%
%
%
%
%
%
%
%
%
%
%
%
%
%
%
%
%
%
%
%
%
%
%
%
%
%
%
%
%
%
%
%
%
%
%
%
%
%
%
%
%
%
%
%
%
%
%
%
%
%
%
%
%
%
%
%
%
%
%
%
%
%
Under some mild conditions, we can follow \cite{Firpo2007} to show that 
\begin{eqnarray}
\Pi _{\tau } &:=&\lim_{\delta \rightarrow 0}\frac{\rho _{\tau }[F_{Y_{\delta
}}]-\rho _{\tau }[F_{Y}]}{E[D_{\delta }]-E[D]}  \notag \\
&=&\frac{1}{f_{Y}(y_{\tau })}\left( \Pr [Y>y_{\tau }|D=1]-\Pr [Y>y_{\tau
}|D=0]\right) .  \label{pi_no_W}
\end{eqnarray}%
Hence, under the distributional invariance assumption\emph{, }the UQE can be
consistently estimated by regressing $1\left\{ Y\geq y_{\tau }\right\}
/f_{Y}(y_{\tau })$ on a constant and $D.$ Such a regression with no
additional regressor $W$ is a special case of more general unconditional
quantile regressions.

%
%
%
%
%
%
%
%
%
%
%
%
%
%
%
%
%
%
%
%
%
%
%
%
%
%
%
%
%
%
%
%
%
%
%
%
%
%
%
%
%
%
%
%
%
%
%
%
%
%
%
%
%
%
%
%
%
%
%
%
%
%
%
%
%
%
%
%

\subsection{Asymptotic Bias of the UQR Estimator}

The distributional invariance assumption given in equation %
\eqref{key_assumption} is crucial for achieving\ the identification result
in \eqref{pi_no_W}. It states that the conditional distribution of the
outcome variable given the treatment status remains the same across the two
policy regimes. If treatments are randomly assigned under both policy
regimes (e.g., $D_{\delta }=\mathds{1}\left\{ U_{D}\leq P_{\delta
}(W)\right\} $ and $\left( U_{D},W\right) $ is independent of $\left(
U_{0},U_{1}\right) ),$ then $\left( U_{0},U_{1}\right) $ is clearly
independent of $D_{\delta }$. In this case, both $\Pr \left[ Y_{\delta }\leq
y|D_{\delta }=d\right] $ and $\Pr \left[ Y\leq y|D=d\right] $ are equal to $%
\Pr \left[ Y\left( d\right) \leq y\right] ,$ and the distributional
invariance assumption is satisfied. However, when $D_{\delta }$ is allowed
to be correlated with $U=\left( U_{0},U_{1}\right) ^{\prime }$, the
distributional invariance assumption does not hold in general. For example,
when $d=1,$ 
\begin{equation*}
\Pr \left[ Y_{\delta }\leq y|D_{\delta }=1\right] =\Pr \left[ Y\left(
1\right) \leq y|D_{\delta }=1\right] =\Pr \left[ r_{1}\left( X,U_{1}\right)
\leq y|U_{D}\leq P_{\delta }\left( W\right) \right] ,
\end{equation*}%
and $\Pr [Y\leq y|D=1]=\Pr \left[ r_{1}\left( X,U_{1}\right) \leq
y|U_{D}\leq P\left( W\right) \right] .$ These two conditional probabilities
are different under the general dependence of $(W,U,U_{D}).$

To allow for the endogeneity of $D$, we have dropped the distributional
invariance assumption and assumed a threshold-crossing model as in (\ref%
{model1D}). The next corollary decomposes the unconditional quantile effect
given in Corollary \ref{Corollary_UQTE0} into two components. The
decomposition reveals that the UQR estimator of \cite{Firpo2009} is
asymptotically biased under a wide range of conditions, including when $D$
is exogenous.

\begin{corollary}
\label{Corollary_UQTE}Let Assumptions \ref{Assumption_primary}--\ref%
{Assumption_domination} hold. Assume further that $f_{Y}(y_{\tau })>0$. Then%
\footnote{%
We use \textquotedblleft A\textquotedblright\ to denote the \textbf{A}%
pparent component and use \textquotedblleft B\textquotedblright\ to denote
the \textbf{B}ias component.} 
\begin{equation*}
\Pi _{\tau }=A_{\tau }-B_{\tau },
\end{equation*}%
where 
\begin{eqnarray}
A_{\tau } &=&\frac{1}{f_{Y}(y_{\tau })}\int_{\mathcal{W}}E\left[ \mathds{1}%
\left\{ Y\leq y_{\tau }\right\} |D=0,W=w\right] dF_{W}\left( w\right)  \notag
\\
&-&\frac{1}{f_{Y}(y_{\tau })}\int_{\mathcal{W}}E\left[ \mathds{1}\left\{
Y\leq y_{\tau }\right\} |D=1,W=w\right] dF_{W}\left( w\right) ,
\label{app_eff}
\end{eqnarray}%
and $B_{\tau }=B_{1\tau }+B_{2\tau }$, for 
\begin{eqnarray*}
B_{1\tau } &=&\frac{1}{f_{Y}(y_{\tau })}\int_{\mathcal{W}}\left[
F_{Y|D,W}\left( y_{\tau }|1,w\right) -F_{Y|D,W}\left( y_{\tau }|0,w\right) %
\right] \mathcal{\dot{P}}\left( w\right) dF_{W}\left( w\right) \\
&-&\frac{1}{f_{Y}(y_{\tau })}\int_{\mathcal{W}}\left[ F_{Y|D,W}\left(
y_{\tau }|1,w\right) -F_{Y|D,W}\left( y_{\tau }|0,w\right) \right]
dF_{W}\left( w\right)
\end{eqnarray*}%
and 
\begin{eqnarray*}
B_{2\tau } &=&\frac{1}{f_{Y}(y_{\tau })}\int_{\mathcal{W}}\left[ F_{Y\left(
0\right) |D,W}\left( y_{\tau }|0,w\right) -F_{Y\left( 0\right)
|U_{D},W}\left( y_{\tau }|P\left( w\right) ,w\right) \right] \mathcal{\dot{P}%
}\left( w\right) dF_{W}\left( w\right) \\
&-&\frac{1}{f_{Y}(y_{\tau })}\int_{\mathcal{W}}\left[ F_{Y\left( 1\right)
|D,W}\left( y_{\tau }|1,w\right) -F_{Y\left( 1\right) |U_{D},W}\left(
y_{\tau }|P\left( w\right) ,w\right) \right] \mathcal{\dot{P}}\left(
w\right) dF_{W}\left( w\right) . \\
&&
\end{eqnarray*}
\end{corollary}

To facilitate understanding of Corollary \ref{Corollary_UQTE}, we define and
organize the average influence functions (AIF) in a table:%

\begin{equation*}
\begin{tabular}{c|c}
\hline\hline
& AIF for $Y\left( 1\right) $ $-$ AIF for $Y\left( 0\right) $ \\ \hline
$U_{D}$ & $E_{w}\left[ \psi _{\tau }\left( Y\left( 1\right) \right) -\psi
_{\tau }\left( Y\left( 0\right) \right) |U_{D}=P\left( w\right) \right] $ \\ 
$D$ & $E_{w}\left[ \psi _{\tau }\left( Y\left( 1\right) \right) |D=1\right]
-E_{w}\left[ \psi _{\tau }\left( Y\left( 0\right) \right) |D=0\right] $ \\ 
\hline
\end{tabular}%
\end{equation*}

\noindent where $\psi _{\tau }\left( \cdot \right) $ is short for $\psi
(\cdot ,\rho _{\tau },F_{Y})$, the influence function of the quantile
functional. In the above, $E_{w}\left[ \cdot \right] $ stands for the
conditional mean operator given $W=w.$ For example, $E_{w}\left[ \psi _{\tau
}\left( Y\left( 0\right) \right) |D=0\right] $ stands for $E\left[ \psi
_{\tau }\left( Y\left( 0\right) \right) |D=0,W=w\right] .$ Let 
\begin{eqnarray*}
\psi _{\Delta ,U_{D}}\left( w\right) := &&E_{w}\left[ \psi _{\tau }\left(
Y\left( 1\right) \right) -\psi _{\tau }\left( Y\left( 0\right) \right)
|U_{D}=P\left( w\right) \right] , \\
\psi _{\Delta ,D}\left( w\right) := &&E_{w}\left[ \psi _{\tau }\left(
Y\left( 1\right) \right) |D=1\right] -E_{w}\left[ \psi _{\tau }\left(
Y\left( 0\right) \right) |D=0\right] .
\end{eqnarray*}%
The unconditional quantile effect $\Pi _{\tau }$ is the average of the
difference $\psi _{\Delta ,U_{D}}\left( w\right) $ with respect to the
distribution of $W$ over the \emph{marginal subpopulation}. The average 
\emph{apparent} effect $A_{\tau }$ is the average of the difference $\psi
_{\Delta ,D}\left( w\right) $ with respect to the distribution of $W$ over
the whole population distribution. It is also equal to the limit of the UQR
estimator of \cite{Firpo2009}, where the endogeneity of the treatment
selection is ignored.\footnote{%
To see why this is the case, we note that, in its simplest form, the UQR
involves regressing the \textquotedblleft influence
function\textquotedblright\ $\hat{f}_{Y}(\hat{y}_{\tau })^{-1}\left( \tau -%
\mathds{1}\left\{ Y_{i}\leq \hat{y}_{\tau }\right\} \right) $ on $D_{i}$ and 
$W_{i}$ by OLS and using the estimated coefficient on $D_{i}$ as the
estimator of the unconditional quantile effect. Here, $\hat{f}_{Y}(y_{\tau
}) $ is a consistent estimator of $f_{Y}(y_{\tau })$ and $\hat{y}_{\tau }$
is a consistent estimator of $y_{\tau }.$ It is now easy to see that the UQR
estimator converges in probability to $A_{\tau }$ if the conditional
expectations in (\ref{app_eff}) are linear in $W.$} Note that if our model
contains no covariate $W,$ then 
\begin{eqnarray}
A_{\tau } &=&\frac{1}{f_{Y}(y_{\tau })}E\left[ \mathds{1}\left\{ Y\leq
y_{\tau }\right\} |D=0\right] -\frac{1}{f_{Y}(y_{\tau })}E\left[ \mathds{1}%
\left\{ Y\leq y_{\tau }\right\} |D=1\right]  \notag \\
&=&\frac{1}{f_{Y}(y_{\tau })}\left( \Pr [Y>y_{\tau }|D=1]-\Pr [Y>y_{\tau
}|D=0]\right) .
\end{eqnarray}%
This is identical to the unconditional quantile effect given in (\ref%
{pi_no_W}).

The discrepancy between $\Pi _{\tau }$ and $A_{\tau }$ gives rise to the
asymptotic \emph{bias} $B_{\tau }$ of the UQR estimator:%
\begin{eqnarray}
B_{\tau } &=&A_{\tau }-\Pi _{\tau }=E\left[ \psi _{\Delta ,D}\left( W\right) %
\right] -E\left[ \psi _{\Delta ,U_{D}}\left( W\right) \mathcal{\dot{P}}%
\left( W\right) \right]  \notag \\
&=&\underset{B_{1\tau }}{\underbrace{E\left\{ \psi _{\Delta ,D}\left(
W\right) \left[ 1-\mathcal{\dot{P}}\left( W\right) \right] \right\} }}+%
\underset{B_{2\tau }}{\underbrace{E\left\{ \left[ \psi _{\Delta ,D}\left(
W\right) -\psi _{\Delta ,U_{D}}\left( W\right) \right] \mathcal{\dot{P}}%
\left( W\right) \right\} }}.  \label{Bias_decompose}
\end{eqnarray}%
It is easy to see that $B_{1\tau }$ and $B_{2\tau }$ given above are
identical to those given in Corollary \ref{Corollary_UQTE}.

The decomposition in Equation (\ref{Bias_decompose}) traces the asymptotic
bias back to two sources. The first one, $B_{1\tau }$, captures\ the
heterogeneity of the averaged apparent effects averaged over two different
subpopulations. For every $w,$ $\psi _{\Delta ,D}\left( w\right) $ is the
average effect of $D$ on $\left[ \tau -\mathds{1}\left\{ Y\leq y_{\tau
}\right\} \right] /f_{Y}\left( y_{\tau }\right) $ for the individuals with $%
W=w.$ These effects are averaged over two different distributions of $W$:
the distribution of $W$ for the marginal subpopulation (i.e., $\mathcal{\dot{%
P}}\left( w\right) f_{W}\left( w\right) )$ and the distribution of $W$ for
the whole population (i.e., $f_{W}\left( w\right) )$. $B_{1\tau }$ is equal
to the difference between these two average effects.\emph{\ }If the effect $%
\psi _{\Delta ,D}\left( w\right) $ does not depend on $w$, then $B_{1\tau
}=0 $. If $\mathcal{\dot{P}}\left( \cdot \right) $ is a constant function
that always equals 1, then the distribution of $W$ over the whole population
is the same as that over the marginal subpopulation, and hence $B_{1\tau }=0$
as well. For $B_{1\tau }\neq 0,$ it is necessary that there is an effect
heterogeneity (i.e., $\psi _{\Delta ,D}\left( w\right) $ depends on $w$) and
a distributional heterogeneity (i.e., $\mathcal{\dot{P}}\left( \cdot \right) 
$ does not always equal $1$, and as a result, the distribution of $W$ over
the marginal subpopulation is different from that over the whole
population). To highlight the necessary conditions for a nonzero $B_{1\tau
}, $ we refer to $B_{1\tau }$ as the \emph{marginal heterogeneity bias. }

The second bias component, $B_{2\tau },$ embodies the second source of the
bias and has a difference-in-differences interpretation.\ Each of $\psi
_{\Delta ,D}\left( \cdot \right) $ and $\psi _{\Delta ,U_{D}}\left( \cdot
\right) $ is the difference in the average influence functions associated
with the counterfactual outcomes $Y\left( 1\right) $ and $Y\left( 0\right) .$
However, $\psi _{\Delta ,D}\left( \cdot \right) $ is the difference over the
two subpopulations who actually choose $D=1$ and $D=0$, while $\psi _{\Delta
,U_{D}}\left( \cdot \right) $ is the difference over the marginal
subpopulation. So $\psi _{\Delta ,D}\left( \cdot \right) -\psi _{\Delta
,U_{D}}\left( \cdot \right) $ is a difference in differences. $B_{2\tau }$
is simply the average of this difference in differences with respect to the
distribution of $W$ over the marginal subpopulation. This term arises
because the change in the distributions of $Y$ for those with $D=1$ and
those with $D=0$ is different from that for those whose $U_{D}$ is just
above $P\left( w\right) $ and those whose $U_{D}$ is just below $P\left(
w\right) $. Thus, we can label $B_{2\tau }$ as a \emph{marginal selection
bias}.

If $\psi _{\Delta ,D}\left( w\right) =\psi _{\Delta ,U_{D}}\left( w\right) $
for almost all $w\in \mathcal{W},$ then $B_{2\tau }=0.$ The condition $\psi
_{\Delta ,D}\left( w\right) =\psi _{\Delta ,U_{D}}\left( w\right) $ is the
same as 
\begin{equation*}
E_{w}\left[ \psi _{\tau }\left( Y\left( 1\right) \right) -\psi _{\tau
}\left( Y\left( 0\right) \right) |U_{D}=P\left( w\right) \right] =E_{w}\left[
\psi _{\tau }\left( Y\left( 1\right) \right) |D=1\right] -E_{w}\left[ \psi
_{\tau }\left( Y\left( 0\right) \right) |D=0\right] .
\end{equation*}%
Equivalently, 
\begin{eqnarray*}
&&E_{w}\left[ \psi _{\tau }\left( Y\left( 1\right) \right) |U_{D}=P\left(
w\right) \right] -E_{w}\left[ \psi _{\tau }\left( Y\left( 1\right) \right)
|D=1\right] \\
&=&E_{w}\left[ \psi _{\tau }\left( Y\left( 0\right) \right) |U_{D}=P\left(
w\right) \right] -E_{w}\left[ \psi _{\tau }\left( Y\left( 0\right) \right)
|D=0\right] .
\end{eqnarray*}%
The condition resembles the parallel-paths assumption or the constant-bias
assumption in a difference-in-differences analysis. If $U_{D}$ is
independent of $\left( U_{0},U_{1}\right) $ given $W,$ then this condition
holds, and $B_{2\tau }=0.$

In general, when $U_{D}$ is not independent of $\left( U_{0},U_{1}\right) $
given $W$, and $W$ enters the selection equation, we have $B_{1\tau }\neq 0$
and $B_{2\tau }\neq 0,$ hence $\Pi _{\tau }\neq A_{\tau }.$ If $\mathcal{%
\dot{P}}\left( w\right) $ is not identified, then $B_{1\tau }$ is not
identified. In general, $B_{2\tau }$ is not identified without additional
assumptions. Therefore, in the absence of additional assumptions, the
asymptotic bias can not be eliminated, and $\Pi _{\tau }$ is not identified.

It is not surprising that in the presence of endogeneity, the UQR estimator
is asymptotically biased. The virtue of Corollary \ref{Corollary_UQTE} is
that it provides a closed-form characterization and clear interpretations of
the asymptotic bias. To the best of our knowledge, this bias formula is new
in the literature. If point identification can not be achieved, then the
bias formula can be used in a bound analysis or sensitivity analysis. From a
broad perspective, the asymptotic bias $B_{\tau }$ is the unconditional
quantile counterpart of the endogenous bias of the OLS estimator in a linear
regression framework.\footnote{%
The bias decomposition is not unique. Corollary \ref{Corollary_UQTE} gives
only one possibility. We can also write 
\begin{equation}
B_{\tau }=\text{ }\underset{\tilde{B}_{1\tau }}{\underbrace{E\left\{ \psi
_{\Delta ,U_{D}}\left( W\right) \left[ 1-\dot{P}\left( W\right) \right]
\right\} }}+\underset{\tilde{B}_{2\tau }}{\underbrace{E\left[ \psi _{\Delta
,D}\left( W\right) -\psi _{\Delta ,U_{D}}\left( W\right) \right] }}.  \notag
\end{equation}%
The interpretations of $\tilde{B}_{1\tau }$ and $\tilde{B}_{2\tau }$ are
similar to those of $B_{1\tau }$ and $B_{2\tau }$ with obvious and minor
modifications. In this case, it may be more revealing to call $\tilde{B}%
_{1\tau }$ the marginal heterogeneity bias, because a necessary condition
for a nonzero $\tilde{B}_{1\tau }$ is that there is an effect heterogeneity
among the marginal subpopulation.}

\section{Unconditional Quantile Effect under Instrumental Intervention\label%
{sect:IV}}

In this section, we consider two types of interventions: one on a valid
instrumental variable and the other on an invalid instrumental variable. We
show how we may recover the effect of the latter by using a valid instrument.

\subsection{Instrumental Intervention with a valid instrument}

To introduce the instrumental intervention, we partition $Z$ into two parts
and write $Z=\left( Z_{1},Z_{-1}\right) $\textbf{\ }where $Z_{1}$\ is a
univariate continuous policy variable and $Z_{-1}$ consists of other
variables. If there is only one variable in $Z,$ then $Z=Z_{1},$ and $Z_{-1}$
is not present. We consider the intervention 
\begin{equation}
\mathcal{G}\left( W,\delta \right) =\left( Z_{1}+g\left( W\right) s\left(
\delta \right) ,Z_{2},\cdots ,Z_{d_{Z}}\right)  \label{loc_inst}
\end{equation}%
where $g\left( \cdot \right) $ is a measurable function and $s\left( \delta
\right) $ is a smooth function satisfying $s\left( 0\right) =0$. That is, we
intervene to change the first component of $Z.$ Note that while $s(\delta )$
is the same for all individuals, $g\left( W\right) $ depends on the value of 
$W,$ and hence it is individual-specific. Thus, we allow the intervention to
be heterogeneous. In empirical applications, $Z$ may consist of a few
variables, and $Z_{1}$ is the \emph{continuous} target variable that we
attempt to change. There are no continuity requirements for the other
elements of $Z.$

To present our identification assumptions, we write $W=(W_{1},W_{-1})$ where 
$W_{1}$ is the first element of $W$ and $W_{-1}$ consists of other elements
of $W.$ Since $W=(Z,X),$ the first element of $W$ is also the first element
of $Z$, and hence $W_{1}=Z_{1}$. By definition, we have $W_{-1}=\left(
Z_{-1},X\right) .$ We maintain the following assumptions, which are similar
to the corresponding assumptions in \cite{Heckman1999, Heckman2001,
Heckman2005}.

\begin{assumption}
\textbf{Relevance and Exogeneity} \label{Assumption_heckman}

\begin{enumerate}[(a)]%

\item \label{relevance}Conditional on $W_{-1}=\left( Z_{-1},X\right) ,$ $\mu
(Z,X)$ is a non-degenerate random variable.

\item \label{exogeneity}Conditional on $W_{-1}=\left( Z_{-1},X\right) ,$ $%
Z_{1}$ is independent of $(U_{0},U_{1},V).$

\end{enumerate}%
\end{assumption}

Assumption \ref{Assumption_heckman}(\ref{relevance}) is a relevance
assumption: for any given level of $W_{-1}$, $Z_{1}$ can induce some
variation in $D$. Assumption \ref{Assumption_heckman}(\ref{exogeneity}) is a
conditional exogeneity assumption: for any given level of $W_{-1}$, $Z_{1}$
is independent of the unobservables. These two assumptions are essentially
the conditions for $Z_{1}$ to be a valid instrumental variable, hence we
will refer to $Z_{1}$ as the instrumental variable, and the intervention in %
\eqref{loc_inst} as the instrumental intervention.\footnote{%
If all variables in $Z$ satisfy the conditional exogeneity assumption, we
may replace Assumption \ref{Assumption_heckman} by the following: \textit{%
(a) Conditional on }$X,$ $\mu \left( Z,X\right) $\textit{\ is a
non-degenerate random variable,} and \textit{(b) Conditional on }$X,$ $Z$%
\textit{\ is independent of }$(U_{0},U_{1},V)$\textit{.} With minor
modifications, our results will remain valid. More specifically, in the
statement of each result, we only need to replace \textquotedblleft
conditioning on $W_{-1}$\textquotedblright\ by \textquotedblleft
conditioning on $X$\textquotedblright . The working paper \cite{sun2021} is
based on this alternative assumption. Here we will work with Assumption \ref%
{Assumption_heckman}, which appears to be more plausible.}

Let $w:=\left( w_{1},w_{-1}\right) =(z_{1},w_{-1}).$ Under Assumption \ref%
{Assumption_heckman}(\ref{exogeneity}), the unconditional MTE for the
quantile functional $\rho _{\tau }$ becomes 
\begin{eqnarray}
\mathrm{MTE}_{\tau }\left( u,w\right) &=&\frac{1}{f_{Y}\left( y_{\tau
}\right) }E\left[ \mathds{1}\left\{ Y(0)\leq y_{\tau }\right\} -\mathds{1}%
\left\{ Y(1)\leq y_{\tau }\right\} \mid U_{D}=u,W=w\right]  \notag \\
&=&\frac{1}{f_{Y}\left( y_{\tau }\right) }E\left[ \mathds{1}\left\{ Y(0)\leq
y_{\tau }\right\} -\mathds{1}\left\{ Y(1)\leq y_{\tau }\right\} \mid
U_{D}=u,W_{-1}=w_{-1}\right]  \notag \\
&:=&\widetilde{\mathrm{MTE}}_{\tau }\left( u,w_{-1}\right) ,
\label{MTE_tilde}
\end{eqnarray}%
where in the second line above, conditioning on $Z_{1}$ is not necessary and
has been dropped.

Next, we present a lemma that characterizes the weighting function for the
instrumental intervention in \eqref{loc_inst}.

\begin{lemma}
\label{Lemma_SZ}Assume that (i) for almost all $w\in \mathcal{W}$, the
conditional distribution of$\ V$ conditional on $W=w$ is absolutely
continuous conditional density $f_{V|W}\left( v|w\right) $;\ (ii) $\mu
\left( w\right) $\ is differentiable in $z_{1}$\ for almost all $w\in 
\mathcal{W}$ with derivative $\mu _{z_{1}}^{\prime }\left( \cdot \right) $
such that $E\left[ f_{V|W}\left( \mu (W)|W\right) \mu _{z_{1}}^{\prime
}\left( W\right) g\left( W\right) \right] $ is well defined and is not equal
to zero; (iii) $s\left( \delta \right) $\ is a differentiable function in a
neighborhood of zero and $\left. \partial s\left( \delta \right) /\partial
\delta \right\vert _{\delta =0}\neq 0$. Then%
\begin{equation}
\mathcal{\dot{P}}\left( w\right) =\frac{f_{V|W}\left( \mu (w)|w\right) \mu
_{z_{1}}^{\prime }\left( w\right) g(w)}{E\left[ f_{V|W}\left( \mu \left(
W\right) |W\right) \mu _{z_{1}}^{\prime }\left( W\right) g(W)\right] }.
\label{P_dot_w}
\end{equation}
\end{lemma}

The lemma shows that the weighting function $\mathcal{\dot{P}}\left(
w\right) $ does not depend on the function form of $s\left( \cdot \right) .$
Hence, the unconditional policy effect does not depend on $s\left( \cdot
\right) .$

\begin{corollary}
\label{Corollary_UQTE_w} Let Assumptions \ref{Assumption_primary}--\ref%
{Assumption_domination} and \ref{Assumption_heckman}, and the assumptions of
Lemma \ref{Lemma_SZ} hold. Assume further that $f_{Y}(y_{\tau })>0$. Then,
the unconditional quantile effect of the instrumental intervention given in %
\eqref{loc_inst} is 
\begin{equation}
\Pi _{\tau }=\int_{\mathcal{W}}\widetilde{\mathrm{MTE}}_{\tau }\left(
P(w),w_{-1}\right) \mathcal{\dot{P}}\left( w\right) dF_{W}\left( w\right). 
\notag
\end{equation}%
%
%
%
%
%
%
%
%
%
\end{corollary}

The proposition below establishes the identifiability of $\widetilde{\mathrm{%
MTE}}_{\tau }$ and of the weighting function $\mathcal{\dot{P}}\left(
w\right) $ given in Corollary \ref{Corollary_UQTE_w}.

\begin{proposition}
\label{Prop_MTE_identification}Let Assumptions \ref{Assumption_regularity}(%
\ref{regularity_x_abs}), \ref{Assumption_regularity}(\ref{f_y_u_x}), and \ref%
{Assumption_heckman}(\ref{exogeneity}), and the assumptions in Lemma \ref%
{Lemma_SZ} hold. Then, for every $u=P(w)$ for some $w=(w_{1},w_{-1})\in 
\mathcal{W}$, we have\footnote{%
For a more general functional $\rho ,$ we can show that 
\begin{equation*}
\widetilde{\mathrm{MTE}}_{\rho }(u,w_{-1})=\frac{\partial E\left[ \psi
(Y,\rho ,F_{Y})|P(W)=u,W_{-1}=w_{-1}\right] }{\partial u}
\end{equation*}%
for every $u$ equal to $P(w)$ for some $w=(w_{1},w_{-1})\in \mathcal{W}$.} 
\begin{equation*}
\widetilde{\mathrm{MTE}}_{\tau }\left( u,w_{-1}\right) =-\frac{1}{%
f_{Y}\left( y_{\tau }\right) }\frac{\partial E\left[ \mathds{1}\left\{ Y\leq
y_{\tau }\right\} |P(W)=u,W_{-1}=w_{-1}\right] }{\partial u},
\end{equation*}
and 
\begin{equation}
\mathcal{\dot{P}}\left( w\right) =\frac{\frac{\partial P(w)}{\partial z_{1}}%
g(w)}{E\left[ \frac{\partial P(W)}{\partial z_{1}}g\left( W\right) \right] }
\label{id_prop_weight}
\end{equation}%
where $\frac{\partial P(W)}{\partial z_{1}}$ is short for $\left. \frac{%
\partial P(w)}{\partial z_{1}}\right\vert _{w=W}.$
\end{proposition}

Using Proposition \ref{Prop_MTE_identification}, and the fact that $g(w)$ is
known, we can represent $\Pi _{\tau }$ as%
\begin{equation}
\Pi _{\tau }=-\frac{1}{f_{Y}(y_{\tau })}\int_{\mathcal{W}}\frac{\partial E%
\left[ \mathds{1}\left\{ Y\leq y_{\tau }\right\} |P(W)=P(w),W_{-1}=w_{-1}%
\right] }{\partial P(w)}\frac{\frac{\partial P(w)}{\partial z_{1}}g(w)}{E%
\left[ \frac{\partial P(W)}{\partial z_{1}}g\left( W\right) \right] }%
dF_{W}\left( w\right) .  \label{UQE_representation}
\end{equation}%
All objects in the above are point identified, hence $\Pi _{\tau }$ is point
identified.\footnote{%
We note that Assumption \ref{Assumption_heckman}(\ref{exogeneity}) plays a
key role in identifying $\mathcal{\dot{P}}\left( w\right) .$ Without the
assumption that $V$ is independent of $Z_{1}$ conditional on $\left(
Z_{-1},X\right) ,$ we can have only that%
\begin{equation*}
\frac{\partial P(w)}{\partial z_{1}}=f_{V|Z,X}(\mu (w)|w)\mu
_{z_{1}}^{\prime }\left( w\right) +\frac{\partial F_{V|Z,X}(\mu (w)|\tilde{z}%
,x)}{\partial \tilde{z}_{1}}\bigg |_{\tilde{z}=z}.
\end{equation*}%
The presence of the second term in the above equation invalidates the
identification result in (\ref{id_prop_weight}).}

\subsection{Instrumental Intervention with an invalid instrument}

While identifying the unconditional effect of an instrumental intervention
is of interest in its own right, the identification result can be further
leveraged to identify the unconditional policy effect of another
intervention. Theorem \ref{Theorem general} has shown that the unconditional
effects of two interventions will be the same if their weighting functions
coincide. Consider a counterfactual policy with a target weighting function $%
\mathcal{\dot{P}}^{\circ }\left( \cdot \right) .$ By strategically choosing
the instrument function $g\left( \cdot \right) $, we can ensure that the
weighting function under the policy intervention in \eqref{loc_inst} is the
same as $\mathcal{\dot{P}}^{\circ }\left( \cdot \right) .$ In other words,
with appropriate choices of $g\left( \cdot \right) ,$ the unconditional
effect of intervening $Z_{1}$ is the same as the unconditional effect of
another counterfactual policy. If the former is identified, then the latter
is also identified.

As an example, suppose we intervene on the second element $Z_{2}$ of $Z$
with 
\begin{equation}
\mathcal{G}^{\circ }\left( W,\delta \right) =\left( Z_{1},Z_{2}+g^{\circ
}\left( W\right) s^{\circ }\left( \delta \right) ,\cdots ,Z_{d_{Z}}\right)
\label{loc_intervention}
\end{equation}%
for some $g^{\circ }\left( \cdot \right) $ and $s^{\circ }\left( \delta
\right) .$ Under conditions similar to those in Lemma \ref{Lemma_SZ}, the
weighting function for this intervention is 
\begin{equation*}
\mathcal{\dot{P}}^{\circ }\left( w\right) =\frac{f_{V|W}\left( \mu
(w)|w\right) \mu _{z_{2}}^{\prime }\left( w\right) g^{\circ }(w)}{E\left[
f_{V|W}\left( \mu \left( W\right) |W\right) \mu _{z_{2}}^{\prime }\left(
W\right) g^{\circ }(W)\right] }.
\end{equation*}%
Letting $\mathcal{\dot{P}}^{\circ }\left( w\right) =\mathcal{\dot{P}}\left(
w\right) $ for $\mathcal{\dot{P}}\left( w\right) $ given in Lemma \ref%
{Lemma_SZ} and solving for $g(\cdot )$ yields%
\begin{equation*}
g(w)=\frac{\mu _{z_{2}}^{\prime }\left( w\right) }{\mu _{z_{1}}^{\prime
}\left( w\right) }g^{\circ }(w).
\end{equation*}%
So, the unconditional effect of the intervention given in (\ref%
{loc_intervention}) is the same as that of the intervention given in %
\eqref{loc_inst} when $g(w)$ is chosen appropriately. It is important to
point out that $Z_{2}$ may not be a valid instrument, and its unconditional
effect is identified via \textquotedblleft intervention
matching.\textquotedblright

In general, identifying an unconditional effect of an invalid instrument via
intervention matching\ is feasible only if we have prior knowledge of the
ratio $\mu _{z_{2}}^{\prime }\left( w\right) /\mu _{z_{1}}^{\prime }\left(
w\right) .$ This information may be available from economic theory. When
such knowledge is unavailable, an alternative approach can be employed.

The alternative approach hinges on the assumption that $\left(
Z_{1},Z_{2}\right) $ is independent of $V$ conditional on the remaining
elements in $W$, denoted as $W_{-(1,2)}$. In this case, for $%
W=(Z_{1},Z_{2},W_{-(1,2)})$ and $w=(z_{1},z_{2},w_{-\left( 1,2\right) }),$
we have 
\begin{eqnarray*}
P(w) &=&\Pr \left[ V\leq \mu \left( z_{1},z_{2},w_{-\left( 12\right)
}\right) |W=w\right] \\
&=&\Pr \left[ V\leq \mu \left( z_{1},z_{2},w_{-\left( 12\right) }\right)
|W_{-(1,2)}=w_{-\left( 1,2\right) }\right] \\
&=&F_{V|W_{-\left( 1,2\right) }}\left( \mu \left( z_{1},z_{2},w_{-\left(
12\right) }\right) |w_{-\left( 1,2\right) }\right) ,
\end{eqnarray*}%
and 
\begin{equation*}
\left( \frac{\partial P(w)}{\partial z_{2}}\right) /\left( \frac{\partial
P(w)}{\partial z_{1}}\right) =\frac{f_{V|W_{-\left( 1,2\right) }}\left( \mu
\left( z_{1},z_{2},w_{-12}\right) |w_{-\left( 1,2\right) }\right) \mu
_{z_{2}}^{\prime }\left( w\right) }{f_{V|W_{-\left( 1,2\right) }}\left( \mu
\left( z_{1},z_{2},w_{-12}\right) |w_{-\left( 1,2\right) }\right) \mu
_{z_{1}}^{\prime }\left( w\right) }=\frac{\mu _{z_{2}}^{\prime }\left(
w\right) }{\mu _{z_{1}}^{\prime }\left( w\right) }.
\end{equation*}%
Thus, the ratio $\mu _{z_{2}}^{\prime }\left( w\right) /\mu _{z_{1}}^{\prime
}\left( w\right) $ can be identified via the ratio of two partial
derivatives of the propensity score function.

It is important to note that the conditional independence of $Z_{2}$ from $V$
given $W_{-(1,2)}$ does not rule out the possibility that $Z_{2}$ may still
be dependent on $(U_{0},U_{1})$, and consequently, $Z_{2}$ could still be an
invalid instrument. This shows that intervention matching may be used to
identify the effect of intervening on an invalid instrument.

\section{Unconditional Instrumental Quantile Estimation}

\label{estimation}

This section is devoted to the estimation and inference of the UQE under the
instrumental intervention in (\ref{loc_inst}). We assume that the propensity
score function is parametric, and we leave the case with a nonparametric
propensity score to Section \ref{non_parametric_ps_appendix} of the
supplementary appendix. In order to simplify the notation, we set $g(\cdot
)\equiv 1$ for the remainder of this paper.\footnote{%
The presence of $g\left( \cdot \right) $ amounts to a change of measure:
from a measure with density $f_{W}(w)$ to a measure with density $%
g(w)f_{W}\left( w\right) $. When $g\left( \cdot \right) $ is not equal to a
constant function, we only need to change the population expectation
operator $E\left[ h(W)\right] $ that involves the distribution of $W$ into $E%
\left[ h(W)g(W)\right] $ and the empirical average operator $\mathbb{P}_{n}%
\left[ h\left( W\right) \right] $ into $\mathbb{P}_{n}\left[ h(W)g\left(
W\right) \right] .$ All of our results will remain valid.}

Letting 
\begin{equation*}
m_{0}(y_{\tau },P(w),w_{-1}):=E\left[ \mathds{1}\left\{ Y\leq y_{\tau
}\right\} |P(W)=P(w),W_{-1}=w_{-1}\right]
\end{equation*}%
and using (\ref{UQE_representation}), we have 
\begin{equation}
\Pi _{\tau }=-\frac{1}{f_{Y}(y_{\tau })}E\left[ \frac{\partial P(W)}{%
\partial Z_{1}}\right] ^{-1}E\left[ \frac{\partial m_{0}(y_{\tau
},P(W),W_{-1})}{\partial Z_{1}}\right] .  \label{equ_pi_tau}
\end{equation}%
$\Pi _{\tau }$ consists of two average derivatives and a density evaluated
at a point, some of which depend on the unconditional $\tau $-quantile $%
y_{\tau }$. Altogether $\Pi _{\tau }$ depends on four unknown quantities.

The method of unconditional instrumental quantile estimation involves first
estimating the four quantities separately and then plugging these estimates
into $\Pi _{\tau }$ to obtain the estimator $\hat{\Pi}_{\tau }.$ See (\ref%
{UNIQUE_Expression}) in Subsection \ref{parap_uqr} for the formula of $\hat{%
\Pi}_{\tau }.$

We consider estimating the four quantities in the next few subsections. For
a given sample $\left\{ O_{i}=(Y_{i},Z_{i},X_{i},D_{i})\right\} _{i=1}^{n}$,
we will use $\mathbb{P}_{n}$ to denote the empirical measure. The
expectation of a function $\chi \left( O\right) $ with respect to $\mathbb{P}%
_{n}$ is then $\mathbb{P}_{n}\chi =n^{-1}\sum_{i=1}^{n}\chi (O_{i})$.
Similarly, we use $\mathbb{P}$ to denote the population measure, and so $%
\mathbb{P}\chi \left( O\right) =E\chi \left( O\right) .$

\subsection{Estimating the Quantile and Density}

\label{two_step_den_section}

For a given $\tau $, we estimate $y_{\tau }$ using the (generalized) inverse
of the empirical distribution function of $Y$: $\hat{y}_{\tau }=\inf \left\{
y:\mathbb{F}_{n}(y)\geq \tau \right\} $, where 
\begin{equation*}
\mathbb{F}_{n}(y):=\frac{1}{n}\sum_{i=1}^{n}\mathds{1}\left\{ Y_{i}\leq
y\right\} .
\end{equation*}%
By Lemma \ref{quantile_an} in the appendix, we can write $\hat{y}_{\tau
}-y_{\tau }=\mathbb{P}_{n}\psi _{Q}(Y,y_{\tau })+o_{p}(n^{-1/2})$ 
where 
\begin{equation*}
\psi _{Q}(Y,y_{\tau }):=\frac{\tau -\mathds{1}\left\{ Y\leq y_{\tau
}\right\} }{f_{Y}(y_{\tau })}.
\end{equation*}%
Here the subscript \textquotedblleft $Q$\textquotedblright\ on $\psi _{Q}$
signifies that it is the influence function for a \textbf{Q}uantile
functional.

We use a kernel density estimator to estimate $f_{Y}(y)$. We maintain the
following assumptions on the kernel function and the bandwidth.

\begin{assumption}
\textbf{Kernel Assumption} \label{Assumption_Kernel}

\item The kernel function $K(\cdot )$ satisfies (i) $\int_{-\infty }^{\infty
}K(u)du=1$, (ii) $\int_{-\infty }^{\infty }u^{2}K(u)du<\infty $, and (iii) $%
K(u)=K(-u)$, and it is twice differentiable with Lipschitz continuous
second-order derivative $K^{\prime \prime }\left( u\right) $ satisfying (i) $%
\int_{-\infty }^{\infty }K^{\prime \prime }(u)udu<\infty $ and $\left(
ii\right) $ there exist positive constants $C_{1}$ and $C_{2}$ such that $%
\left\vert K^{\prime \prime }\left( u_{1}\right) -K^{\prime \prime }\left(
u_{2}\right) \right\vert \leq C_{2}\left\vert u_{1}-u_{2}\right\vert ^{2}$
for $\left\vert u_{1}-u_{2}\right\vert \geq C_{1}.$
\end{assumption}

\begin{assumption}
\textbf{Rate Assumption }\label{Assumption_rate}: $n\uparrow \infty $ and $%
h\downarrow 0$ such that $nh^{3}\uparrow \infty $ but $nh^{5}=O(1)$.
\end{assumption}

The non-standard condition $nh^{3}\uparrow \infty $ is due to the estimation
of $y_{\tau }$. Since we need to expand $\hat{f}_{Y}(\hat{y}_{\tau })-\hat{f}%
_{Y}(y_{\tau })$, which involves the derivative of $\hat{f}_{Y}(y),$ we have
to impose a slower rate of decay for $h$ to control the remainder. The
details can be found in the proof of Lemma \ref{two_step_density}. We note,
however, that $nh^{3}\uparrow \infty $ implies the usual rate condition $%
nh\uparrow \infty $.

The estimator of $f_{Y}(y)$ is then given by 
\begin{equation*}
\hat{f}_{Y}(y)=\frac{1}{n}\sum_{i=1}^{n}K_{h}\left( Y_{i}-y\right) ,\text{ }
\end{equation*}%
where $K_{h}\left( u\right) :=K(u/h)/h.$

By Lemma \ref{two_step_density} in the appendix, we can isolate the
estimation errors from estimating the density and quantile and write: 
\begin{align}
\hat{f}_{Y}(\hat{y}_{\tau })-{f}_{Y}(y_{\tau })& =\hat{f}_{Y}({y}_{\tau })-{f%
}_{Y}(y_{\tau })+f_{Y}(\hat{y}_{\tau })-{f}_{Y}(y_{\tau
})+o_{p}(n^{-1/2}h^{-1/2})  \notag \\
& =\mathbb{P}_{n}\psi _{f_{Y}}\left( Y,y_{\tau }\right) +B_{f_{Y}}(y_{\tau
})+f_{Y}^{\prime }(y_{\tau })\mathbb{P}_{n}\psi _{Q}(Y,y_{\tau
})+o_{p}(n^{-1/2}h^{-1/2}).  \label{decom_f_y}
\end{align}%
Here, $\psi _{f_{Y}}\left( Y,y_{\tau }\right) :=K_{h}\left( Y-y_{\tau
}\right) -E\left[ K_{h}\left( Y-y_{\tau }\right) \right] $ and $%
B_{f_{Y}}(y_{\tau })$ is the bias term, which is $O(h^{2})$. The first two
terms on the right-hand side of (\ref{decom_f_y}) represent the dominating
terms in the error from estimating $f_{Y}$. The third term reflects the
error from estimating $y_{\tau }$.

\subsection{Estimating the Average Derivatives}

\label{param_ps_section}

To estimate the two average derivatives, we make a parametric assumption on
the propensity score, leaving the nonparametric specification to Section \ref%
{non_parametric_ps_appendix} of the supplementary appendix. 

\begin{assumption}
\label{Assumption_parametric_ps} The propensity score $P(Z,X,\alpha _{0})$
is known up to a finite-dimensional vector $\alpha _{0}\in \mathbb{R}%
^{d_{\alpha }}$.
\end{assumption}

Under Assumption \ref{Assumption_parametric_ps}, the UQE $\Pi _{\tau }$ can
be written as 
\begin{equation}
\Pi _{\tau }=-\frac{1}{f_{Y}(y_{\tau })}\cdot \frac{1}{T_{1}}\cdot T_{2},
\label{param_interest_3}
\end{equation}%
where 
\begin{equation*}
T_{1}=E\left[ \frac{\partial P(Z,X,\alpha _{0})}{\partial Z_{1}}\right] 
\text{ and }T_{2}=E\left[ \frac{\partial m_{0}(y_{\tau },P(Z,X,\alpha
_{0}),W_{-1})}{\partial Z_{1}}\right] \text{.}
\end{equation*}

First, we estimate $T_{1},$ which is the mean of the derivative of the
propensity score, by 
\begin{equation*}
T_{1n}(\hat{\alpha}):=\frac{1}{n}\sum_{i=1}^{n}\frac{\partial P(z,x,\hat{%
\alpha})}{\partial z_{1}}\bigg |_{\left( z,x\right) =\left(
Z_{i},X_{i}\right) }
\end{equation*}%
where $\hat{\alpha}$ is an estimator of $\alpha _{0}$ satisfying $\hat{\alpha%
}-\alpha _{0}=\mathbb{P}_{n}\psi _{\alpha _{0}}\left( D,W\right)
+o_{p}(n^{-1/2})$ for some measurable function $\psi _{\alpha _{0}}\left(
\cdot ,\cdot \right) .$ To save space, we slightly abuse notation and write 
\begin{equation*}
\frac{\partial P(Z,X,\alpha )}{\partial z_{1}}=\frac{\partial P(z,x,\alpha )%
}{\partial z_{1}}\bigg |_{\left( z,x\right) =\left( Z,X\right) }.
\end{equation*}%
%
%
%
%
%
%
%
%
%
%
%
%
%
%
%
%
%
%
%
%
%
We adopt this convention in the rest of the paper. Under Lemma \ref%
{ps_estimation_param} in the appendix, we have 
\begin{eqnarray}
T_{1n}(\hat{\alpha})-T_{1} &=&\mathbb{P}_{n}\psi _{\partial P}\left(
W\right) +\left\{ E\left[ \frac{\partial ^{2}P(Z,X,\alpha _{0})}{\partial
z_{1}\partial \alpha }\right] \right\} ^{\prime }\mathbb{P}_{n}\psi _{\alpha
_{0}}\left( D,W\right) +o_{p}(n^{-1/2}),  \notag \\
&&  \label{decom_ps_non}
\end{eqnarray}%
where 
\begin{equation*}
\psi _{{\partial P}}\left( W\right) :=\frac{\partial P(Z,X,\alpha _{0})}{%
\partial z_{1}}-E\left[ \frac{\partial P(Z,X,\alpha _{0})}{\partial z_{1}}%
\right] .
\end{equation*}%
Equation (\ref{decom_ps_non}) has a similar interpretation to equation (\ref%
{decom_f_y}). It consists of a term that ignores the estimation uncertainty
in $\hat{\alpha}$ but accounts for the variability of the sample mean, and
another term that accounts for the uncertainty in $\hat{\alpha}$ but ignores
the variability of the sample mean.

We estimate the second average derivative $T_{2}$ by 
\begin{equation}
T_{2n}(\hat{y}_{\tau },\hat{m},\hat{\alpha}):=\frac{1}{n}\sum_{i=1}^{n}\frac{%
\partial \hat{m}(\hat{y}_{\tau },P(Z_{i},X_{i},\hat{\alpha}),W_{-1,i})}{%
\partial z_{1}};  \label{t2n}
\end{equation}%
See (\ref{t2n_2}) for an explicit construction. We can regard $T_{2n}(\hat{y}%
_{\tau },\hat{m},\hat{\alpha})$ as a four-step estimator. The first step
estimates $y_{\tau }$, the second step estimates $\alpha _{0}$, the third
step estimates the conditional expectation $m_{0}(y,P(Z,X,\alpha
_{0}),W_{-1})$ using the generated regressor $P(Z,X,\hat{\alpha})$, and the
fourth step averages the derivative (with respect to $Z_{1})$ over the
generated regressor $P(Z,X,\hat{\alpha})$ and $W_{-1}$.

We use the series method to estimate $m_{0}$. To alleviate notation, define
the vector $\tilde{w}(\alpha ):=(P(z,x,\alpha ),w_{-1})^{\prime }\text{ and }%
\tilde{W}_{i}(\alpha ):=(P(Z_{i},X_{i},\alpha ),W_{-1,i})^{\prime }.$ We
write $\tilde{w}=\tilde{w}\left( \alpha _{0}\right) $ and $\tilde{W}_{i}=%
\tilde{W}_{i}(\alpha _{0})$ to suppress their dependence on the true
parameter value $\alpha _{0}.$ 
Both $\tilde{w}(\alpha )$ and $\tilde{W}_{i}(\alpha )$ are in $\mathbb{R}%
^{d_{W}}.$ Let $\phi ^{J}(\tilde{w}(\alpha ))=(\phi _{1J}(\tilde{w}(\alpha
)),\ldots ,\phi _{JJ}(\tilde{w}(\alpha )))^{\prime }$ 
be a vector of $J$ basis functions of $\tilde{w}(\alpha )$ with finite
second moments\footnote{%
If any variable in $W_{-1}$ is discrete with a small number of possible
values, we can exclude it in the basis functions. Instead, we can construct
the basis functions using only the remaining variables and apply the series
method to each subsample defined by the values of the discrete variable.}.
Here, each $\phi _{jJ}\left( \cdot \right) $ is a differentiable basis
function. Then, the series estimator of $m_{0}(y_{\tau },\tilde{w}(\alpha ))$
is $\hat{m}(\hat{y}_{\tau },\tilde{w}(\hat{\alpha}))=\phi ^{J}(\tilde{w}(%
\hat{\alpha}))^{\prime }\hat{b}(\hat{\alpha},\hat{y}_{\tau }),$ 
where $\hat{b}(\hat{\alpha},\hat{y}_{\tau })$ is: 
\begin{equation*}
\hat{b}(\hat{\alpha},\hat{y}_{\tau })=\left( \sum_{i=1}^{n}\phi ^{J}(\tilde{W%
}_{i}(\hat{\alpha}))\phi ^{J}(\tilde{W}_{i}(\hat{\alpha}))^{\prime }\right)
^{-1}\sum_{i=1}^{n}\phi ^{J}(\tilde{W}_{i}(\hat{\alpha}))\mathds{1}\left\{
Y_{i}\leq \hat{y}_{\tau }\right\} .
\end{equation*}%
%
%
%
%
%
%
%
%
%
%
%
%
%
%
%
%
%
%
%
%
%
%
%
%
%
%
%
%
%
%
%
%
%
%
%
%
%
%
%
%
%
%
%
%
%
%
%
%
%
%
%
%
%
%
%
%
%
%
%
%
%
%
%
%
%
%
%
%
%
%
%
%
%
%
%
%
%
%
The estimator of the average derivative $T_{2}$ is then 
\begin{equation}
T_{2n}(\hat{y}_{\tau },\hat{m},\hat{\alpha})=\frac{1}{n}\sum_{i=1}^{n}\frac{%
\partial \phi ^{J}(\tilde{W}_{i}(\hat{\alpha}))}{\partial z_{1}}^{\prime }%
\hat{b}(\hat{\alpha},\hat{y}_{\tau }).  \label{t2n_2}
\end{equation}

We use the path derivative approach of \cite{newey1994} to obtain a
decomposition of $T_{2n}(\hat{y}_{\tau },\hat{m},\hat{\alpha})-T_{2},$ which
is similar to that in Section 2.1 of \cite{hahn2013}. To describe the idea,
let $\left\{ F_{\theta }\right\} $ be a path of distributions indexed by $%
\theta \in \mathbb{R}$ such that $F_{\theta _{0}}$ is the true distribution
of $O:=(Y,Z,X,D)$. The parametric assumption on the propensity score does
not need to be imposed on the path.\footnote{%
As we show later, the error from estimating the propensity score does not
affect the asymptotic variance of $T_{2n}(\hat{y}_{\tau },\hat{m},\hat{\alpha%
}).$} The score of the parametric submodel is $S(O)=\frac{\partial \log
dF_{\theta }(O)}{\partial \theta }\big|_{\theta =\theta _{0}}.$ 
For any $\theta ,$ we define%
\begin{equation*}
T_{2,\theta }=E_{\theta }\left[ \frac{\partial m_{\theta }(y_{\tau ,\theta },%
\tilde{W}\left( \alpha _{\theta }\right) )}{\partial z_{1}}\right]
\end{equation*}%
where $m_{\theta },$ $y_{\tau ,\theta },$ and $\alpha _{\theta }$ are the
probability limits of $\hat{m},$ $\hat{y}_{\tau },$ and $\hat{\alpha},$
respectively, when the distribution of $O$ is $F_{\theta }$. Note that when $%
\theta =\theta _{0},$ we have $\alpha _{\theta _{0}}=\alpha _{0},m_{\theta
_{0}}=m_{0}$ and $T_{2,\theta _{0}}=T_{2}.$ Suppose the set of scores $%
\left\{ S(O)\right\} $ for all parametric submodels $\left\{ F_{\theta
}\right\} $ can approximate any zero-mean, finite-variance function of $O$
in the mean square sense.\footnote{%
This is the \textquotedblleft generality\textquotedblright\ requirement of
the family of distributions in \cite{newey1994}.} If the function $\theta
\rightarrow T_{2,\theta }$ is differentiable at $\theta _{0}$ and we can
write 
\begin{equation}
\frac{\partial T_{2,\theta }}{\partial \theta }\bigg|_{\theta =\theta _{0}}=E%
\left[ \Gamma (O)S(O)\right]  \label{path_derivative}
\end{equation}%
for some mean-zero and finite second-moment function $\Gamma (\cdot )$ and
any path $F_{\theta },$ then, by Theorem 2.1 of \cite{newey1994}, the
asymptotic variance of $T_{2n}(\hat{y}_{\tau },\hat{m},\hat{\alpha})$ is\ $%
E[\Gamma (O)^{2}]$.

In Lemma \ref{param_ave_estimation} in the appendix, we show that $\theta
\rightarrow T_{2,\theta }$ is differentiable at $\theta _{0}$. Then, by the
chain rule, we can write 
\begin{eqnarray*}
\left. \frac{\partial T_{2,\theta }}{\partial \theta }\right\vert _{\theta
=\theta _{0}} &=&\left. \frac{\partial }{\partial \theta }E_{\theta }\left[ 
\frac{\partial m_{\theta _{0}}(y_{\tau },\tilde{W}\left( \alpha _{0}\right) )%
}{\partial z_{1}}\right] \right\vert _{\theta =\theta _{0}}+\left. \frac{%
\partial }{\partial \theta }E\left[ \frac{\partial m_{\theta }(y_{\tau },%
\tilde{W}\left( \alpha _{0}\right) )}{\partial z_{1}}\right] \right\vert
_{\theta =\theta _{0}} \\
&+&\left. \frac{\partial }{\partial \theta }E\left[ \frac{\partial m_{\theta
_{0}}(y_{\tau ,\theta },\tilde{W}\left( \alpha _{0}\right) )}{\partial z_{1}}%
\right] \right\vert _{\theta =\theta _{0}}+\left. \frac{\partial }{\partial
\theta }E\left[ \frac{\partial m_{\theta _{0}}(y_{\tau },\tilde{W}\left(
\alpha _{\theta }\right) )}{\partial z_{1}}\right] \right\vert _{\theta
=\theta _{0}}.
\end{eqnarray*}%
To use Theorem 2.1 of \cite{newey1994}, we need to write all these terms in
an outer-product form, namely the form of the right-hand side of (\ref%
{path_derivative}). To search for the required function $\Gamma (\cdot )$,
we follow \cite{newey1994} and examine the components of $T_{2,\theta }$ one
at a time, treating the remaining components as known.

Lemma \ref{lemma_hahn_ridder} in the appendix shows that under some
conditions 
\begin{equation*}
\left. \frac{\partial }{\partial \theta }E\left[ \frac{\partial m_{\theta
_{0}}(y_{\tau },\tilde{W}\left( \alpha _{\theta }\right) )}{\partial z_{1}}%
\right] \right\vert _{\theta =\theta _{0}}=0,
\end{equation*}%
that is, we can ignore the error from estimating the propensity score in our
asymptotic analysis. Lemma \ref{param_ave_estimation} in the appendix
characterizes the influence function of $T_{2n}(\hat{y}_{\tau },\hat{m},\hat{%
\alpha})$ and establishes a stochastic approximation of $T_{2n}(\hat{y}%
_{\tau },\hat{m},\hat{\alpha})-T_{2}$ as follows: 
\begin{equation}
T_{2n}(\hat{y}_{\tau },\hat{m},\hat{\alpha})-T_{2}=\mathbb{P}_{n}\psi
_{\partial m_{0}}\left( W,y_{\tau }\right) +\mathbb{P}_{n}\psi
_{m_{0}}\left( Y,W,y_{\tau }\right) +\mathbb{P}_{n}\tilde{\psi}%
_{Q}(Y,y_{\tau })+o_{p}(n^{-1/2}),  \label{decom_derps_non0}
\end{equation}%
where%
\begin{align*}
\psi _{\partial m_{0}}\left( W,y_{\tau }\right) & :=\frac{\partial
m_{0}(y_{\tau },\tilde{W})}{\partial z_{1}}-T_{2}, \\
\psi _{m_{0}}\left( Y,W,y_{\tau }\right) & :=-\left[ \mathds{1}\left\{ Y\leq
y_{\tau }\right\} -m_{0}(y_{\tau },\tilde{W})\right] \times E\left[ \frac{%
\partial \log f_{W}(W)}{\partial z_{1}}\bigg|\tilde{W}\right] ,
\end{align*}%
and 
\begin{equation*}
\tilde{\psi}_{Q}(Y,y_{\tau }):=\left[ \frac{\tau -\mathds{1}\left\{ Y\leq
y_{\tau }\right\} }{f_{Y}(y_{\tau })}\right] \times E\left[ \frac{\partial
f_{Y|\tilde{W}}(y_{\tau }|\tilde{W})}{\partial z_{1}}\right] .
\end{equation*}

This characterizes the contribution of each stage to the influence function
of $T_{2n}(\hat{y}_{\tau },\hat{m},\hat{\alpha})$. The contribution from
estimating $m_{0}$, given by $\mathbb{P}_{n}\psi _{m_{0}}$, corresponds to
the one in Proposition 5 of \cite{newey1994} (p. 1362).

\subsection{Estimating the UQE}

\label{parap_uqr}

With $\hat{y}_{\tau },\hat{f}_{Y},\hat{m},\hat{\alpha}$ given in the
previous subsections, we estimate the UQE by 
\begin{equation}
\hat{\Pi}_{\tau }(\hat{y}_{\tau },\hat{f}_{Y},\hat{m},\hat{\alpha})=-\frac{1%
}{\hat{f}_{Y}(\hat{y}_{\tau })}\frac{T_{2n}(\hat{y}_{\tau },\hat{m},\hat{%
\alpha})}{T_{1n}(\hat{\alpha})}.  \label{UNIQUE_Expression}
\end{equation}%
This is our unconditional instrumental quantile estimator (UNIQUE). With the
asymptotic linear representations of all three components $\hat{f}_{Y}(\hat{y%
}_{\tau }),$ $T_{1n}(\hat{\alpha}),$ and $T_{2n}(\hat{y}_{\tau },\hat{m},%
\hat{\alpha}),$ we can obtain the asymptotic linear representation of $\hat{%
\Pi}_{\tau }(\hat{y}_{\tau },\hat{f}_{Y},\hat{m},\hat{\alpha}).$ The next
theorem follows from combining Lemmas \ref{two_step_density}, \ref%
{ps_estimation_param}, \ref{lemma_hahn_ridder}, and \ref%
{param_ave_estimation}.

\begin{theorem}
\label{uqr_if_param} Under the assumptions of Lemmas \ref{two_step_density}, %
\ref{ps_estimation_param}, \ref{lemma_hahn_ridder}, and \ref%
{param_ave_estimation}, we have 
\begin{eqnarray}
\hat{\Pi}_{\tau }-\Pi _{\tau } &=&\frac{T_{2}}{f_{Y}(y_{\tau })^{2}T_{1}}%
\left[ \mathbb{P}_{n}\psi _{f_{Y}}(Y,y_{\tau })+B_{f_{Y}}(y_{\tau })\right] +%
\frac{T_{2}}{f_{Y}(y_{\tau })^{2}T_{1}}f_{Y}^{\prime }(y_{\tau })\mathbb{P}%
_{n}\psi _{Q}(Y,y_{\tau })  \notag \\
&+&\frac{T_{2}}{f_{Y}(y_{\tau })T_{1}^{2}}\mathbb{P}_{n}\psi _{\partial
P}\left( W\right) +\frac{T_{2}}{f_{Y}(y_{\tau })T_{1}^{2}}E\left[ \frac{%
\partial ^{2}P(Z,X,\alpha _{0})}{\partial z_{1}\partial \alpha _{0}^{\prime }%
}\right] \mathbb{P}_{n}\psi _{\alpha _{0}}\left( D,W\right)
\label{est_param_pi_decomp_5} \\
&-&\frac{1}{f_{Y}(y_{\tau })T_{1}}\mathbb{P}_{n}\psi _{\partial m_{0}}\left(
W,y_{\tau }\right) -\frac{1}{f_{Y}(y_{\tau })T_{1}}\mathbb{P}_{n}\psi
_{m_{0}}\left( Y,W,y_{\tau }\right)  \notag \\
&&-\frac{1}{f_{Y}(y_{\tau })T_{1}}\mathbb{P}_{n}\tilde{\psi}_{Q}(Y,y_{\tau
})+R_{\Pi },  \notag
\end{eqnarray}%
where 
\begin{eqnarray*}
R_{\Pi } &=&O_{p}\left( |\hat{f}_{Y}(\hat{y}_{\tau })-f_{Y}(y_{\tau
})|^{2}\right) +O_{p}\left( n^{-1}\right) +O_{p}\left( n^{-1/2}|\hat{f}_{Y}(%
\hat{y}_{\tau })-f_{Y}(y_{\tau })|\right) \\
&+&o_{p}\left( n^{-1/2}h^{-1/2}\right) +o_{p}(h^{2}).
\end{eqnarray*}

Furthermore, under Assumption \ref{Assumption_rate}, $\sqrt {nh}
R_\Pi=o_p(1).$
\end{theorem}

Equation (\ref{est_param_pi_decomp_5}) consists of six influence functions
and a bias term. The bias term $B_{f_{Y}}(y_{\tau })$ arises from estimating
the density and is of order $O(h^{2})$. The six influence functions reflect
the impact of each estimation stage. The rate of convergence of $\hat{\Pi}%
_{\tau }$ is slowed down through $\mathbb{P}_{n}\psi _{f_{Y}}(Y)$, which is
of order $O_{p}(n^{-1/2}h^{-1/2})$. We can summarize the results of Theorem %
\ref{uqr_if_param} in a single equation: 
\begin{equation*}
\hat{\Pi}_{\tau }-\Pi _{\tau }=\mathbb{P}_{n}\psi _{\Pi _{\tau }}\left(
O\right) +\tilde{B}_{f_{Y}}(y_{\tau })+o_{p}(n^{-1/2}h^{-1/2}),
\end{equation*}%
where $\psi _{\Pi _{\tau }}$ collects all the influence functions in (\ref%
{est_param_pi_decomp_5}) except for the bias, and 
\begin{equation*}
\tilde{B}_{f_{Y}}(y_{\tau }):=\frac{T_{2}}{f_{Y}(y_{\tau })^{2}T_{1}}%
B_{f_{Y}}(y_{\tau }).
\end{equation*}%
If $nh^{5}\rightarrow 0,$ then the bias term is $o(n^{-1/2}h^{-1/2})$. The
following corollary provides the asymptotic distribution of $\hat{\Pi}_{\tau
}$.

\begin{corollary}
\label{corollary_param}Under the assumptions of Theorem \ref{uqr_if_param}
and the assumption that $nh^{5}\rightarrow 0,$ 
\begin{equation*}
\sqrt{nh}\left( \hat{\Pi}_{\tau }-\Pi _{\tau }\right) =\sqrt{n}\mathbb{P}_{n}%
\sqrt{h}\psi _{\Pi _{\tau }}\left( O\right) +o_{p}(1)\Rightarrow \mathcal{N}%
(0,V_{\tau }),
\end{equation*}%
where 
\begin{equation}
V_{\tau }=\lim_{h\downarrow 0}E\left\{ h\left[ \psi _{\Pi _{\tau }}\left(
O\right) \right] ^{2}\right\} .  \label{variance_param}
\end{equation}
\end{corollary}

From the perspective of asymptotic theory, all of the following terms are
all of order $O_{p}\left( h\right) =o_{p}\left( 1\right) $ and hence can be
ignored in large samples: $\sqrt{nh}\mathbb{P}_{n}\psi _{Q},$ $\sqrt{nh}%
\mathbb{P}_{n}\psi _{\partial P},$ $\sqrt{nh}\mathbb{P}_{n}\psi _{\alpha
_{0}},$ $\sqrt{nh}\mathbb{P}_{n}\psi _{\partial m_{0}},$ $\sqrt{nh}\mathbb{P}%
_{n}\psi _{m_{0}},$ and $\sqrt{nh}\mathbb{P}_{n}\tilde{\psi}_{Q}.$ The
asymptotic variance is then given by 
\begin{equation*}
V_{\tau }=\frac{T_{2}^{2}}{f_{Y}(y_{\tau })^{4}T_{1}^{2}}\lim_{h\downarrow
0}E\left[ h\psi _{f_{Y}}^{2}(y_{\tau })\right] =\frac{T_{2}^{2}}{%
f_{Y}(y_{\tau })^{3}T_{1}^{2}}\int_{-\infty }^{\infty }K^{2}\left( u\right)
du.
\end{equation*}%
However, $V_{\tau }$ ignores all estimation uncertainties except that in $%
\hat{f}_{Y}(y_{\tau })$, and we do not expect it to reflect the
finite-sample variability of $\sqrt{nh}(\hat{\Pi}_{\tau }-\Pi _{\tau })$
well. To improve the finite-sample performances, we keep the dominating term
from each source of estimation errors and employ a sample counterpart of $%
E[h\psi _{\Pi _{\tau }}^{2}]$ to estimate $V_{\tau }.$ The details can be
found in Section \ref{estimation_variance_appendix} of the supplementary
appendix.%

\subsection{Testing the Null of No Effect}

\label{param_hypothesis}

We can apply Corollary \ref{corollary_param} to conduct hypothesis testing
on $\Pi _{\tau }$. Since $\hat{\Pi}_{\tau }$ converges to $\Pi _{\tau }$ at
a nonparametric rate, in general, the test will have power only against a
local departure of a nonparametric rate. However, if we are interested in
testing the null of a zero effect, that is, $H_{0}:\Pi _{\tau }=0$ vs. $%
H_{1}:\Pi _{\tau }\neq 0,$ we can detect a parametric rate of departure from
the null. The reason is that, by (\ref{param_interest_3}), $\Pi _{\tau }=0$
if and only if $T_{2}=0,$ and $T_{2}$ can be estimated at the usual
parametric rate. Hence, instead of testing $H_{0}:\Pi _{\tau }=0$ vs. $%
H_{1}:\Pi _{\tau }\neq 0,$ we can test the equivalent hypotheses $%
H_{0}:T_{2}=0$ vs. $H_{1}:T_{2}\neq 0.$

Our test is based on the estimator $T_{2n}(\hat{y}_{\tau },\hat{m},\hat{%
\alpha})$ of $T_{2}$. In view of its influence function given in Lemma \ref%
{param_ave_estimation}, we can estimate the asymptotic variance of $T_{2n}(%
\hat{y}_{\tau },\hat{m},\hat{\alpha})$ by 
\begin{equation*}
\hat{V}_{2}=\frac{1}{n}\sum_{i=1}^{n}\left( \hat{\psi}_{\partial m,i}+\hat{%
\psi}_{m,i}+\hat{E}\left[ \frac{\partial f_{Y|\tilde{W}\left( \alpha
_{0}\right) }(\hat{y}_{\tau }|\tilde{W}\left( \hat{\alpha}\right) )}{%
\partial z_{1}}\right] \hat{\psi}_{Q,i}\right) ^{2}
\end{equation*}%
where $\hat{\psi}_{\partial m,i},$ $\hat{\psi}_{m,i}$ and $\hat{\psi}_{Q,i}$
are plug-in estimates of $\psi _{\partial m_{0}}\left( W_{i},y_{\tau
}\right) ,\psi _{m_{0}}\left( Y_{i},W_{i},y_{\tau }\right) ,$ and $\hat{\psi}%
_{Q}\left( Y_{i},y_{\tau }\right) ,$ respectively.

%
%
%
%
%
%
%
%
%
%
%
%
%
%
%
%
%
%
%
%
%
%
%
%
%
%
%
%
%
%
%
%
%
%
%
%
%
%
%
%
%
%
%
%
%
%
%
%
%
%
%
%
%
%
%
%
%
%
%
%
%
%
%
%
%
%
%
%
%
%
%
%
%
%
%
%
%
%
%
%
%
%
%
%
%
%
%
%
%
We can then form the test statistic: 
\begin{equation}
T_{2n}^{o}:=\sqrt{n}\frac{T_{2n}(\hat{y}_{\tau },\hat{m},\hat{\alpha})}{%
\sqrt{\hat{V}_{2}}}.  \label{test_statistic_t2}
\end{equation}%
By Lemma \ref{param_ave_estimation} and using standard arguments, we can
show that $T_{2n}^{o}\Rightarrow \mathcal{N}(0,1).$ To save space, we omit
the details here.

\section{Simulation Evidence}

\label{simulation}

In the simulation study, we consider the following structural equations: 
\begin{align*}
Y(0)& =X_{1}+X_{2}+U_{0}, \\
Y(1)& =\beta +X_{1}+X_{2}+U_{1}, \\
D& =\mathds{1}\left\{ (Z_{1}+Z_{2}+Z_{3}+X_{1}+X_{2})/\sqrt{5}>V\right\} .
\end{align*}%
Here, $(U_{0},U_{1},Z_{1},Z_{2},Z_{3},X_{1},X_{2})$ are jointly normal,
independent variables, with mean $0$ and unit variances, and 
\begin{equation*}
V=\frac{\varrho}{\sqrt{2}} U_{0}+\frac{\varrho}{\sqrt{2}} U_{1}+\sqrt{%
1-\varrho ^{2}}e,
\end{equation*}%
where $e$ is standard normal, independent of $%
(U_{0},U_{1},Z_{1},Z_{2},Z_{3},X_{1},X_{2})$.\footnote{%
Note that $V$ has also unit variance.} The parameter $\varrho $ governs the
endogeneity of $D$.

The target of interest is the UQE $\Pi _{\tau }$, as defined in (\ref%
{equ_pi_tau}), involving $y_{\tau },f_{Y},$ $T_{1},$ and $T_{2}$.\footnote{%
Detailed calculations of $\Pi _{\tau }$ are available from Appendix B of an
earlier working paper \cite{sun2021},} We estimate $\Pi _{\tau }$ using the
UNIQUE outlined in Section \ref{estimation}. More specifically, for
estimating the (population) quantile $y_{\tau },$ we utilize the sample
quantile. To estimate $f_{Y}$, we use a Gaussian kernel with bandwidth $%
h=1.06\times \hat{\sigma}_{Y}\times n^{-1/4}$, where $\hat{\sigma}_{Y}$ is
the sample standard deviation of $Y$. For estimating $T_{1}$, we employ a
probit model, and for $T_{2}$, we run a cubic series regression.

\subsection{Testing the Null Hypothesis of No Effect}

If we set $\beta =0$, $\Pi _{\tau }=0$ and so the null hypothesis of a zero
effect holds. The test statistic $T_{2n}^{o}$ is constructed following
equation \eqref{test_statistic_t2}. Because the test statistic does not
involve estimating the density (or $T_{1}$), the test has nontrivial power
again $1/\sqrt{n}$-departures (i.e., $\beta =c/\sqrt{n}$ for some $c\neq 0)$
from the null.

To simulate the power function of the nominal 5\% test, we consider a range
of 25 values of $\beta $ between $-1$ and $1$. The endogeneity, governed by
the parameter $\varrho $, takes five values: 0, 0.25, 0.5, 0.75, and 0.9. We
perform 1,000 simulations with 1,000 observations. For different values of $%
\tau $, the power functions are shown below in Figure \ref{power_1}.%
\ The test has the desired null rejection probability, except for the
extreme quantile $\tau =0.1$, where under high endogeneity, the rejection
probability does not increase fast enough. The power function for the
median, $\tau =0.5$, is not shown here as it is almost identical to that of $%
\tau =0.4$. Furthermore, simulation results not reported here show that the
power functions for $\tau =0.6,0.7,0.8,0.9$ are very similar to those of $%
\tau =0.4,0.3,0.2,0.1,$ respectively.

\begin{figure}[h]
\centering
\includegraphics[scale =0.7]{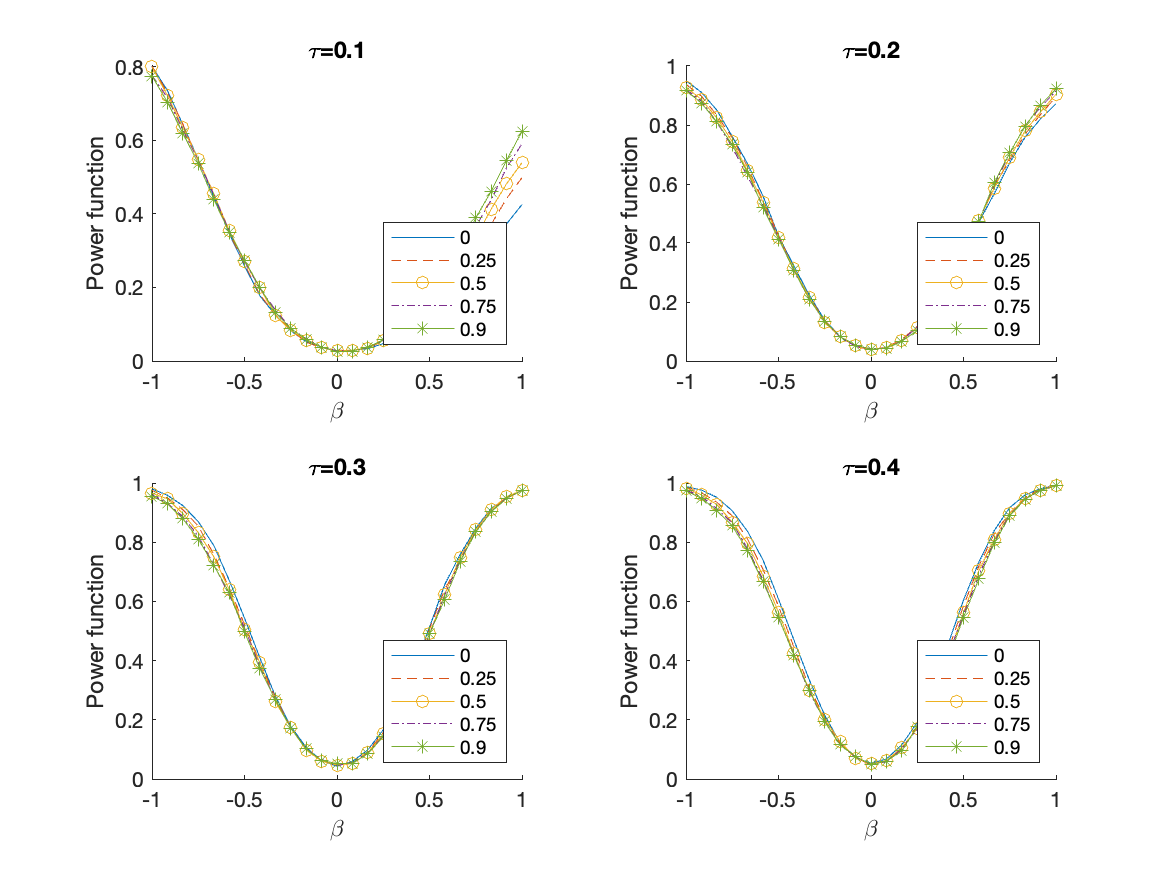}
\caption{ Power functions for different values of $\protect\varrho $ and $%
\protect\tau $.}
\label{power_1}
\end{figure}

\subsection{Empirical Coverage of Confidence Intervals}

In this subsection, we investigate the empirical coverage of confidence
intervals built using $\hat{V}_{\tau }$, the variance estimator given in %
\eqref{var_param_if}. Since $\sqrt{nh}\left( \hat{\Pi}_{\tau }-\Pi _{\tau
}\right) \approx \mathcal{N}(0,\hat{V}_{\tau }),$ a 95\% confidence interval
for $\Pi _{\tau }$ can be constructed using 
\begin{equation*}
\hat{\Pi}_{\tau }\pm 1.96\times \frac{\sqrt{\hat{V}_{\tau }}}{\sqrt{nh}}.
\end{equation*}

We use a grid of $\beta $ that takes values $-1,-0.5,-0.25,0,0.25,0.5,$ and $%
1$. For the endogeneity parameter, $\varrho $, we take the values $%
0,0.25,0.5,0.75,$ and $0.9$. Finally, $\tau $ takes the values from $0.1$ to 
$0.9$ with an increment of 0.1. We note that, for values of $\beta \neq 0$,
where the effect is not 0, we need to numerically compute the value of $\Pi
_{\tau }$. We perform 10,000 simulations with 1,000 observations each. The
results are reported in the tables below for $\tau =0.1$ and $\tau =0.5$. It
is clear that the confidence intervals have reasonable coverage accuracy in
almost all cases.

\begin{table}[h]
\centering%
\begin{tabular}{llllll}
$\beta \setminus \varrho $ & 0 & 0.25 & 0.5 & 0.75 & 0.9 \\ \hline\hline
-1 & 0.9701 & 0.9739 & 0.9717 & 0.9715 & 0.9740 \\ 
-0.5 & 0.9696 & 0.9706 & 0.9754 & 0.9691 & 0.9745 \\ 
-0.25 & 0.9732 & 0.9724 & 0.9753 & 0.9743 & 0.9731 \\ 
0 & 0.9754 & 0.9727 & 0.9746 & 0.9750 & 0.9782 \\ 
0.25 & 0.9793 & 0.9780 & 0.9758 & 0.9775 & 0.9768 \\ 
0.5 & 0.9807 & 0.9824 & 0.9800 & 0.9817 & 0.9772 \\ 
1 & 0.9808 & 0.9829 & 0.9821 & 0.9789 & 0.9808 \\ \hline
\end{tabular}%
\caption{Empirical coverage of 95\% confidence intervals for $\protect\tau =
0.1$.}
\end{table}

\bigskip

\bigskip

\begin{table}[h]
\centering%
\begin{tabular}{llllll}
$\beta \setminus \varrho $ & 0 & 0.25 & 0.5 & 0.75 & 0.9 \\ \hline\hline
-1 & 0.9679 & 0.9741 & 0.9737 & 0.9769 & 0.9748 \\ 
-0.5 & 0.9608 & 0.9622 & 0.9623 & 0.9705 & 0.9697 \\ 
-0.25 & 0.9587 & 0.9605 & 0.9630 & 0.9607 & 0.9602 \\ 
0 & 0.9616 & 0.9565 & 0.9569 & 0.9591 & 0.9524 \\ 
0.25 & 0.9579 & 0.9598 & 0.9594 & 0.9571 & 0.9579 \\ 
0.5 & 0.9640 & 0.9587 & 0.9551 & 0.9575 & 0.9569 \\ 
1 & 0.9621 & 0.9667 & 0.9666 & 0.9672 & 0.9668 \\ \hline
\end{tabular}%
\caption{Empirical coverage of 95\% confidence intervals for $\protect\tau =
0.5$.}
\end{table}


\section{Empirical Application}

\label{empirical}

We estimate the unconditional quantile effect of expanding college
enrollment on (log) wages. The outcome variable $Y$ is the log wage, and the
binary treatment is the college enrollment status. Thus, $p=\Pr [D=1]$ is
the proportion of individuals who ever enrolled in a college. Arguably, the
cost of tuition $(Z_{1})$, assumed to be continuous, is an important factor
that affects the college enrollment status but not the wage directly. In
order to alter the proportion of enrolled individuals, we consider a policy
that subsidizes tuition by a certain amount. The UQE is the effect of this
policy on the different quantiles of the unconditional distribution of wages
when the subsidy is small. 
This policy shifts $Z_{1}$, the tuition, to $Z_{1\delta }=Z_{1}+s(\delta )$
for some $s(\delta )$, which is the same for all individuals, and induces a
small change in college enrollment. Note that we do not need to specify $%
s(\delta )$ because we look at the limiting version as $\delta \rightarrow 0$%
. In practice, we may set $s(\delta )$ equal to a number that is relatively
small compared to the total tuition.

We use the same data as in \cite{Carneiro2010} and \cite{Carneiro2011}: a
sample of white males from the 1979 National Longitudinal Survey of Youth
(NLSY1979). The web appendix to \cite{Carneiro2011} contains a detailed
description of the variables. The outcome variable $Y$ is the log wage in
1991. The treatment indicator $D$ is equal to $1$ if the individual ever
enrolled in college by 1991, and $0$ otherwise. The other covariates are
AFQT score, mother's education, number of siblings, average log earnings
1979--2000 in the county of residence at age 17, average unemployment
1979--2000 in the state of residence at age 17, urban residence dummy at age
14, cohort dummies, years of experience in 1991, average local log earnings
in 1991, and local unemployment in 1991. We collect these variables into a
vector and denote it by $X$.

We assume that the following four variables (denoted by $%
Z_{1},Z_{2},Z_{3},Z_{4}$) enter the selection equation but not the outcome
equation: tuition at local public four-year colleges at age 17, presence of
a four-year college in the county of residence at age 14, local earnings at
age 17, and local unemployment at age 17. The total sample size is 1747, of
which 882 individuals had never enrolled in a college ($D=0$) by 1991, and
865 individuals had enrolled in a college by 1991 ($D=1)$. We estimate the
UQE of a marginal shift in the tuition at local public four-year colleges at
age 17 ($Z_{1})$ using the UNIQUE.

Here are some details of the UNIQUE. To estimate the propensity score, we
use a parametric logistic specification. To estimate the conditional
expectation function $m_{0}$,$\ $we run a series regression using the
estimated propensity score and the covariates $Z_{2},Z_{3},Z_{4},$ and $X$\
as the regressors. Due to the large number of variables involved, a
penalization of $\lambda =10^{-4}$ was imposed on the $L_{2}$-norm of the
coefficients, excluding the constant term as in ridge regressions. We
estimate the UQE at the quantile level $\tau =0.1,0.15,\ldots ,0.9$. For
each $\tau ,$ we also construct the 95\% (pointwise) confidence interval.

Figure \ref{emp_app_1} presents the results. The estimated UQE ranges
between 0.22 and 0.47 across the quantiles with an average of 0.37. When we
estimate the unconditional mean effect, we obtain an estimate of 0.21, which
is somewhat consistent with the quantile cases. We interpret these estimates
in the following way: the effect of a $\delta $ (small) increase in college
enrollment induced by an additive change in tuition increases (log) wages
between $0.22\times \delta $ and $0.47\times \delta $ across quantiles. For
example, for $\delta =0.01$, we obtain an increase in the quantiles of the
wage distribution between $0.22\%$ and $0.47\%$.

\begin{figure}[h]
\centering
\includegraphics[scale =0.7]{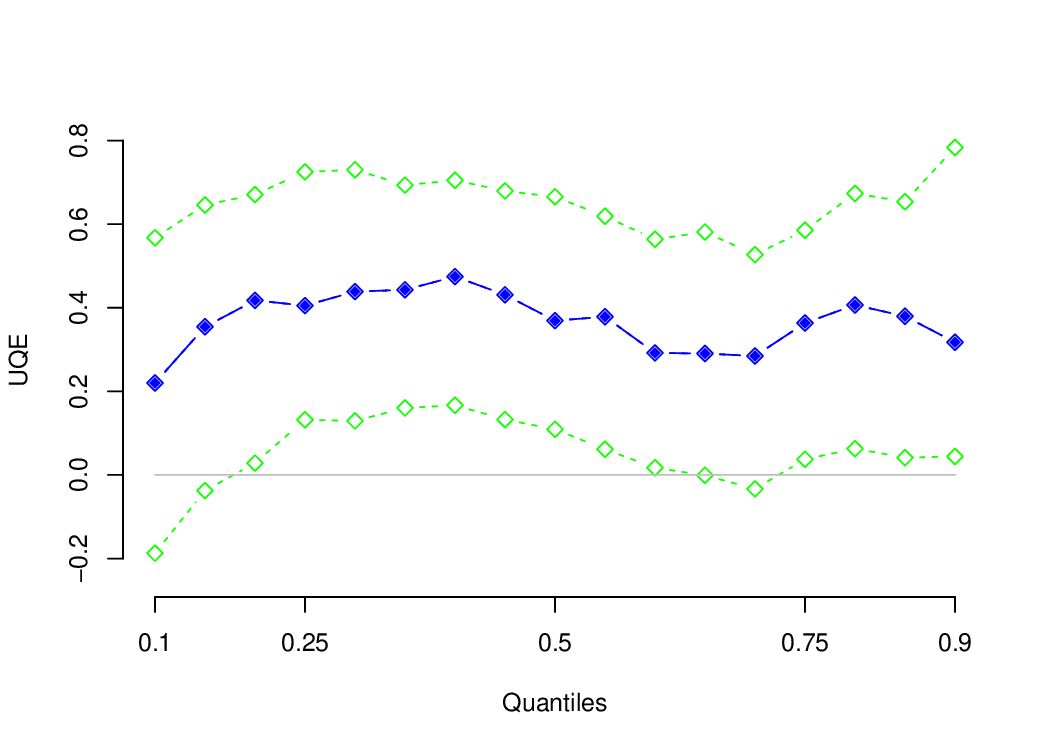}
\caption{ \emph{Solid line: point estimates of the UQEs; Dashed line: 95\%
confidence intervals.}}
\label{emp_app_1}
\end{figure}

\section{Conclusion}

\label{conclusion}

In this paper we study the unconditional policy effect with an endogenous
binary treatment. Framing the selection equation as a threshold-crossing
model allows us to introduce a novel class of unconditional marginal
treatment effects and represent the unconditional effect as a weighted
average of these unconditional marginal treatment effects. When the policy
variable used to change the participation rate satisfies a conditional
exogeneity condition, it is possible to recover the unconditional policy
effect using the proposed UNIQUE method.

To illustrate the usefulness of unconditional MTEs, we focus on the
unconditional quantile effect. We find that the unconditional quantile
regression estimator that neglects endogeneity can be severely biased.
Moreover, the bias may not be uniform across quantiles. Any attempt to sign
the bias \emph{a priori} requires very strong assumptions on the
data-generating process. Intriguingly, the unconditional quantile regression
estimator can be inconsistent even if the treatment status is exogenously
determined. This happens when the treatment selection is partly determined
by some covariates that also influence the outcome variable.

Our findings reveal that the unconditional quantile effect (UQE) and the
marginal policy-relevant treatment effect (MPRTE) can be seen as part of the
same family of policy effects. From a purely robustness perspective, the
unconditional median effect---a special UQE---can be considered a more
robust version of the MPRTE, in the same way that the median is the robust
counterpart of the mean. Both represent specific examples of a general
unconditional policy effect. To the best of our knowledge, this connection
has not been established in the literature on either UQE or MPRTE.

\section*{Appendix: Proof of the Main Results}

\renewcommand{\thesubsection}{\Alph{subsection}} \setcounter{equation}{0} %
\renewcommand{\theequation}{\textcolor{red}{A.\arabic{equation}}}


\label{appendix_marginal}



To prove Theorem \ref{Theorem general}, we first present an auxiliary
proposition with some remarks.

\begin{proposition}
\label{Prop_linearization}Let Assumptions \ref{Assumption_primary}--\ref%
{Assumption_domination} hold. Then%
\begin{eqnarray*}
F_{Y_{\delta }}(y) &=&F_{Y}(y)+\delta E\left[ \left\{
F_{Y(1)|U_{D},W}(y|P\left( W\right) ,W)-F_{Y(0)|U_{D},W}(y|P\left( W\right)
,W)\right\} \dot{P}\left( W\right) \right] \\
&+&R_{F}(\delta ;y),
\end{eqnarray*}%
where $R_{F}(\delta ;y)/\delta \rightarrow 0$ uniformly over $y\in \mathcal{Y%
}:=\mathcal{Y}(0)\cup \mathcal{Y}(1)$ as $\delta \rightarrow 0$.
\end{proposition}

\begin{remark}
Proposition \ref{Prop_linearization} provides a linear approximation to $%
F_{Y_{\delta }}$, the CDF of the outcome variable under $D_{\delta }.$
Essentially, it says that the proportion of individuals with outcome below $%
y $ under the new policy regime, that is, $F_{Y_{\delta }}(y)$, will be
equal to the proportion of individuals with outcome below $y$ under the
existing policy regime, that is, $F_{Y}(y)$, plus an adjustment given by the
marginal entrants. Consider $\delta >0$ and $P_{\delta }\left( w\right)
>P\left( w\right) $ for all $w\in \mathcal{W}$ as an example. In this case,
because of the policy intervention, the individuals who are on the margin,
namely those with $u_{D}=P(w)$, will switch their treatment status from $0$
to $1$. Such a switch contributes to $F_{Y_{\delta }}(y)$ by the amount $%
F_{Y(1)|U_{D},W}\left( P(w\right) ,w)-F_{Y(0)|U_{D},W}\left( P(w\right) ,w)$%
, averaged over the distribution of $W$ for the marginal subpopulation.
\end{remark}

\begin{remark}
The linear approximation to $F_{Y_{\delta }}$ is uniform over $\mathcal{Y}$
as $\delta \rightarrow 0$. We need the uniform approximation because we
consider a general Hadamard differentiable $\rho .$ In the special case when 
$\rho $ does not depend on the whole distribution, the uniformity of the
approximation over the whole support $\mathcal{Y}$ may not be necessary. In
the quantile case, we maintain the uniformity because we consider all
quantile levels in $(0,1).$
\end{remark}

\begin{proof}[Proof of Theorem \protect\ref{Theorem general}]
Define 
\begin{equation}
G\left( y\right) =E\left[ \left\{ F_{Y(1)|U_{D},W}(y|P\left( W\right)
,W)-F_{Y(0)|U_{D},W}(y|P\left( W\right) ,W)\right\} \dot{P}\left( W\right) %
\right] .  \label{equation G}
\end{equation}%
Then, by Proposition \ref{Prop_linearization}, we have $\lim_{\delta
\rightarrow 0}\left\Vert G_{\delta }-G\right\Vert _{\infty }=0$ for $%
G_{\delta }$ defined in (\ref{G_Delta}). Hence, under the Hadamard
differentiability of $\rho ,$ we obtain $\Pi _{\rho }=\dot{\rho}_{F_{Y}}%
\left[ G\right] /E\left[ \dot{P}\left( W\right) \right] .$

Next, we show that $\dot{\rho}_{F_{Y}}\left[ G\right] =\int_{\mathcal{Y}%
}\psi (y,\rho ,F_{Y})dG(y)$. By definition (cf. Definition \ref{Def: IF}), 
\begin{equation*}
\psi (y,\rho ,F_{Y})=\dot{\rho}_{F_{Y}}\left[ \Delta _{y}\left( \cdot
\right) -F_{Y}\left( \cdot \right) \right] .
\end{equation*}%
Using this and $\lim_{y\rightarrow -\infty }G\left( y\right)
=\lim_{y\rightarrow +\infty }G\left( y\right) =0,$ we have 
\begin{eqnarray*}
\int_{\mathcal{Y}}\psi (y,\rho ,F_{Y})dG(y) &=&\int_{\mathcal{Y}}\dot{\rho}%
_{F_{Y}}\left[ \Delta _{y}\left( \cdot \right) -F_{Y}\left( \cdot \right) %
\right] dG(y) \\
&=&\dot{\rho}_{F_{Y}}\left[ \int_{\mathcal{Y}}\left( \Delta _{y}\left( \cdot
\right) -F_{Y}\left( \cdot \right) \right) dG(y)\right] \\
&=&\dot{\rho}_{F_{Y}}\left[ G(\cdot )\right] ,
\end{eqnarray*}%
where the second equality follows from the linearity and continuity of the
functional $\dot{\rho}_{F}$ and the third equality follows from $\int_{%
\mathcal{Y}}\Delta _{y}\left( \cdot \right) dG(y)=G\left( \cdot \right) $
and $\int_{\mathcal{Y}}F_{Y}\left( \cdot \right) dG(y)=F_{Y}\left( \cdot
\right) \int_{\mathcal{Y}}dG(y)\equiv 0.$ Therefore, 
\begin{equation*}
\Pi _{\rho }=\int_{\mathcal{Y}}\psi (y,\rho ,F_{Y})dG(y)/E\left[ \dot{P}%
\left( W\right) \right] .
\end{equation*}

It remains to show that $G\left( y\right) $ is differentiable with 
\begin{equation*}
dG\left( y\right) =E\left[ \left\{ f_{Y(1)|U_{D},W}(y|P\left( W\right)
,W)-f_{Y(0)|U_{D},W}(y|P\left( W\right) ,W)\right\} \dot{P}\left( W\right) %
\right] dy.
\end{equation*}%
To this end, we consider the limit 
\begin{eqnarray*}
&&\lim_{h\rightarrow 0}\frac{1}{h}\left[ G\left( y+h\right) -G\left(
y\right) \right] \\
&=&\lim_{h\rightarrow 0}E\left[ \left\{ \frac{F_{Y(1)|U_{D},W}(y+h|P\left(
W\right) ,W)-F_{Y(1)|U_{D},W}(y|P\left( W\right) ,W)}{h}\right\} \dot{P}%
\left( W\right) \right] \\
&&-\lim_{h\rightarrow 0}E\left[ \left\{ \frac{F_{Y(0)|U_{D},W}(y+h|P\left(
W\right) ,W)-F_{Y(0)|U_{D},W}(y|P\left( W\right) ,W)}{h}\right\} \dot{P}%
\left( W\right) \right] .
\end{eqnarray*}%
But%
\begin{eqnarray*}
&&\sup_{h}\left\vert \frac{F_{Y(d)|U_{D},W}(y+h|P\left( W\right)
,W)-F_{Y(d)|U_{D},W}(y|P\left( W\right) ,W)}{h}\right\vert \\
&\leq &\sup_{\tilde{y}\in \mathcal{Y}(d)}\left\vert f_{Y(d)|U_{D},W}(\tilde{y%
}|P\left( W\right) ,W)\right\vert .\text{ }
\end{eqnarray*}%
This and Assumptions \ref{Assumption_regularity}(\ref{f_y_u_x}.ii) and \ref%
{Assumption_regularity}(\ref{p_delta}.iii) allow us to use the dominated
convergence theorem to obtain%
\begin{equation*}
\lim_{h\rightarrow 0}\frac{1}{h}\left[ G\left( y+h\right) -G\left( y\right) %
\right] =E\left[ \left\{ f_{Y(1)|U_{D},W}(y|P\left( W\right)
,W)-f_{Y(0)|U_{D},W}(y|P\left( W\right) ,W)\right\} \dot{P}\left( W\right) %
\right]
\end{equation*}%
and hence $G\left( y\right) $ is indeed differentiable with the derivative
given above. Therefore, $\Pi _{\rho }$ is 
\begin{eqnarray*}
\int_{\mathcal{Y}}\psi (y,\rho ,F_{Y})dG\left( y\right) &=&\int_{\mathcal{Y}%
}\psi (y,\rho ,F_{Y}) \\
&\times &E\left[ \left\{ f_{Y(1)|U_{D},W}(y|P\left( W\right)
,W)-f_{Y(0)|U_{D},W}(y|P\left( W\right) ,W)\right\} \dot{P}\left( W\right) %
\right] dy.
\end{eqnarray*}%
To represent $\Pi _{\rho }$ in an alternative form, recall that the
unconditional marginal treatment effect for the $\rho $ functional is 
\begin{equation*}
\mathrm{MTE}_{\rho }\left( u,w\right) :=E\left[ \psi (Y\left( 1\right) ,\rho
,F_{Y})-\psi (Y\left( 0\right) ,\rho ,F_{Y})|U_{D}=u,W=w\right] .
\end{equation*}%
By changing the order of integration, we can write $\Pi _{\rho }$ as%
\begin{eqnarray*}
&&\Pi _{\rho } \\
&=&\int_{\mathcal{W}}\left\{ \int_{\mathcal{Y}}\psi (y,\rho ,F_{Y})\left[
f_{Y(1)|U_{D},W}(y|P\left( w\right) ,w)-f_{Y(0)|U_{D},W}(y|P\left( w\right)
,w)\right] dy\right\} \frac{\dot{P}\left( w\right) }{E\left[ \dot{P}\left(
W\right) \right] }dF_{W}\left( w\right) \\
&=&\int_{\mathcal{W}}E\left[ \psi (Y\left( 1\right) ,\rho ,F_{Y})-\psi
(Y\left( 0\right) ,\rho ,F_{Y})|U_{D}=P(w),W=w\right] \mathcal{\dot{P}}%
\left( w\right) dF_{W}\left( w\right) \\
&=&\int_{\mathcal{W}}\mathrm{MTE}_{\rho }\left( P(w),w\right) \dot{P}\left(
w\right) dF_{W}\left( w\right) .
\end{eqnarray*}
\end{proof}

%

\begin{proof}[Proof of Proposition \protect\ref{Prop_linearization}]
Using the selection equation $D_{\delta }=\mathds{1}\left\{ U_{D}\leq
P_{\delta }(W)\right\} ,$ we have%
\begin{eqnarray*}
F_{Y(1)|D_{\delta }}(\tilde{y}|1) &=&\Pr \left[ Y(1)\leq \tilde{y}|D_{\delta
}=1\right] \\
&=&\frac{1}{\Pr [D_{\delta }=1]}\Pr \left[ Y(1)\leq \tilde{y},D_{\delta }=1%
\right] \\
&=&\frac{1}{\Pr [D_{\delta }=1]}\int_{\mathcal{W}}\Pr \left[ Y(1)\leq \tilde{%
y},D_{\delta }=1|W=w\right] dF_{W}\left( w\right) \\
&=&\frac{1}{\Pr [D_{\delta }=1]}\int_{\mathcal{W}}\Pr \left[ Y(1)\leq \tilde{%
y},U_{D}\leq P_{\delta }\left( w\right) |W=w\right] dF_{W}\left( w\right) .
\end{eqnarray*}%
Hence, under Assumptions \ref{Assumption_primary}(\ref{feasibility}) and \ref%
{Assumption_regularity}(\ref{regularity_x_abs}),%
\begin{align*}
F_{Y(1)|D_{\delta }}(\tilde{y}|1)& =\frac{1}{\Pr [D_{\delta }=1]}\int_{%
\mathcal{W}}F_{Y(1),U_{D}|W}(\tilde{y},P_{\delta }(w)|w)dF_{W}\left( w\right)
\\
& =\frac{1}{\Pr [D_{\delta }=1]}\int_{\mathcal{W}}\int_{-\infty }^{\tilde{y}%
}\int_{0}^{P_{\delta }(w)}f_{Y(1),U_{D}|W}(y,u|w)dudydF_{W}\left( w\right) \\
& =\frac{1}{\Pr [D_{\delta }=1]}\int_{-\infty }^{\tilde{y}}\int_{\mathcal{W}}%
\left[ \int_{0}^{P_{\delta }(w)}f_{Y(1),U_{D}|W}(y,u|w)du\right]
dF_{W}\left( w\right) dy,
\end{align*}%
where the order of integration can be switched because the integrands are
non-negative. It then follows that 
\begin{eqnarray}
f_{Y(1)|D_{\delta }}(\tilde{y}|1) &=&\frac{1}{\Pr [D_{\delta }=1]}\int_{%
\mathcal{W}}\left[ \int_{0}^{P_{\delta }(w)}f_{Y(1),U_{D}|W}(\tilde{y},u|w)du%
\right] dF_{W}\left( w\right)  \notag \\
&=&\frac{1}{\Pr [D_{\delta }=1]}\int_{\mathcal{W}}\left[ \int_{0}^{P_{\delta
}(w)}f_{Y(1)|U_{D},W}(\tilde{y}|u,w)du\right] dF_{W}\left( w\right) ,
\label{f_y1_d1}
\end{eqnarray}%
where we have used Assumption \ref{Assumption_primary}(\ref{uniformity}). So,%
\begin{equation}
\Pr [D_{\delta }=1]f_{Y(1)|D_{\delta }}(\tilde{y}|1)=\int_{\mathcal{W}}\left[
\int_{0}^{P_{\delta }(w)}f_{Y(1)|U_{D},W}(\tilde{y}|u,w)du\right]
dF_{W}\left( w\right) .  \label{p_f_y1_d1}
\end{equation}

Under Assumptions \ref{Assumption_regularity}(\ref{f_y_u_x}) and \ref%
{Assumption_regularity}(\ref{p_delta}), we can differentiate both sides of (%
\ref{p_f_y1_d1}) with respect to $\delta $ under the integral sign to get 
\begin{equation}
\frac{\partial \left\{ \Pr [D_{\delta }=1]f_{Y(1)|D_{\delta }}(\tilde{y}%
|1)\right\} }{\partial \delta }=\int_{\mathcal{W}}f_{Y(1)|U_{D},W}(\tilde{y}%
|P_{\delta }(w),w)\frac{\partial P_{\delta }(w)}{\partial \delta }%
dF_{W}\left( w\right) .  \label{partial_p_f_y1_d1}
\end{equation}

Under Assumptions \ref{Assumption_regularity}(\ref{f_y_u_x}.i) and \ref%
{Assumption_regularity}(\ref{p_delta}.ii), $f_{Y(1)|U_{D},W}(\tilde{y}%
|P_{\delta }(w),w)\partial P_{\delta }(w)/\partial \delta $ is continuous in 
$\delta $ for each $\tilde{y}\in \mathcal{Y}\left( 1\right) $ and $w\in 
\mathcal{W}$. In view of Assumptions \ref{Assumption_regularity}(\ref%
{f_y_u_x}.ii) and \ref{Assumption_regularity}(\ref{p_delta}.iii), we can
invoke the dominated convergence theorem to show that the map $\delta
\mapsto \frac{\partial \left\{ \Pr [D_{\delta }=1]f_{Y(1)|D_{\delta }}(%
\tilde{y}|1)\right\} }{\partial \delta }$ is continuous for each $\tilde{y}%
\in \mathcal{Y}\left( 1\right) $.


The same argument can be used to show the continuous differentiability of $%
\Pr [D_{\delta }=0)]f_{Y(0)|D_{\delta }}(\tilde{y}|0)$. More specifically,
under Assumption \ref{Assumption_regularity}(\ref{regularity_x_abs}), we
have 
\begin{align*}
F_{Y(0)|D_{\delta }}(\tilde{y}|0)& =\Pr [Y(0)\leq \tilde{y}|D_{\delta }=0] \\
& =\Pr [Y(0)\leq \tilde{y}|U_{D}>P_{\delta }\left( W\right) ] \\
& =\frac{1}{\Pr [D_{\delta }=0]}\Pr [Y(0)\leq \tilde{y},U_{D}>P_{\delta
}\left( W\right) ] \\
& =\frac{1}{\Pr [D_{\delta }=0]}\left\{ \Pr \left[ Y(0)\leq \tilde{y}\right]
-\Pr \left[ Y(0)\leq \tilde{y},U_{D}\leq P_{\delta }\left( W\right) \right]
\right\} \\
& =\frac{1}{\Pr [D_{\delta }=0]}\left[ F_{Y(0)}(\tilde{y})-\int_{\mathcal{W}%
}F_{Y(0),U_{D}|W}(\tilde{y},P_{\delta }\left( w\right) |w)dF_{W}(w)\right] \\
& =\frac{1}{\Pr [D_{\delta }=0]}\left[ F_{Y(0)}(\tilde{y})-\int_{\mathcal{W}%
}\int_{-\infty }^{\tilde{y}}\int_{0}^{P_{\delta }\left( w\right)
}f_{Y(0),U_{D}|W}(y,u|w)dudydF_{W}(w)\right] \\
& =\frac{1}{\Pr [D_{\delta }=0]}\left[ F_{Y(0)}(\tilde{y})-\int_{-\infty }^{%
\tilde{y}}\int_{\mathcal{W}}\int_{0}^{P_{\delta }\left( w\right)
}f_{Y(0),U_{D}|W}(y,u|w)dudF_{W}(w)dy\right] ,
\end{align*}%
where the orders of integrations can be switched because the integrands are
non-negative. Therefore, 
\begin{equation}
\Pr [D_{\delta }=0]f_{Y(0)|D_{\delta }}(\tilde{y}|0)=f_{Y(0)}(\tilde{y}%
)-\int_{\mathcal{W}}\int_{0}^{P_{\delta }\left( w\right) }f_{Y(0),U_{D}|W}(%
\tilde{y},u|w)dudF_{W}(w).  \label{f_y_d_0_x}
\end{equation}%
Using Assumptions \ref{Assumption_primary}(\ref{uniformity}), \ref%
{Assumption_regularity}(\ref{f_y_u_x}), and \ref{Assumption_regularity}(\ref%
{p_delta}), we have 
\begin{eqnarray}
\frac{\partial \left\{ \Pr [D_{\delta }=0]f_{Y(0)|D_{\delta }}(\tilde{y}%
|0)\right\} }{\partial \delta } &=&-\int_{\mathcal{W}}f_{Y(0),U_{D}|W}(%
\tilde{y},P_{\delta }\left( w\right) )|w)\frac{\partial P_{\delta }(w)}{%
\partial \delta }dF_{W}(w)  \notag \\
&=&-\int_{\mathcal{W}}f_{Y(0)|U_{D},W}(\tilde{y}|P_{\delta }\left( w\right)
,w)\frac{\partial P_{\delta }(w)}{\partial \delta }dF_{W}(w).
\label{f_y_d_0_1_x}
\end{eqnarray}%
The continuity of $\delta \mapsto \frac{\partial \left\{ \Pr [D_{\delta
}=0]f_{Y(0)|D_{\delta }}(\tilde{y}|0)\right\} }{\partial \delta }$ follows
from the same arguments as those for the continuity of $\delta \mapsto \frac{%
\partial \left\{ \Pr [D_{\delta }=1]f_{Y(1)|D_{\delta }}(\tilde{y}%
|1)\right\} }{\partial \delta }$.

For any $\delta $ in $N_{\varepsilon }$, we have 
\begin{eqnarray*}
F_{Y_{\delta }}(y) &=&\int_{\mathcal{Y}(1)}\mathds{1}\left\{ \tilde{y}\leq
y\right\} \Pr [D_{\delta }=1]f_{Y(1)|D_{\delta }}(\tilde{y}|1)d\tilde{y} \\
&+&\int_{\mathcal{Y}(0)}\mathds{1}\left\{ \tilde{y}\leq y\right\} \Pr
[D_{\delta }=0]f_{Y(0)|D_{\delta }}(\tilde{y}|0)d\tilde{y}.
\end{eqnarray*}%
The continuous differentiability of $\delta \mapsto \Pr [D_{\delta
}=1]f_{Y(1)|D_{\delta }}(\tilde{y}|1)$ and $\delta \mapsto \Pr [D_{\delta
}=0]f_{Y(0)|D_{\delta }}(\tilde{y}|0)$ around $\delta =0$ for each $\tilde{y}
$ allows us to take the first-order Taylor expansion of the terms in the
above expression. Using (\ref{partial_p_f_y1_d1}), we have 
\begin{eqnarray}
&&\Pr [D_{\delta }=1]f_{Y(1)|D_{\delta }}(\tilde{y}|1)  \notag \\
&=&pf_{Y(1)|D}(\tilde{y}|1)+\delta \left[ \int_{\mathcal{W}}f_{Y(1)|U_{D},W}(%
\tilde{y}|P(w),w)\dot{P}\left( w\right) dF_{W}(w)\right] +R(\delta ;\tilde{y}%
,1),  \label{expansion_1_x}
\end{eqnarray}%
where 
\begin{equation}
R(\delta ;\tilde{y},1)=\delta \left[ \frac{\partial \left\{ \Pr [D_{\delta
}=1]f_{Y(1)|D_{\delta }}(\tilde{y}|1)\right\} }{\partial \delta }\bigg\rvert%
_{\delta =\tilde{\delta}_{1}}-\frac{\partial \left\{ \Pr [D_{\delta
}=1]f_{Y(1)|D_{\delta }}(\tilde{y}|1)\right\} }{\partial \delta }\bigg\rvert%
_{\delta =0}\right]  \label{expansion_1_x_rem}
\end{equation}%
and $0\leq \tilde{\delta}_{1}\leq \delta $. The middle point $\tilde{\delta}%
_{1}$ depends on $\delta $. For the case of $d=0$, we use (\ref{f_y_d_0_1_x}%
) to obtain a similar expansion: 
\begin{align}
& \Pr [D_{\delta }=0]f_{Y(0)|D_{\delta }}(\tilde{y}|0)  \notag \\
& =(1-p)f_{Y(0)|D}(\tilde{y}|0)-\delta \cdot \left[ \int_{\mathcal{W}%
}f_{Y(0)|U_{D},W}(\tilde{y}|P\left( w\right) ,w)\dot{P}\left( w\right)
dF_{W}(w)\right] +R(\delta ;\tilde{y},0),  \label{expansion_0_x}
\end{align}%
where 
\begin{equation}
R(\delta ;\tilde{y},0)=\delta \left[ \frac{\partial \left\{ \Pr [D_{\delta
}=0]f_{Y(0)|D_{\delta }}(\tilde{y}|0)\right\} }{\partial \delta }\bigg\rvert%
_{\delta =\tilde{\delta}_{0}}-\frac{\partial \left\{ \Pr [D_{\delta
}=0]f_{Y(0)|D_{\delta }}(\tilde{y}|0)\right\} }{\partial \delta }\bigg\rvert%
_{\delta =0}\right]  \label{expansion_0_x_rem}
\end{equation}%
and $0\leq \tilde{\delta}_{0}\leq \delta $. The middle point $\tilde{\delta}%
_{0}$ depends on $\delta $.

Hence 
\begin{eqnarray}
F_{Y_{\delta }}(y) &=&F_{Y}(y)+\delta \int_{\mathcal{Y}(1)}\int_{\mathcal{W}}%
\mathds{1}\left\{ \tilde{y}\leq y\right\} f_{Y(1)|U_{D},W}(\tilde{y}|P(w),w)%
\dot{P}\left( w\right) dF_{W}(w)d\tilde{y}  \notag  \label{expansion_x_1} \\
&-&\delta \int_{\mathcal{Y}(0)}\int_{\mathcal{W}}\mathds{1}\left\{ \tilde{y}%
\leq y\right\} f_{Y(0)|U_{D},W}(\tilde{y}|P\left( w\right) ,w)\dot{P}\left(
w\right) dF_{W}(w)d\tilde{y}  \notag \\
&+&R_{F}(\delta ;y)
\end{eqnarray}%
where the remainder $R_{F}(\delta ;y)$ is 
\begin{equation}
R_{F}(\delta ;y):=\int_{\mathcal{Y}(1)}\mathds{1}\left\{ \tilde{y}\leq
y\right\} R(\delta ;\tilde{y},1)d\tilde{y}+\int_{\mathcal{Y}(0)}\mathds{1}%
\left\{ \tilde{y}\leq y\right\} R(\delta ;\tilde{y},0)d\tilde{y}.
\label{remainder_x_1}
\end{equation}

The next step is to show that the remainder in (\ref{remainder_x_1}) is $%
o(|\delta |)$ uniformly over $y\in \mathcal{Y=Y}(0)\cup \mathcal{Y}(1)$ as $%
\delta \rightarrow 0$, that is, $\lim_{\delta \rightarrow 0}\sup_{y\in 
\mathcal{Y}}\left\vert \frac{R_{F}(\delta ;y)}{\delta }\right\vert =0.$
Using (\ref{expansion_1_x_rem}) and (\ref{expansion_0_x_rem}), we get 
\begin{align*}
& \sup_{y\in \mathcal{Y}}\left\vert \frac{R_{F}(\delta ;y)}{\delta }%
\right\vert \\
& \leq \int_{\mathcal{Y}(1)}\left\vert \frac{\partial \left\{ \Pr [D_{\delta
}=1]f_{Y(1)|D_{\delta }}(\tilde{y}|1)\right\} }{\partial \delta }\bigg\rvert%
_{\delta =\tilde{\delta}_{1}}-\frac{\partial \left\{ \Pr [D_{\delta
}=1]f_{Y(1)|D_{\delta }}(\tilde{y}|1)\right\} }{\partial \delta }\bigg\rvert%
_{\delta =0}\right\vert d\tilde{y} \\
& +\int_{\mathcal{Y}(0)}\left\vert \frac{\partial \left\{ \Pr [D_{\delta
}=0]f_{Y(0)|D_{\delta }}(\tilde{y}|0)\right\} }{\partial \delta }\bigg\rvert%
_{\delta =\tilde{\delta}_{0}}-\frac{\partial \left\{ \Pr [D_{\delta
}=0]f_{Y(0)|D_{\delta }}(\tilde{y}|0)\right\} }{\partial \delta }\bigg\rvert%
_{\delta =0}\right\vert d\tilde{y}.
\end{align*}

Note that Assumption \ref{Assumption_regularity}(\ref{p_delta}.iii) implies
that 
\begin{equation*}
\sup_{\delta \in N_{\varepsilon }}\left\vert \frac{\partial \Pr [D_{\delta
}=d]}{\partial \delta }\right\vert <\infty \text{ for }d=0,1.
\end{equation*}%
This and Assumption \ref{Assumption_domination} allow us to take the limit $%
\delta \rightarrow 0$ under the integral signs. Also, we have shown that,
both $\frac{\partial \{\Pr [D_{\delta }=1]f_{Y(1)|D_{\delta }}(\tilde{y}|1)\}%
}{\partial \delta }$ and $\frac{\partial \{\Pr [D_{\delta
}=0]f_{Y(0)|D_{\delta }}(\tilde{y}|0)\}}{\partial \delta }$ are continuous
in $\delta .$ Therefore, $\lim_{\delta \rightarrow 0}\sup_{y\in \mathcal{Y}%
}\left\vert \frac{R_{F}(\delta ;y)}{\delta }\right\vert =0.$ So, uniformly
over $y\in \mathcal{Y}$ as $\delta \rightarrow 0,$ 
\begin{align*}
F_{Y_{\delta }}(y)& =F_{Y}(y)+\delta \int_{\mathcal{Y}(1)}\int_{\mathcal{W}}%
\mathds{1}\left\{ \tilde{y}\leq y\right\} f_{Y(1)|U_{D},W}(\tilde{y}|P(w),w)%
\dot{P}\left( w\right) dF_{W}(w)d\tilde{y} \\
& -\delta \int_{\mathcal{Y}(0)}\int_{\mathcal{W}}\mathds{1}\left\{ \tilde{y}%
\leq y\right\} f_{Y(0)|U_{D},W}(\tilde{y}|P\left( w\right) ,w)\dot{P}\left(
w\right) dF_{W}(w)d\tilde{y}+o(|\delta |) \\
& =F_{Y}(y)+\delta E\left[ \left\{ F_{Y(1)|U_{D},W}(y|P\left( W\right)
,W)-F_{Y(0)|U_{D},W}(y|P\left( W\right) ,W)\right\} \dot{P}\left( W\right) %
\right] +o(|\delta |).
\end{align*}
\end{proof}

\begin{proof}[Proof of Corollary \protect\ref{Corollary_UQTE}]
For each $d=0$ and $1,$ we have%
\begin{eqnarray*}
&&\frac{1}{f_{Y}(y_{\tau })}\int_{\mathcal{W}}E\left[ \mathds{1}\left\{
Y(d)\leq y_{\tau }\right\} |U_{D}=P\left( w\right) ,W=w\right] \mathcal{\dot{%
P}}\left( w\right) dF_{W}\left( w\right) \\
&=&\frac{1}{f_{Y}(y_{\tau })}\int_{\mathcal{W}}E\left[ \mathds{1}\left\{
Y(d)\leq y_{\tau }\right\} |D=d,W=w\right] \mathcal{\dot{P}}\left( w\right)
dF_{W}\left( w\right) \\
&+&\frac{1}{f_{Y}(y_{\tau })}\int_{\mathcal{W}}E\left[ \mathds{1}\left\{
Y(d)\leq y_{\tau }\right\} |U_{D}=P\left( w\right) ,W=w\right] \mathcal{\dot{%
P}}\left( w\right) dF_{W}\left( w\right) \\
&-&\frac{1}{f_{Y}(y_{\tau })}\int_{\mathcal{W}}E\left[ \mathds{1}\left\{
Y(d)\leq y_{\tau }\right\} |D=d,W=w\right] \mathcal{\dot{P}}\left( w\right)
dF_{W}\left( w\right) \\
&=&\frac{1}{f_{Y}(y_{\tau })}\int_{\mathcal{W}}E\left[ \mathds{1}\left\{
Y(d)\leq y_{\tau }\right\} |D=d,W=w\right] dF_{W}\left( w\right) \\
&-&\frac{1}{f_{Y}(y_{\tau })}\int_{\mathcal{W}}E\left[ \mathds{1}\left\{
Y(d)\leq y_{\tau }\right\} |D=d,W=w\right] \left[ 1-\mathcal{\dot{P}}\left(
w\right) \right] dF_{W}\left( w\right) \\
&-&\frac{1}{f_{Y}(y_{\tau })}\int_{\mathcal{W}}\left[ F_{Y\left( d\right)
|D,W}\left( y_{\tau }|d,w\right) -F_{Y\left( d\right) |U_{D},W}\left(
y_{\tau }|P\left( w\right) ,w\right) \right] \mathcal{\dot{P}}\left(
w\right) dF_{W}\left( w\right) \\
:= &&A_{\tau }\left( d\right) -B_{1\tau }\left( d\right) -B_{2\tau }\left(
d\right)
\end{eqnarray*}%
where 
\begin{eqnarray*}
A_{\tau }\left( d\right) &=&\frac{1}{f_{Y}(y_{\tau })}\int_{\mathcal{W}}E%
\left[ \mathds{1}\left\{ Y(d)\leq y_{\tau }\right\} |D=d,W=w\right]
dF_{W}\left( w\right) , \\
B_{1\tau }\left( d\right) &=&\frac{1}{f_{Y}(y_{\tau })}\int_{\mathcal{W}}E%
\left[ \mathds{1}\left\{ Y(d)\leq y_{\tau }\right\} |D=d,W=w\right] \left[ 1-%
\mathcal{\dot{P}}\left( w\right) \right] dF_{W}\left( w\right) , \\
B_{2\tau }\left( d\right) &=&\frac{1}{f_{Y}(y_{\tau })}\int_{\mathcal{W}}%
\left[ F_{Y\left( d\right) |D,W}\left( y_{\tau }|d,w\right) -F_{Y\left(
d\right) |U_{D},W}\left( y_{\tau }|P\left( w\right) ,w\right) \right] 
\mathcal{\dot{P}}\left( w\right) dF_{W}\left( w\right) .
\end{eqnarray*}%
So 
\begin{eqnarray*}
\Pi _{\tau } &=&A_{\tau }\left( 0\right) -B_{1\tau }\left( 0\right)
-B_{2\tau }\left( 0\right) -\left[ A_{\tau }\left( 1\right) -B_{1\tau
}\left( 1\right) -B_{2\tau }\left( 1\right) \right] \\
&=&A_{\tau }-B_{1\tau }-B_{2\tau },
\end{eqnarray*}%
where 
\begin{eqnarray*}
A_{\tau } &=&A_{\tau }\left( 0\right) -A_{\tau }\left( 1\right) \\
&=&\frac{1}{f_{Y}(y_{\tau })}\int_{\mathcal{W}}E\left[ \mathds{1}\left\{
Y\leq y_{\tau }\right\} |D=0,W=w\right] dF_{W}\left( w\right) \\
&-&\frac{1}{f_{Y}(y_{\tau })}\int_{\mathcal{W}}E\left[ \mathds{1}\left\{
Y\leq y_{\tau }\right\} |D=1,W=w\right] dF_{W}\left( w\right) ,
\end{eqnarray*}%
\begin{eqnarray*}
B_{1\tau } &=&B_{1\tau }\left( 0\right) -B_{1\tau }\left( 1\right) \\
&=&\frac{1}{f_{Y}(y_{\tau })}\int_{\mathcal{W}}\left[ F_{Y\left( 0\right)
|D,W}\left( y_{\tau }|0,w\right) -F_{Y\left( 1\right) |D,W}\left( y_{\tau
}|1,w\right) \right] \left[ 1-\mathcal{\dot{P}}\left( w\right) \right]
dF_{W}\left( w\right) \\
&=&\frac{1}{f_{Y}(y_{\tau })}\int_{\mathcal{W}}\left[ F_{Y|D,W}\left(
y_{\tau }|1,w\right) -F_{Y|D,W}\left( y_{\tau }|0,w\right) \right] \left[ 
\mathcal{\dot{P}}\left( w\right) -1\right] dF_{W}\left( w\right) ,
\end{eqnarray*}%
and 
\begin{eqnarray*}
B_{2\tau } &=&B_{2\tau }\left( 0\right) -B_{2\tau }\left( 1\right) \\
&=&\frac{1}{f_{Y}(y_{\tau })}\int_{\mathcal{W}}\left[ F_{Y\left( 0\right)
|D,W}\left( y_{\tau }|0,w\right) -F_{Y\left( 0\right) |U_{D},W}\left(
y_{\tau }|P\left( w\right) ,w\right) \right] \mathcal{\dot{P}}\left(
w\right) dF_{W}\left( w\right) \\
&-&\frac{1}{f_{Y}(y_{\tau })}\int_{\mathcal{W}}\left[ F_{Y\left( 1\right)
|D,W}\left( y_{\tau }|1,w\right) -F_{Y\left( 1\right) |U_{D},W}\left(
y_{\tau }|P\left( w\right) ,w\right) \right] \mathcal{\dot{P}}\left(
w\right) dF_{W}\left( w\right) .
\end{eqnarray*}
\end{proof}

\begin{proof}[Proof of Theorem \protect\ref{uqr_if_param}]
Consider the following difference 
\begin{eqnarray}
\Pi _{\tau }-\hat{\Pi}_{\tau } &=&\frac{1}{\hat{f}_{Y}(\hat{y}_{\tau })}%
\frac{T_{2n}(\hat{y}_{\tau },\hat{m},\hat{\alpha})}{T_{1n}(\hat{\alpha})}-%
\frac{1}{f_{Y}(y_{\tau })}\frac{T_{2}}{T_{1}}  \notag \\
&=&\frac{T_{2n}(\hat{y}_{\tau },\hat{m},\hat{\alpha})f_{Y}(y_{\tau })T_{1}-%
\hat{f}_{Y}(\hat{y}_{\tau })T_{1n}(\hat{\alpha})T_{2}}{\hat{f}_{Y}(\hat{y}%
_{\tau })T_{1n}(\hat{\alpha})f_{Y}(y_{\tau })T_{1}}  \notag \\
&=&\frac{f_{Y}(y_{\tau })T_{1}}{\hat{f}_{Y}(\hat{y}_{\tau })T_{1n}(\hat{%
\alpha})f_{Y}(y_{\tau })T_{1}}\left[ T_{2n}(\hat{y}_{\tau },\hat{m},\hat{%
\alpha})-T_{2}\right]  \notag \\
&-&\frac{\hat{f}_{Y}(\hat{y}_{\tau })T_{1n}(\hat{\alpha})T_{2}-f_{Y}(y_{\tau
})T_{1}T_{2}}{\hat{f}_{Y}(\hat{y}_{\tau })T_{1n}(\hat{\alpha})f_{Y}(y_{\tau
})T_{1}}  \notag \\
&=&\frac{f_{Y}(y_{\tau })T_{1}}{\hat{f}_{Y}(\hat{y}_{\tau })T_{1n}(\hat{%
\alpha})f_{Y}(y_{\tau })T_{1}}\left[ T_{2n}(\hat{y}_{\tau },\hat{m},\hat{%
\alpha})-T_{2}\right]  \notag \\
&-&\frac{T_{1n}(\hat{\alpha})T_{2}}{\hat{f}_{Y}(\hat{y}_{\tau })T_{1n}(\hat{%
\alpha})f_{Y}(y_{\tau })T_{1}}\left[ \hat{f}_{Y}(\hat{y}_{\tau
})-f_{Y}(y_{\tau })\right]  \notag \\
&-&\frac{T_{2}f_{Y}(y_{\tau })}{\hat{f}_{Y}(\hat{y}_{\tau })T_{1n}(\hat{%
\alpha})f_{Y}(y_{\tau })T_{1}}\left[ T_{1n}(\hat{\alpha})-T_{1}\right] .
\label{est_param_pi_decom}
\end{eqnarray}%
We can rearrange (\ref{est_param_pi_decom}) as 
\begin{eqnarray}
\hat{\Pi}_{\tau }-\Pi _{\tau } &=&\frac{T_{2}}{\hat{f}_{Y}(\hat{y}_{\tau
})f_{Y}(y_{\tau })T_{1}}\left[ \hat{f}_{Y}(\hat{y}_{\tau })-f_{Y}(y_{\tau })%
\right]  \notag \\
&+&\frac{T_{2}}{\hat{f}_{Y}(\hat{y}_{\tau })T_{1n}(\hat{\alpha})T_{1}}\left[
T_{1n}(\hat{\alpha})-T_{1}\right]  \notag \\
&-&\frac{1}{\hat{f}_{Y}(\hat{y}_{\tau })T_{1n}(\hat{\alpha})}\left[ T_{2n}(%
\hat{y}_{\tau },\hat{m},\hat{\alpha})-T_{2}\right] .
\label{est_param_pi_decom_2}
\end{eqnarray}%
By appropriately defining the remainders, we can express (\ref%
{est_param_pi_decom_2}) as 
\begin{eqnarray}
\hat{\Pi}_{\tau }-\Pi _{\tau } &=&\frac{T_{2}}{f_{Y}(y_{\tau })^{2}T_{1}}%
\left[ \hat{f}_{Y}(\hat{y}_{\tau })-f_{Y}(y_{\tau })\right]  \notag \\
&+&\frac{T_{2}}{f_{Y}(y_{\tau })T_{1}^{2}}\left[ T_{1n}(\hat{\alpha})-T_{1}%
\right]  \notag \\
&-&\frac{1}{f_{Y}(y_{\tau })T_{1}}\left[ T_{2n}(\hat{y}_{\tau },\hat{m},\hat{%
\alpha})-T_{2}\right] +R_{1}+R_{2}+R_{3},  \label{est_param_pi_decom_3}
\end{eqnarray}%
where 
\begin{eqnarray*}
R_{1} &=&\left[ \frac{T_{2}}{\hat{f}_{Y}(\hat{y}_{\tau })f_{Y}(y_{\tau
})T_{1}}-\frac{T_{2}}{f_{Y}(y_{\tau })^{2}T_{1}}\right] \left[ \hat{f}_{Y}(%
\hat{y}_{\tau })-f_{Y}(y_{\tau })\right] , \\
R_{2} &=&\left[ \frac{T_{2}}{\hat{f}_{Y}(\hat{y}_{\tau })T_{1n}(\hat{\alpha}%
)T_{1}}-\frac{T_{2}}{f_{Y}(y_{\tau })T_{1}^{2}}\right] \left[ T_{1n}(\hat{%
\alpha})-T_{1}\right] , \\
R_{3} &=&\left[ \frac{1}{f_{Y}(y_{\tau })T_{1}}-\frac{1}{\hat{f}_{Y}(\hat{y}%
_{\tau })T_{1n}(\hat{\alpha})}\right] \left[ T_{2n}(\hat{y}_{\tau },\hat{m},%
\hat{\alpha})-T_{2}\right] .
\end{eqnarray*}%
Now we are ready to separate the contribution of each stage of the
estimation process. Using equations (\ref{decom_f_y}), (\ref{decom_ps_non}),
and (\ref{decom_derps_non0}), we obtain: 
\begin{eqnarray}
\hat{\Pi}_{\tau }-\Pi _{\tau } &=&\frac{T_{2}}{f_{Y}(y_{\tau })^{2}T_{1}}%
\left[ \mathbb{P}_{n}\psi _{f_{Y}}\left( Y,y_{\tau }\right)
+B_{f_{Y}}(y_{\tau })+f_{Y}^{\prime }(y_{\tau })\mathbb{P}_{n}\psi
_{Q}(Y,y_{\tau })+o_{p}(n^{-1/2}h^{-1/2})\right]  \notag \\
&&+\frac{T_{2}}{f_{Y}(y_{\tau })T_{1}^{2}}\left\{ \mathbb{P}_{n}\psi
_{\partial P}\left( W\right) +\left[ E\frac{\partial ^{2}P(Z,X,\alpha _{0})}{%
\partial z_{1}\partial \alpha _{0}^{\prime }}\right] \mathbb{P}_{n}\psi
_{\alpha _{0}}\left( D,W\right) +o_{p}(n^{-1/2})\right\}  \notag \\
&-&\frac{1}{f_{Y}(y_{\tau })T_{1}}\left[ \mathbb{P}_{n}\psi _{\partial
m_{0}}\left( W,y_{\tau }\right) +\mathbb{P}_{n}\psi _{m_{0}}\left(
Y,W,y_{\tau }\right) +\mathbb{P}_{n}\tilde{\psi}_{Q}(Y,y_{\tau
})+o_{p}(n^{-1/2})\right]  \notag \\
&+&R_{1}+R_{2}+R_{3}  \notag \\
&=&\frac{T_{2}}{f_{Y}(y_{\tau })^{2}T_{1}}\left[ \mathbb{P}_{n}\psi
_{f_{Y}}\left( Y,y_{\tau }\right) +B_{f_{Y}}(y_{\tau })+f_{Y}^{\prime
}(y_{\tau })\mathbb{P}_{n}\psi _{Q}(Y,y_{\tau })\right]  \notag \\
&&+\frac{T_{2}}{f_{Y}(y_{\tau })T_{1}^{2}}\left\{ \mathbb{P}_{n}\psi
_{\partial P}\left( W\right) +\left[ E\frac{\partial ^{2}P(Z,X,\alpha _{0})}{%
\partial z_{1}\partial \alpha _{0}^{\prime }}\right] \mathbb{P}_{n}\psi
_{\alpha _{0}}\left( D,W\right) \right\}  \notag \\
&&-\frac{1}{f_{Y}(y_{\tau })T_{1}}\left[ \mathbb{P}_{n}\psi _{\partial
m_{0}}\left( W,y_{\tau }\right) +\mathbb{P}_{n}\psi _{m_{0}}\left(
Y,W,y_{\tau }\right) +\mathbb{P}_{n}\tilde{\psi}_{Q}(Y,y_{\tau })\right] 
\notag \\
&&+R_{1}+R_{2}+R_{3}+o_{p}(n^{-1/2}h^{-1/2}).  \label{est_param_pi_decom_4}
\end{eqnarray}

Define $R_{\Pi }$ according to (\ref{est_param_pi_decomp_5}) in the theorem.
Then 
\begin{equation*}
R_{\Pi }=R_{1}+R_{2}+R_{3}+o_{p}(n^{-1/2}h^{-1/2}).
\end{equation*}

Finally, we establish the rates for the remainders $R_{1},$ $R_{2},$ and $%
R_{3}$ one by one. The first remainder is 
\begin{eqnarray}
R_{1} &=&\frac{T_{2}}{f_{Y}(y)T_{1}}\left[ \hat{f}_{Y}(\hat{y}_{\tau
})-f_{Y}(y_{\tau })\right] \left[ \frac{1}{\hat{f}_{Y}(\hat{y}_{\tau })}-%
\frac{1}{f_{Y}(y_{\tau })}\right]  \notag  \label{r_1} \\
&=&-\frac{T_{2}}{f_{Y}(y)^{2}T_{1}\hat{f}_{Y}(\hat{y}_{\tau })}\left[ \hat{f}%
_{Y}(\hat{y}_{\tau })-f_{Y}(y_{\tau })\right] ^{2}  \notag \\
&=&O_{p}\left( |\hat{f}_{Y}(\hat{y}_{\tau })-f_{Y}(y_{\tau })|^{2}\right) .
\end{eqnarray}%
The second remainder is 
\begin{eqnarray}
R_{2} &=&\frac{T_{2}}{T_{1}}\left[ T_{1n}(\hat{\alpha})-T_{1}\right] \left[ 
\frac{1}{\hat{f}_{Y}(\hat{y}_{\tau })T_{1n}(\hat{\alpha})}-\frac{1}{%
f_{Y}(y_{\tau })T_{1}}\right]  \notag \\
&=&\frac{T_{2}}{T_{1}}\left[ T_{1n}(\hat{\alpha})-T_{1}\right] \left[ \frac{%
f_{Y}(y_{\tau })T_{1}-\hat{f}_{Y}(\hat{y}_{\tau })T_{1n}(\hat{\alpha})}{\hat{%
f}_{Y}(\hat{y}_{\tau })T_{1n}(\hat{\alpha})f_{Y}(y_{\tau })T_{1}}\right] 
\notag \\
&=&\frac{T_{2}}{T_{1}}\left[ T_{1n}(\hat{\alpha})-T_{1}\right] \left[ \frac{%
f_{Y}(y_{\tau })(T_{1}-T_{1n}(\hat{\alpha}))-T_{1n}(\hat{\alpha})(\hat{f}%
_{Y}(\hat{y}_{\tau })-f_{Y}(y_{\tau }))}{\hat{f}_{Y}(\hat{y}_{\tau })T_{1n}(%
\hat{\alpha})f_{Y}(y_{\tau })T_{1}}\right]  \notag \\
&=&O_{p}\left( |T_{1n}(\hat{\alpha})-T_{1}|^{2}\right) +O_{p}\left( |T_{1n}(%
\hat{\alpha})-T_{1}||\hat{f}_{Y}(\hat{y}_{\tau })-f_{Y}(y_{\tau })|\right) 
\notag \\
&=&O_{p}(n^{-1})+O_{p}\left( n^{-1/2}|\hat{f}_{Y}(\hat{y}_{\tau
})-f_{Y}(y_{\tau })|\right) .  \label{r_2}
\end{eqnarray}%
The third remainder satisfies 
\begin{eqnarray}
R_{3} &=&O_{p}\left( |T_{2n}(\hat{y}_{\tau },\hat{m},\hat{\alpha}%
)-T_{2}||T_{1n}(\hat{\alpha})-T_{1}|\right)  \notag  \label{r_3} \\
&+&O_{p}\left( |T_{2n}(\hat{y}_{\tau },\hat{m},\hat{\alpha})-T_{2}||\hat{f}%
_{Y}(\hat{y}_{\tau })-f_{Y}(y_{\tau })|\right)  \notag \\
&=&O_{p}(n^{-1})+O_{p}\left( n^{-1/2}|\hat{f}_{Y}(\hat{y}_{\tau
})-f_{Y}(y_{\tau })|\right) .
\end{eqnarray}%
Therefore, 
\begin{eqnarray}
R_{\Pi } &=&R_{1}+R_{2}+R_{3}+o_{p}(n^{-1/2}h^{-1/2})  \notag \\
&=&O_{p}\left( |\hat{f}_{Y}(\hat{y}_{\tau })-f_{Y}(y_{\tau })|^{2}\right)
+O_{p}(n^{-1})+O_{p}\left( n^{-1/2}|\hat{f}_{Y}(\hat{y}_{\tau
})-f_{Y}(y_{\tau })|\right) +o_{p}(n^{-1/2}h^{-1/2}).  \notag
\end{eqnarray}

Finally, Equation (\ref{decom_f_y}) tells us 
\begin{equation}
\hat{f}_{Y}(\hat{y}_{\tau })-{f}_{Y}(y_{\tau
})=O_{p}(n^{-1/2}h^{-1/2})+O_{p}(h^{2}).  \notag  \label{decom_f_y_2}
\end{equation}%
So, 
\begin{eqnarray*}
R_{\Pi }
&=&O_{p}(n^{-1}h^{-1}+h^{4})+O_{p}(n^{-1})+O_{p}(n^{-1}h^{-1/2}+n^{-1/2}h^{2})+o_{p}(n^{-1/2}h^{-1/2})
\\
&=&o_{p}(h^{2})+o_{p}(n^{-1/2}h^{-1/2})=o_{p}(n^{-1/2}h^{-1/2}),
\end{eqnarray*}%
where we have used the rate conditions maintained in Assumption \ref%
{Assumption_rate}.
\end{proof}

\bibliographystyle{aea}
\bibliography{references}

\clearpage\pagebreak

\section*{Online Supplementary Appendix}

\textbf{Title:} Identification and Estimation of Unconditional Policy
Effects of an Endogenous Binary Treatment: An Unconditional MTE Approach 
\newline
\textbf{Authors:} Julian Martinez-Iriarte and Yixiao Sun

\renewcommand{\thesubsection}{\Alph{subsection}} \setcounter{equation}{0}%
\setcounter{subsection}{0} \renewcommand{\theequation}{\textcolor{red}{S.%
\arabic{equation}}}

\pagenumbering{arabic} \renewcommand{\thepage}{S.\arabic{page}}

\subsection{Technical Conditions for Lemmas \protect\ref{lemma_hahn_ridder}
and \protect\ref{param_ave_estimation}}

Recall that $Z_{-1}=(Z_{2},\mathcal{\ldots },Z_{d_{Z}}),$ $W_{-1}=\left(
W_{2},\ldots ,W_{d_{W}}\right) =\left( Z_{-1},X\right) $ for $%
d_{W}=d_{X}+d_{Z},$ and $\mathcal{W}_{-1}$ is the support of $W_{-1}.$

\begin{assumption}
\textbf{Ignorability of errors in estimating the propensity score}

\label{Assumption_ignore_P_error}

\begin{enumerate}[(a)]%

\item[(a)] Conditional on $W_{-1},$ the distribution of $Z_{1}$ is
absolutely continuous with conditional density $%
f_{Z_{1}|W_{-1}}(z_{1}|w_{-1})$ and

\begin{enumerate}
\item[(i)] for each $\left( z_{1}^{\prime },w_{-1}^{\prime }\right) ^{\prime
}\in \mathcal{W}$, $f_{Z_{1}|W_{-1}}(z_{1}|w_{-1})$ is continuously
differentiable with respect to $z_{1}.$

\item[(ii)] for each $w_{-1}\in \mathcal{W}_{-1}$, $%
f_{Z_{1}|W_{-1}}(z_{1}|w_{-1})=0$ for any $z_{1}$ on the boundary of $%
\mathcal{Z}_{1}\left( w_{-1}\right) ,$ the support of $Z_{1}$ conditional on 
$W_{-1}=w_{-1}.$
\end{enumerate}

\item[(b)] $m(y_{\tau },\tilde{w})$ is continuously differentiable with
respect to $z_{1}$ for all orders, and for a neighborhood $\Theta _{0}$ of $%
\theta _{0},$ the following holds:%
\begin{align*}
& (i)\text{ }E\left[ \sup_{\theta \in \Theta _{0}}\left\vert \frac{\partial 
}{\partial \alpha _{\theta }}\frac{\partial m_{0}(y_{\tau },\tilde{W}\left(
\alpha _{\theta }\right) )}{\partial z_{1}}\right\vert \right] <\infty ,\; \\
& (ii)\text{ }E\sup_{\theta \in \Theta _{0}}\left\vert \frac{\partial }{%
\partial \alpha _{\theta }}E\left[ \left. \frac{\partial \log
f_{Z_{1}|W_{-1}}\left( Z_{1}|W_{-1}\right) }{\partial z_{1}}\right\vert 
\tilde{W}\left( \alpha _{\theta }\right) \right] \right\vert <\infty , \\
& (iii)\text{ }E\sup_{\theta \in \Theta _{0}}\left\vert \frac{\partial
m_{0}(y_{\tau },\tilde{W}\left( \alpha _{\theta }\right) )}{\partial \alpha
_{\theta }}E\left[ \left. \frac{\partial \log f_{Z_{1}|W_{-1}}\left(
Z_{1}|W_{-1}\right) }{\partial z_{1}}\right\vert \tilde{W}\left( \alpha
_{\theta }\right) \right] \right\vert <\infty .
\end{align*}%
\end{enumerate}%
\end{assumption}

\bigskip

\begin{assumption}
\textbf{Expansion of }$T_{2n}(\hat{y}_{\tau },\hat{m},\hat{\alpha})$%
\footnote{%
The assumption of the lemma is adapted from \cite{newey1994} and is not
necessarily the weakest possible.}

\label{Assumption_T2}

\begin{enumerate}[(a)]%

\item[(a)] (i) $W$ is absolutely continuous with density $f_{W}(\cdot )$
supported on $[w_{1l},w_{1u}]\times \cdots \times \lbrack
w_{d_{W}l},w_{d_{W}u}];\footnote{%
The density (Radon-Nikodym derivative) is with respect to the product of the
Lebesgue measure on $\mathbb{R}$ and another measure, which could be the
Lebesgue measure, counting measure, or a product of Lebesgue and counting
measures. The counting measures may assign a nonzero measure to singletons,
allowing for discrete variables.}$ (ii) for $w=\left( w_{1},\mathcal{\ldots }%
,w_{d_{W}}\right) ,$ $f_{W}(w)$ is bounded below by $C\times \Pi
_{j=1}^{d_{W}}(w_{j}-w_{jl})^{\varkappa }(w_{ju}-w_{j})^{\varkappa }$ for
some $C>0,\varkappa >0;$ (iii) $\int_{\mathcal{W}}\sup_{\theta \in \Theta
_{0}}\left\vert \frac{\partial }{\partial \theta }f_{W}(w;\theta
)\right\vert dw<\infty ;$ (iv) $E\{[\frac{\partial }{\partial z_{1}}\log
f_{Z_{1}|W_{-1}}(z_{1}|w_{-1};\theta )]^{2}\}<\infty $ for $\theta \in
\Theta _{0};$

\item[(b)] there is a constant $C$ such that $\left\vert \partial ^{\ell
}m_{0}(y_{\tau },\tilde{w})/\partial z_{1}^{\ell }\right\vert \leq C^{\ell }$
for all $\ell \in \mathbb{N};$

\item[(c)] the number of series terms, $J,$ satisfies $J(n)=O(n^{\varpi })$
for some $\varpi >0$, and $J^{7+2\varkappa }=O(n);$

\item[(d)] for each $w_{-1}\in \mathcal{W}_{-1}$ and $\theta \in \Theta
_{0}, $ the map $(y,z_{1})\mapsto m_{\theta }(y,\tilde{w})$ is
differentiable, and 
\begin{eqnarray*}
&(i)\text{ }E\left[ \sup_{\theta \in \Theta _{0}}\left\vert \frac{\partial }{%
\partial \theta }\frac{\partial m_{\theta }(y_{\tau },\tilde{W})}{\partial
z_{1}}\right\vert \right] <&\infty , \\
&(ii)\text{ }E\left[ \sup_{y\in \mathcal{Y}}\left\vert \frac{\partial
^{2}m_{0}(y,\tilde{W})}{\partial y\partial z_{1}}\right\vert \right]
<&\infty ;
\end{eqnarray*}

\item[(e)] the following stochastic equicontinuity condition holds:%
\begin{equation*}
\sqrt{n}\left( \mathbb{P}_{n}-\mathbb{P}\right) \left[ \frac{\partial \hat{m}%
(y_{\tau },\tilde{W}_{i})}{\partial z_{1}}-\frac{\partial m_{0}(y_{\tau },%
\tilde{W}_{i})}{\partial z_{1}}\right] =o_{p}(1).
\end{equation*}

\item[(f)] Assumption \ref{Assumption_ignore_P_error} holds.

\end{enumerate}%
\end{assumption}

\subsection{Proof of Lemmas and Proposition \protect\ref%
{Prop_MTE_identification}}

\begin{proof}[Proof of Lemma \protect\ref{Lemma_SZ}]
By definition,%
\begin{equation*}
P_{\delta }\left( w\right) =\Pr \left[ D_{\delta }=1|Z=z,X=x\right]
=F_{V|W}\left( \mu (\mathcal{G}(w,\delta ),x)|w\right) .
\end{equation*}%
So 
\begin{equation*}
\frac{\partial P_{\delta }\left( w\right) }{\partial \delta }=f_{V|W}\left(
\mu (\mathcal{G}(w,\delta ),x)|w\right) \mu _{z_{1}}^{\prime }\left( 
\mathcal{G}(w,\delta ),x\right) g\left( w\right) \frac{\partial s\left(
\delta \right) }{\partial \delta }
\end{equation*}%
and 
\begin{equation*}
\dot{P}\left( w\right) =\left. \frac{\partial P_{\delta }(w)}{\partial
\delta }\right\vert _{\delta =0}=f_{V|W}\left( \mu (w)|w\right) \mu
_{z_{1}}^{\prime }\left( w\right) g\left( w\right) \left. \frac{\partial
s\left( \delta \right) }{\partial \delta }\right\vert _{\delta =0}.
\end{equation*}%
Hence,%
\begin{equation*}
\mathcal{\dot{P}}\left( w\right) =\frac{\dot{P}\left( w\right) }{E\left[ 
\dot{P}\left( W\right) \right] }=\frac{f_{V|W}\left( \mu (w)|w\right) \mu
_{z_{1}}^{\prime }\left( w\right) g\left( w\right) }{E\left[ f_{V|W}\left(
\mu (W)|W\right) \mu _{z_{1}}^{\prime }\left( W\right) g\left( W\right) %
\right] }.
\end{equation*}
\end{proof}

\begin{proof}[Proof of Proposition \protect\ref{Prop_MTE_identification}]
The first part of the Proposition, the identification of $\widetilde{\mathrm{%
MTE}}_{\tau }\left( u,w_{-1}\right) ,$ can be obtained by a slight change in
the proof given in the seminal paper \cite{Heckman2001} which focuses on the
case of the mean. The general case is treated in Theorem 1 in \cite%
{Carneiro2009}. See also the working paper \cite{sun2021}, which gives a
complete and self-contained proof.

Now we turn to the identification of the weighting function $\mathcal{\dot{P}%
}\left( w\right) $ given in Corollary \ref{Corollary_UQTE_w}. Under
Assumption \ref{Assumption_heckman}(\ref{exogeneity}), the propensity score
becomes%
\begin{equation*}
P(w)=\Pr \left[ V\leq \mu \left( W\right) |W=w\right] =F_{V|W}(\mu
(z,x)|z,x)=F_{V|W_{-1}}(\mu (z,x)|w_{-1}).
\end{equation*}%
Therefore, 
\begin{equation*}
\frac{\partial P(w)}{\partial z_{1}}=f_{V|W_{-1}}(\mu (z,x)|w_{-1})\mu
_{z_{1}}^{\prime }\left( z,x\right) =f_{V|W}(\mu (w)|w)\mu _{z_{1}}^{\prime
}\left( w\right) .
\end{equation*}%
It is now clear that $\mathcal{\dot{P}}\left( w\right) $ can be represented
using $\frac{\partial P(w)}{\partial z_{1}}$ and $g\left( w\right) ,$ as
given in (\ref{id_prop_weight}).
\end{proof}

\begin{lemma}
\label{quantile_an} If $f_{Y}\left( y_{\tau }\right) >0,$ then $\hat{y}%
_{\tau }-y_{\tau }=\mathbb{P}_{n}\psi _{Q}(Y,y_{\tau })+o_{p}(n^{-1/2})$ 
where 
\begin{equation*}
\psi _{Q}(Y,y_{\tau }):=\frac{\tau -\mathds{1}\left\{ Y\leq y_{\tau
}\right\} }{f_{Y}(y_{\tau })}.
\end{equation*}
\end{lemma}

\begin{proof}[Proof of Lemma \protect\ref{quantile_an}]
This is Corollary 21.5 of \cite{vaart_1998} and a proof can be found there.
\end{proof}

\begin{lemma}
\label{two_step_density} Let Assumptions \ref{Assumption_Kernel} and \ref%
{Assumption_rate} hold. Then, for any fixed $y\in \mathcal{Y},$ 
\begin{equation*}
\hat{f}_{Y}(y)-f_{Y}(y)=\mathbb{P}_{n}\psi _{f_{Y}}\left( Y,y\right)
+B_{f_{Y}}(y)+o_{p}(h^{2}),
\end{equation*}%
where 
\begin{equation*}
\psi _{f_{Y}}\left( Y,y\right) :=K_{h}\left( Y-y\right) -E\left[ K_{h}\left(
Y-y\right) \right] =O_{p}(n^{-1/2}h^{-1/2})
\end{equation*}%
and 
\begin{equation*}
B_{f_{Y}}(y)=\frac{1}{2}h^{2}f_{Y}^{\prime \prime 2}(y)\int_{-\infty
}^{\infty }u^{2}K(u)du.
\end{equation*}%
Furthermore, for the quantile estimator $\hat{y}_{\tau }$ of $y_{\tau }$
that satisfies Lemma \ref{quantile_an}, we have 
\begin{equation*}
\hat{f}_{Y}(\hat{y}_{\tau })-\hat{f}_{Y}(y_{\tau })={f}_{Y}(\hat{y}_{\tau })-%
{f}_{Y}(y_{\tau })+R_{f_{Y}}=f_{Y}^{\prime }(y_{\tau })\mathbb{P}_{n}\psi
_{Q}(Y,y_{\tau })+R_{f_{Y}},
\end{equation*}%
where $R_{f_{Y}}=o_{p}(n^{-1/2}h^{-1/2})$ and so 
\begin{equation*}
\hat{f}_{Y}(\hat{y}_{\tau })-f_{Y}(y_{\tau })=\mathbb{P}_{n}\psi
_{f_{Y}}\left( Y,y_{\tau }\right) +B_{f_{Y}}(y_{\tau })+f_{Y}^{\prime
}(y_{\tau })\mathbb{P}_{n}\psi _{Q}(Y,y_{\tau })+o_{p}\left(
n^{-1/2}h^{-1/2}\right) .
\end{equation*}
\end{lemma}

\begin{proof}[Proof of Lemma \protect\ref{two_step_density}]
We have 
\begin{align*}
\hat{f}_{Y}(y)-f_{Y}(y)& =\frac{1}{nh}\sum_{i=1}^{n}K\left( \frac{Y_{i}-y}{h}%
\right) -f_{Y}(y) \\
& =\frac{1}{n}\sum_{i=1}^{n}\frac{1}{h}K\left( \frac{Y_{i}-y}{h}\right) -E%
\hat{f}_{Y}(y)+E\hat{f}_{Y}(y)-f_{Y}(y) \\
& =\frac{1}{n}\sum_{i=1}^{n}\frac{1}{h}K\left( \frac{Y_{i}-y}{h}\right) -E%
\hat{f}_{Y}(y)+B_{f_{Y}}(y)+o_{p}(h^{2}),
\end{align*}%
where 
\begin{equation*}
B_{f_{Y}}(y)=\frac{1}{2}h^{2}f_{Y}^{\prime \prime 2}(y)\int_{-\infty
}^{\infty }u^{2}K(u)du.
\end{equation*}%
We write this concisely as 
\begin{equation*}
\hat{f}_{Y}(y)-f_{Y}(y)=\frac{1}{n}\sum_{i=1}^{n}\psi
_{f_{Y}}(Y_{i},y)+B_{f_{Y}}(y)+o_{p}(h^{2}),
\end{equation*}%
where 
\begin{equation*}
\psi _{f_{Y}}(Y_{i},y):=\frac{1}{h}K\left( \frac{Y_{i}-y}{h}\right) -E\frac{1%
}{h}K\left( \frac{Y_{i}-y}{h}\right) =O_{p}(n^{-1/2}h^{-1/2}).
\end{equation*}%
This completes the proof of the first result.

We now prove the second result of the lemma. Since $K(u)$ is twice
continuously differentiable, we use a Taylor expansion to obtain 
\begin{equation}
\hat{f}_{Y}(\hat{y}_{\tau })-\hat{f}_{Y}(y_{\tau })=\hat{f}_{Y}^{\prime
}(y_{\tau })(\hat{y}_{\tau }-y_{\tau })+\frac{1}{2}\hat{f}^{\prime \prime
}\left( \tilde{y}_{\tau }\right) (\hat{y}_{\tau }-y_{\tau })^{2}
\label{diff_f_hat}
\end{equation}
for some $\tilde{y}_{\tau }$ between $\hat{y}_{\tau }$ and $y_{\tau }.$ The
first and second derivatives are 
\begin{equation*}
\hat{f}_{Y}^{\prime }(y)=-\frac{1}{nh^{2}}\sum_{i=1}^{n}K^{\prime }\left( 
\frac{Y_{i}-y}{h}\right) ,\text{ }\hat{f}_{Y}^{\prime \prime }(y)=\frac{1}{
nh^{3}}\sum_{i=1}^{n}K^{\prime \prime }\left( \frac{Y_{i}-y}{h}\right) .
\end{equation*}
To find the order of $\hat{f}_{Y}^{\prime \prime }(y),$ we calculate its
mean and variance. We have 
\begin{eqnarray*}
E\left[ \hat{f}_{Y}^{\prime \prime }(y)\right] &=&\frac{1}{nh^{3}}
\sum_{i=1}^{n}K^{\prime \prime }\left( \frac{Y_{i}-y}{h}\right) =O\left( 
\frac{1}{h^{2}}\right) , \\
var\left[ \hat{f}_{Y}^{\prime \prime }(y)\right] &\leq &\frac{n}{\left(
nh^{3}\right) ^{2}}E\left[ K^{\prime \prime }\left( \frac{Y_{i}-y}{h}\right) %
\right] ^{2}=O\left( \frac{1}{nh^{5}}\right) .
\end{eqnarray*}
Therefore, when $nh^{3}\rightarrow \infty ,$ 
\begin{equation*}
\hat{f}_{Y}^{\prime \prime }(y)=O_{p}\left( h^{-2}\right) ,
\end{equation*}
for any $y.$ That is, for any $\epsilon >0,$ there exists an $M$ $>0$ such
that 
\begin{equation*}
\Pr \left[ h^{2}\left\vert \hat{f}_{Y}^{\prime \prime }(y_{\tau
})\right\vert >\frac{M}{2}\right] <\frac{\epsilon }{2}
\end{equation*}
when $n$ is large enough.

Suppose we choose $M$ so large that we also have 
\begin{equation*}
\Pr \left[ \sqrt{n}\left\vert \tilde{y}_{\tau }-y_{\tau }\right\vert >M%
\right] <\frac{\epsilon }{2}
\end{equation*}
when $n$ is large enough. Then, when $n$ is large enough, 
\begin{eqnarray}
\Pr \left[ h^{2}\hat{f}_{Y}^{\prime \prime }(\tilde{y}_{\tau })>M\right]
&\leq &\Pr \left[ h^{2}\left\vert \hat{f}_{Y}^{\prime \prime }(\tilde{y}
_{\tau })-\hat{f}_{Y}^{\prime \prime }(y_{\tau })\right\vert >\frac{M}{2} %
\right] +\Pr \left[ h^{2}\left\vert \hat{f}_{Y}^{\prime \prime }(y_{\tau
})\right\vert >\frac{M}{2}\right]  \notag \\
&\leq &\Pr \left[ h^{2}\left[ \hat{f}_{Y}^{\prime \prime }(\tilde{y}_{\tau
})-\hat{f}_{Y}^{\prime \prime }(y_{\tau })\right] >\frac{M}{2}\right] +\frac{
\epsilon }{2}  \notag \\
&\leq &\Pr \left[ h^{2}\left[ \hat{f}_{Y}^{\prime \prime }(\tilde{y}_{\tau
})-\hat{f}_{Y}^{\prime \prime }(y_{\tau })\right] >\frac{M}{2},\sqrt{n}
\left\vert \tilde{y}_{\tau }-y_{\tau }\right\vert <M\right] +\epsilon .
\label{order_f_hat_second_der}
\end{eqnarray}
We split $\sqrt{n}\left\vert \tilde{y}_{\tau }-y_{\tau }\right\vert <M$ into
two cases. When $\sqrt{n}\left\vert \tilde{y}_{\tau }-y_{\tau }\right\vert < 
\sqrt{h},$ we have 
\begin{eqnarray*}
h^{2}\left\vert \hat{f}_{Y}^{\prime \prime }(\tilde{y}_{\tau })-\hat{f}
_{Y}^{\prime \prime }(y_{\tau })\right\vert &\leq &\frac{1}{nh}
\sum_{i=1}^{n}\left\vert K^{\prime \prime }\left( \frac{Y_{i}-y_{\tau }}{h}
\right) -K^{\prime \prime }\left( \frac{Y_{i}-\tilde{y}_{\tau }}{h}\right)
\right\vert \\
&\leq &\frac{1}{nh}\cdot n\cdot L_{K}\frac{1}{h}\left\vert y_{\tau }-\tilde{%
y }_{\tau }\right\vert \leq L_{K}\cdot \frac{1}{h^{2}}\frac{\sqrt{h}}{\sqrt{n%
}} =L_{K}\cdot \frac{1}{\sqrt{nh^{3}}}
\end{eqnarray*}
by the Lipschitz continuity of $K^{\prime \prime }\left( \cdot \right) $
with Lipschitz constant $L_{K}.$ When $\sqrt{h}\leq \sqrt{n}\left\vert 
\tilde{y}_{\tau }-y_{\tau }\right\vert <M,$ we have 
\begin{equation*}
\left\vert \frac{Y_{i}-y_{\tau }}{h}-\frac{Y_{i}-\tilde{y}_{\tau }}{h}
\right\vert =\frac{\sqrt{n}\left\vert \tilde{y}_{\tau }-y_{\tau }\right\vert 
}{h}>\frac{1}{\sqrt{h}}\rightarrow \infty .
\end{equation*}
Using the second condition on $K^{\prime \prime }\left( \cdot \right) ,$ we
have, for $\sqrt{h}\leq \sqrt{n}\left\vert \tilde{y}_{\tau }-y_{\tau
}\right\vert <M,$ 
\begin{equation*}
\left\vert K^{\prime \prime }\left( \frac{Y_{i}-y_{\tau }}{h}\right)
-K^{\prime \prime }\left( \frac{Y_{i}-\tilde{y}_{\tau }}{h}\right)
\right\vert \leq C_{2}\frac{\left( \tilde{y}_{\tau }-y_{\tau }\right) ^{2}}{
h^{2}}\leq C_{2}\frac{M^{2}}{nh^{2}},
\end{equation*}
and 
\begin{eqnarray*}
h^{2}\left\vert \hat{f}_{Y}^{\prime \prime }(\tilde{y}_{\tau })-\hat{f}
_{Y}^{\prime \prime }(y_{\tau })\right\vert &\leq &\frac{1}{nh}
\sum_{i=1}^{n}\left\vert K^{\prime \prime }\left( \frac{Y_{i}-y_{\tau }}{h}
\right) -K^{\prime \prime }\left( \frac{Y_{i}-\tilde{y}_{\tau }}{h}\right)
\right\vert \\
&\leq &C_{2}\frac{M^{2}}{nh^{3}}=O_{p}\left( \frac{1}{\sqrt{nh^{3}}}\right) .
\end{eqnarray*}
Hence, in both cases, $h^{2}\left\vert \hat{f}_{Y}^{\prime \prime }(\tilde{y}
_{\tau })-\hat{f}_{Y}^{\prime \prime }(y_{\tau })\right\vert =O_{p}\left(
n^{-1/2}h^{-3/2}\right) .$ As a result, 
\begin{equation*}
\Pr \left[ h^{2}\left[ \hat{f}_{Y}^{\prime \prime }(\tilde{y}_{\tau })-\hat{%
f }_{Y}^{\prime \prime }(y_{\tau })\right] >\frac{M}{2},\sqrt{n}\left\vert 
\tilde{y}_{\tau }-y_{\tau }\right\vert <M\right] \rightarrow 0.
\end{equation*}
Combining this with (\ref{order_f_hat_second_der}), we obtain 
\begin{equation*}
h^{2}\hat{f}_{Y}^{\prime \prime }(\tilde{y}_{\tau })=O_{p}\left( 1\right) 
\text{ and }\hat{f}_{Y}^{\prime \prime }(\tilde{y}_{\tau })(\hat{y}_{\tau
}-y_{\tau })^{2}=O_{p}\left( n^{-1}h^{-2}\right) .
\end{equation*}
In view of (\ref{diff_f_hat}), we then have 
\begin{equation*}
\hat{f}_{Y}(\hat{y}_{\tau })-\hat{f}_{Y}(y_{\tau })=\hat{f}_{Y}^{\prime
}(y_{\tau })(\hat{y}_{\tau }-y_{\tau })+O_{p}(n^{-1}h^{-2}).
\end{equation*}

Now, using Lemma \ref{quantile_an}, we can write 
\begin{eqnarray*}
\hat{f}_{Y}(\hat{y}_{\tau })-\hat{f}_{Y}(y_{\tau }) &=&f_{Y}^{\prime
}(y_{\tau })(\hat{y}_{\tau }-y_{\tau })+\left[ \hat{f}_{Y}^{\prime }(y_{\tau
})-f_{Y}^{\prime }(y_{\tau })\right] (\hat{y}_{\tau }-y_{\tau
})+O_{p}(n^{-1}h^{-2}) \\
&=&f_{Y}^{\prime }(y_{\tau })\frac{1}{n}\sum_{i=1}^{n}\psi
_{Q}(Y_{i},y_{\tau })+R_{f_{Y}},
\end{eqnarray*}%
where 
\begin{equation*}
R_{f_{Y}}:=\left[ \hat{f}_{Y}^{\prime }(y_{\tau })-f_{Y}^{\prime }(y_{\tau })%
\right] \left[ \hat{y}_{\tau }-y_{\tau }\right]
+o_{p}(n^{-1/2})+O_{p}(n^{-1}h^{-2})
\end{equation*}%
and the $o_{p}(n^{-1/2})$ term is the error of the linear asymptotic
representation of $\hat{y}_{\tau }-y_{\tau }.$

In order to obtain the order of $R_{f_{Y}}$, we use the following results: 
\begin{eqnarray*}
\hat{f}_{Y}^{\prime }(y_{\tau }) &=&f_{Y}^{\prime }(y_{\tau })+O_{p}\left( 
\frac{1}{\sqrt{nh^{3}}}+h^{2}\right) , \\
\hat{y}_{\tau } &=&y_{\tau }+O_{p}\left( \frac{1}{\sqrt{n}}\right) .
\end{eqnarray*}%
The rate of convergence of $\hat{f}_{Y}^{\prime }(y_{\tau })$ can be found
on page 56 of \citesupp{Ullah1999}. Therefore, 
\begin{eqnarray*}
R_{f_{Y}} &=&o_{p}(n^{-1/2})+O_{p}\left( \frac{1}{\sqrt{nh^{3}}}%
+h^{2}\right) O_{p}\left( \frac{1}{\sqrt{n}}\right) +O_{p}(n^{-1}h^{-2}). \\
&=&o_{p}(n^{-1/2})+O_{p}\left( n^{-1}h^{-3/2}\right) +O_{p}\left(
n^{-1/2}h^{2}\right) +O_{p}(n^{-1}h^{-2}). \\
&=&o_{p}(n^{-1/2})+O_{p}\left( n^{-1}h^{-3/2}\right) +O_{p}(n^{-1}h^{-2}),
\end{eqnarray*}%
because $h\downarrow 0$ under Assumption \ref{Assumption_rate} and so $%
O_{p}\left( n^{-1/2}h^{2}\right) =o_{p}(n^{-1/2})$.

We need to show that $\sqrt{nh}R_{f_{Y}}=o_{p}(1)$. We do this term by term.
First, 
\begin{equation*}
\sqrt{nh}\times o_{p}(n^{-1/2})=o_{p}(h^{1/2})=o_{p}(1)
\end{equation*}%
because $h\downarrow 0.$ Second, 
\begin{equation*}
\sqrt{nh}\times O_{p}\left( n^{-1}h^{-3/2}\right) =O_{p}\left(
n^{-1/2}h^{-1}\right) =o_{p}(1)
\end{equation*}%
as long as $nh^{2}\uparrow \infty ,$ which is guaranteed by Assumption \ref%
{Assumption_rate}, since it is implied by $nh^{3}\uparrow \infty $. Finally, 
\begin{equation*}
\sqrt{nh}\times O_{p}(n^{-1}h^{-2})=O_{p}(n^{-1/2}h^{-3/2})=o_{p}(1)
\end{equation*}%
since by Assumption \ref{Assumption_rate} $nh^{3}\uparrow \infty $.
Therefore, $\sqrt{nh}R_{f_{Y}}=o_{p}(1)$.

The last result in the lemma follows from the first two results in the lemma
and the rate condition $h^{2}=O(n^{-1/2}h^{-1/2})$ under Assumption \ref%
{Assumption_rate}.
\end{proof}

\begin{lemma}
\label{ps_estimation_param} Suppose that

\begin{enumerate}
\item[(a)] $\hat{\alpha}$ admits the representation $\hat{\alpha}-\alpha
_{0}=\mathbb{P}_{n}\psi _{\alpha _{0}}(D,W)+o_{p}(n^{-1/2}),$ 
where $\psi _{\alpha _{0}}(D,W)$ is a mean-zero $d_{\alpha }\times 1$ random
vector with $E(\Vert \psi _{\alpha _{0}}(D,W)\Vert ^{2})<\infty $, and $%
\Vert \cdot \Vert $ denotes the Euclidean norm;

\item[(b)] the variance of $\frac{\partial P(Z,X,\alpha _{0})}{\partial z_{1}%
}:=\frac{\partial P(z,x,\alpha _{0})}{\partial z_{1}}\big |_{(z,x)=(Z,X)}$
is finite;

\item[(c)] the $d_{\alpha }\times 1$ derivative vector $\frac{\partial
^{2}P(z,x,\alpha )}{\partial \alpha \partial z_{1}}$ exists for all $z$ and $%
x$ and for $\alpha $ in an open neighborhood $\mathcal{A}_{0}$ around $%
\alpha _{0};$

\item[(d)] for $\alpha \in \mathcal{A}_{0}$, the map $\alpha \mapsto E\left[ 
\frac{\partial ^{2}P(Z,X,\alpha )}{\partial \alpha \partial z_{1}}\right] $
is continuous, and a uniform law of large numbers holds: $\sup_{\alpha \in 
\mathcal{A}_{0}}\left\Vert \left( \mathbb{P}_{n}-\mathbb{P}\right) \frac{%
\partial ^{2}P(Z_{i},X_{i},\alpha )}{\partial \alpha \partial z_{1}}%
\right\Vert \overset{p}{\rightarrow }0$. 
\end{enumerate}

Then, $T_{1n}(\hat{\alpha})-T_{1}$ has the following stochastic
approximation 
\begin{equation*}
T_{1n}(\hat{\alpha})-T_{1}=\left\{ E\left[ \frac{\partial ^{2}P(Z,X,\alpha
_{0})}{\partial z_{1}\partial \alpha }\right] \right\} ^{\prime }\mathbb{P}%
_{n}\psi _{\alpha _{0}}\left( D,W\right) +\mathbb{P}_{n}\psi _{\partial
P}\left( W\right) +o_{p}(n^{-1/2}),
\end{equation*}%
where 
\begin{equation*}
\psi _{{\partial P}}\left( W\right) :=\frac{\partial P(Z,X,\alpha _{0})}{%
\partial z_{1}}-E\left[ \frac{\partial P(Z,X,\alpha _{0})}{\partial z_{1}}%
\right] .
\end{equation*}
\end{lemma}

\begin{proof}[Proof of Lemma \protect\ref{ps_estimation_param}]
We have the following decomposition: 
\begin{align*}
\frac{1}{n}\sum_{i=1}^{n}\frac{\partial P(Z_{i},X_{i},\hat{\alpha})}{%
\partial z_{1}}-E\left[ \frac{\partial P(Z,X,\alpha _{0})}{\partial z_{1}}%
\right] & =\frac{1}{n}\sum_{i=1}^{n}\frac{\partial P(Z_{i},X_{i},\hat{\alpha}%
)}{\partial z_{1}}-\frac{1}{n}\sum_{i=1}^{n}\frac{\partial
P(Z_{i},X_{i},\alpha _{0})}{\partial z_{1}} \\
& +\frac{1}{n}\sum_{i=1}^{n}\frac{\partial P(Z_{i},X_{i},\alpha _{0})}{%
\partial z_{1}}-E\left[ \frac{\partial P(Z,X,\alpha _{0})}{\partial z_{1}}%
\right] .
\end{align*}

Under Condition (b) of the lemma, we have 
\begin{equation*}
\frac{1}{n}\sum_{i=1}^{n}\frac{\partial P(Z_{i},X_{i},\alpha _{0})}{\partial
z_{1}}-E\left[ \frac{\partial P(Z,X,\alpha _{0})}{\partial z_{1}}\right] =%
\frac{1}{n}\sum_{i=1}^{n}\psi _{{\partial P}}\left( W_{i}\right)
=O_{p}(n^{-1/2}).
\end{equation*}

For the first term, we have, by applying the mean value theorem, 
\begin{equation*}
\frac{1}{n}\sum_{i=1}^{n}\frac{\partial P(Z_{i},X_{i},\hat{\alpha})}{%
\partial z_{1}}-\frac{1}{n}\sum_{i=1}^{n}\frac{\partial P(Z_{i},X_{i},\alpha
_{0})}{\partial z_{1}}=\left( \frac{1}{n}\sum_{i=1}^{n}\frac{\partial
^{2}P(Z_{i},X_{i},\tilde{\alpha})}{\partial \alpha \partial z_{1}}\right)
^{\prime }(\hat{\alpha}-\alpha _{0}),
\end{equation*}%
where $\tilde{\alpha}$ is between $\alpha _{0}$ and $\hat{\alpha}$. Under
Conditions (c) and (d) of the lemma, we have 
\begin{equation}
\frac{1}{n}\sum_{i=1}^{n}\frac{\partial ^{2}P(Z_{i},X_{i},\tilde{\alpha})}{%
\partial \alpha \partial z_{1}}\overset{p}{\rightarrow }E\left[ \frac{%
\partial ^{2}P(Z,X,\alpha _{0})}{\partial \alpha \partial z_{1}}\right] .
\label{unif_lln}
\end{equation}%
Using (\ref{unif_lln}) together with the linear representation of $\hat{%
\alpha}$ in Condition (a), we obtain 
\begin{eqnarray*}
&&\frac{1}{n}\sum_{i=1}^{n}\frac{\partial P(Z_{i},X_{i},\hat{\alpha})}{%
\partial z_{1}}-\frac{1}{n}\sum_{i=1}^{n}\frac{\partial P(Z_{i},X_{i},\alpha
_{0})}{\partial z_{1}} \\
&=&\left\{ E\left[ \frac{\partial ^{2}P(Z,X,\alpha _{0})}{\partial \alpha
\partial z_{1}}\right] ^{\prime }+o_{p}\left( 1\right) \right\} \left\{ 
\mathbb{P}_{n}\psi _{\alpha _{0}}\left( D,W\right) +o_{p}(n^{-1/2})\right\}
\\
&=&E\left[ \frac{\partial ^{2}P(Z,X,\alpha _{0})}{\partial \alpha \partial
z_{1}}\right] ^{\prime }\mathbb{P}_{n}\psi _{\alpha _{0}}\left( D,W\right)
+o_{p}(n^{-1/2}).
\end{eqnarray*}%
We can then write 
\begin{eqnarray*}
T_{1n}(\hat{\alpha})-T_{1} &=&\frac{1}{n}\sum_{i=1}^{n}\frac{\partial
P(Z_{i},X_{i},\hat{\alpha})}{\partial z_{1}}-E\left[ \frac{\partial
P(Z,X,\alpha _{0})}{\partial z_{1}}\right] \\
&=&E\left[ \frac{\partial ^{2}P(Z,X,\alpha _{0})}{\partial \alpha \partial
z_{1}}\right] ^{\prime }\mathbb{P}_{n}\psi _{\alpha _{0}}\left( D,W\right) +%
\mathbb{P}_{n}\psi _{{\partial P}}(W)+o_{p}(n^{-1/2}).
\end{eqnarray*}
\end{proof}

\begin{lemma}
\label{lemma_hahn_ridder} Under Assumption \ref{Assumption_ignore_P_error}
given in the supplementary appendix, 
\begin{equation*}
\left. \frac{\partial }{\partial \theta }E\left[ \frac{\partial m_{\theta
_{0}}(y_{\tau },\tilde{W}\left( \alpha _{\theta }\right) )}{\partial z_{1}}%
\right] \right\vert _{\theta =\theta _{0}}=0.
\end{equation*}
\end{lemma}

\begin{proof}[Proof of Lemma \protect\ref{lemma_hahn_ridder}]
Recall that 
\begin{eqnarray*}
m_{\theta _{0}}\left( y_{\tau },\tilde{w}\left( \alpha _{\theta }\right)
\right) &:&=m_{0}\left( y_{\tau },P(w,\alpha _{\theta }),w_{-1}\right) \\
&=&E\left[ \mathds{1}\left\{ Y\leq y_{\tau }\right\} |P(W,\alpha _{\theta
})=P(w,\alpha _{\theta }),W_{-1}=w_{-1}\right] .
\end{eqnarray*}%
In order to emphasize the dual roles of $\alpha _{\theta }$, we define 
\begin{equation*}
\tilde{m}_{0}(y_{\tau },u,w_{-1};P\left( \cdot ,\alpha _{\theta _{2}}\right)
):=E\left[ \mathds{1}\left\{ Y\leq y_{\tau }\right\} |P(W,\alpha _{\theta
_{2}})=u,W_{-1}=w_{-1}\right] .
\end{equation*}%
Since $y_{\tau }$ is fixed, we regard $\tilde{m}_{0}$ as a function of $%
\left( u,w_{-1}\right) $ that depends on the function $P\left( \cdot ,\alpha
_{\theta _{2}}\right) .$ Then 
\begin{eqnarray*}
\tilde{m}_{0}(y_{\tau },P(w,\alpha _{\theta _{1}}),w_{-1};P\left( \cdot
,\alpha _{\theta _{2}}\right) )|_{\theta _{1}=\theta _{2}=\theta } &=&E\left[
\mathds{1}\left\{ Y\leq y_{\tau }\right\} |P(W,\alpha _{\theta })=P(w,\alpha
_{\theta }),W_{-1}=w_{-1}\right] \\
&=&m_{0}\left( y_{\tau },P(w,\alpha _{\theta }),w_{-1}\right) .
\end{eqnarray*}%
As in \citesupp{hahn2013}, we employ $\tilde{m}_{0}\left( y_{\tau
},u,w_{-1};P\left( \cdot ,\alpha _{\theta _{2}}\right) \right) $ as an
expositional device only.

The functional of interest is 
\begin{eqnarray*}
\mathbb{H}[m_{0}]=: &&\mathbb{H}[m_{0}\left( y_{\tau },P(\cdot ,\alpha
_{\theta }),\cdot \right) ]=\int_{\mathcal{W}}\frac{\partial m_{0}\left(
y_{\tau },P(w,\alpha _{\theta }),w_{-1}\right) }{\partial z_{1}}dF_{W}\left(
w\right) \\
&=&\int_{\mathcal{W}}\frac{\partial \tilde{m}_{0}(y_{\tau },P(w,\alpha
_{\theta _{1}}),w_{-1};P\left( \cdot ,\alpha _{\theta _{2}}\right) )}{%
\partial z_{1}}\bigg |_{\theta _{1}=\theta _{2}=\theta }dF_{W}\left(
w\right) .
\end{eqnarray*}%
Under Assumption \ref{Assumption_ignore_P_error}(b.i), we can exchange $%
\frac{\partial }{\partial \alpha _{\theta }}$ with $\int_{\mathcal{W}}$ and
obtain 
\begin{eqnarray*}
\left. \frac{\partial }{\partial \alpha _{\theta }}\mathbb{H}%
[m_{0}]\right\vert _{\theta =\theta _{0}} &=&\int_{\mathcal{W}}\left. \frac{%
\partial }{\partial z_{1}}\frac{\partial \tilde{m}_{0}(y_{\tau },P(w,\alpha
_{\theta _{1}}),w_{-1};P\left( \cdot ,\alpha _{\theta _{2}}\right) )}{%
\partial \alpha _{\theta _{1}}}\right\vert _{\theta _{1}=\theta _{2}=\theta
_{0}}dF_{W}\left( w\right) \\
&+&\int_{\mathcal{W}}\left. \frac{\partial }{\partial z_{1}}\frac{\partial 
\tilde{m}_{0}(y_{\tau },P(w,\alpha _{\theta _{1}}),w_{-1};P\left( \cdot
,\alpha _{\theta _{2}}\right) )}{\partial \alpha _{\theta _{2}}}\right\vert
_{\theta _{1}=\theta _{2}=\theta _{0}}dF_{W}\left( w\right) \\
&=&\int_{\mathcal{W}}\frac{\partial }{\partial z_{1}}\left[ \tilde{m}%
_{0,\alpha }^{\ast }\left( y_{\tau },P(w,\alpha _{\theta
_{0}}),w_{-1};P\left( \cdot ,\alpha _{\theta _{0}}\right) \right) \right]
dF_{W}\left( w\right) ,
\end{eqnarray*}%
where 
\begin{eqnarray*}
\tilde{m}_{0,\alpha }^{\ast }\left( y_{\tau },P(w,\alpha _{\theta
}),w_{-1};P\left( \cdot ,\alpha _{\theta }\right) \right) &=&\frac{\partial 
\tilde{m}_{0}(y_{\tau },P(w,\alpha _{\theta _{1}}),w_{-1};P\left( \cdot
,\alpha _{\theta }\right) )}{\partial \alpha _{\theta _{1}}}\bigg\vert%
_{\theta _{1}=\theta } \\
&+&\frac{\partial \tilde{m}_{0}(y_{\tau },P(w,\alpha _{\theta
}),w_{-1};P\left( \cdot ,\alpha _{\theta _{2}}\right) )}{\partial \alpha
_{\theta _{2}}}\bigg\vert_{\theta _{2}=\theta }.
\end{eqnarray*}%
Under Assumption \ref{Assumption_ignore_P_error}(a), we can use integration
by parts. Letting $z_{-1}=\left( z_{2},\mathcal{\ldots },z_{d_{Z}}\right) $
and $\left[ z_{1l}\left( w_{-1}\right) ,z_{1u}\left( w_{-1}\right) \right] $
be the support of $Z_{1}$ given $W_{-1}=w_{-1}$, we have 
\begin{eqnarray*}
&&\int_{\mathcal{W}}\frac{\partial }{\partial z_{1}}\left[ \tilde{m}%
_{0,\alpha }^{\ast }\left( y_{\tau },P(w,\alpha _{\theta
_{0}}),w_{-1};P\left( \cdot ,\alpha _{\theta _{0}}\right) \right) \right]
dF_{W}\left( w\right) \\
&=&\int_{\mathcal{W}_{-1}}\int_{z_{1l}\left( w_{-1}\right) }^{z_{1u}\left(
w_{-1}\right) }\frac{\partial }{\partial z_{1}}\left[ \tilde{m}_{0,\alpha
}^{\ast }\left( y_{\tau },P(w,\alpha _{\theta _{0}}),w_{-1};P\left( \cdot
,\alpha _{\theta _{0}}\right) \right) \right] f_{Z_{1}|W_{-1}}\left(
z_{1}|w_{-1}\right) dz_{1}dF_{W_{-1}}\left( w_{-1}\right) \\
&=&\int_{\mathcal{W}_{-1}}\left. \tilde{m}_{0,\alpha }^{\ast }\left( y_{\tau
},P(w,\alpha _{\theta _{0}}),w_{-1};P\left( \cdot ,\alpha _{\theta
_{0}}\right) \right) f_{Z_{1}|W_{-1}}\left( z_{1}|w_{-1}\right) \right\vert
_{z_{1l}\left( w_{-1}\right) }^{z_{1u}\left( w_{-1}\right)
}dF_{W_{-1}}\left( w_{-1}\right) \\
&-&\int_{\mathcal{W}}\tilde{m}_{0,\alpha }^{\ast }\left( y_{\tau
},P(w,\alpha _{\theta _{0}}),w_{-1};P\left( \cdot ,\alpha _{\theta
_{0}}\right) \right) \frac{\partial \log f_{Z_{1}|W_{-1}}\left(
z_{1}|w_{-1}\right) }{\partial z_{1}}dF_{W}\left( w\right) \\
&=&-\int_{\mathcal{W}}\tilde{m}_{0,\alpha }^{\ast }\left( y_{\tau
},P(w,\alpha _{\theta _{0}}),w_{-1};P\left( \cdot ,\alpha _{\theta
_{0}}\right) \right) \frac{\partial \log f_{Z_{1}|W_{-1}}\left(
z_{1}|w_{-1}\right) }{\partial z_{1}}dF_{W}\left( w\right) .
\end{eqnarray*}

The goal is to show that the above is equal to zero. To this end, define 
\begin{equation*}
\nu \left( u,w_{-1};P(\cdot ,\alpha _{\theta })\right) =E\left[ \frac{%
\partial \log f_{Z_{1}|W_{-1}}\left( Z_{1}|W_{-1}\right) }{\partial z_{1}}%
\bigg |P(W,\alpha _{\theta })=u,W_{-1}=w_{-1}\right] .
\end{equation*}%
By the law of iterated expectations, we have 
\begin{eqnarray}
&&\int_{\mathcal{W}}\tilde{m}_{0}\left( y_{\tau },P(w,\alpha _{\theta
}),w_{-1};P\left( \cdot ,\alpha _{\theta }\right) \right) \nu \left(
P(w,\alpha _{\theta }),w_{-1};P(\cdot ,\alpha _{\theta _{0}})\right)
dF_{W}\left( w\right)  \notag \\
&=&E\left[ \mathds{1}\left\{ Y\leq y_{\tau }\right\} \nu \left( P(W,\alpha
_{\theta }),W_{-1};P(\cdot ,\alpha _{\theta _{0}})\right) \right] .
\label{LIE_in_the_proof}
\end{eqnarray}%
This holds because 
\begin{equation*}
\tilde{m}_{0}\left( y_{\tau },P(W,\alpha _{\theta }),W_{-1};P\left( \cdot
,\alpha _{\theta }\right) \right) =E\left[ \mathds{1}\left\{ Y\leq y_{\tau
}\right\} |\mathcal{I}_{\theta }\right]
\end{equation*}%
where $\mathcal{I}_{\theta }$ is the sub-$\sigma $ algebra generated by $%
\left( P(W,\alpha _{\theta }),W_{-1}\right) $ and $\nu \left( P(W,\alpha
_{\theta }),W_{-1};P(\cdot ,\alpha _{\theta _{0}})\right) $ is a function of 
$\left( P(W,\alpha _{\theta }),W_{-1}\right) $ and is thus $\mathcal{I}%
_{\theta }$ -measurable.

Differentiating (\ref{LIE_in_the_proof}) with respect to $\alpha _{\theta }$
and evaluating the resulting equation at $\theta =\theta _{0},$ we have 
\begin{eqnarray}
&&E\left[ \left. \frac{\partial \tilde{m}_{0}(y_{\tau },P(W,\alpha _{\theta
_{1}}),W_{-1};P\left( \cdot ,\alpha _{\theta _{0}}\right) )}{\partial \alpha
_{\theta _{1}}}\right\vert _{\theta _{1}=\theta _{0}}\nu \left( P(W,\alpha
_{\theta _{0}}),W_{-1};P(\cdot ,\alpha _{\theta _{0}})\right) \right]  \notag
\\
&+&E\left[ \left. \frac{\partial \tilde{m}_{0}(y_{\tau },P(W,\alpha _{\theta
_{0}}),W_{-1};P\left( \cdot ,\alpha _{\theta _{2}}\right) )}{\partial \alpha
_{\theta _{2}}}\right\vert _{\theta _{2}=\theta _{0}}\nu \left( P(W,\alpha
_{\theta _{0}}),W_{-1};P(\cdot ,\alpha _{\theta _{0}})\right) \right]  \notag
\\
&=&E\left\{ \left[ \mathds{1}\left\{ Y\leq y_{\tau }\right\} -m\left(
y_{\tau },P(W,\alpha _{\theta _{0}}),W_{-1}\right) \right] \left. \frac{%
\partial \nu \left( P(W,\alpha _{\theta }),W_{-1};P(\cdot ,\alpha _{\theta
_{0}})\right) }{\partial \alpha _{\theta }}\right\vert _{\theta =\theta
_{0}}\right\} ,  \label{LIE}
\end{eqnarray}%
where we have used Assumption \ref{Assumption_ignore_P_error}(b.ii, b.iii)
to exchange the differentiation with the expectation.

Using (\ref{LIE}) and the law of iterated expectations, we then have 
\begin{eqnarray*}
&&\int_{\mathcal{W}}\frac{\partial }{\partial z_{1}}\left[ \tilde{m}%
_{0,\alpha }^{\ast }\left( y_{\tau },P(w,\alpha _{\theta
_{0}}),w_{-1};P\left( \cdot ,\alpha _{\theta _{0}}\right) \right) \right]
dF_{W}\left( w\right) \\
&=&-\int_{\mathcal{W}}\tilde{m}_{0,\alpha }^{\ast }\left( y_{\tau
},P(w,\alpha _{\theta _{0}}),w_{-1};P\left( \cdot ,\alpha _{\theta
_{0}}\right) \right) \frac{\partial \log f_{Z_{1}|W_{-1}}\left(
z_{1}|w_{-1}\right) }{\partial z_{1}}dF_{W}\left( w\right) \\
&=&\int_{\mathcal{W}}\tilde{m}_{0,\alpha }^{\ast }\left( y_{\tau
},P(w,\alpha _{\theta _{0}}),w_{-1};P\left( \cdot ,\alpha _{\theta
_{0}}\right) \right) \nu \left( P(w,\alpha _{\theta _{0}}),w_{-1};P(\cdot
,\alpha _{\theta _{0}})\right) dF_{W}\left( w\right) \\
&=&E\left\{ \left[ \mathds{1}\left\{ Y\leq y_{\tau }\right\} -m\left(
y_{\tau },P(W,\alpha _{\theta _{0}}),W_{-1}\right) \right] \left. \frac{%
\partial \nu \left( P(W,\alpha _{\theta }),W_{-1};P(\cdot ,\alpha _{\theta
_{0}})\right) }{\partial \alpha _{\theta }}\right\vert _{\theta =\theta
_{0}}\right\} .
\end{eqnarray*}%
Note that $\left. \frac{\partial \nu \left( P(W,\alpha _{\theta
}),W_{-1};P(\cdot ,\alpha _{\theta _{0}})\right) }{\partial \alpha _{\theta }%
}\right\vert _{\theta =\theta _{0}}$ is a measurable function of $P(W,\alpha
_{\theta _{0}})$, $\partial P(W,\alpha _{\theta _{0}})/\partial \alpha
_{\theta _{0}}$ and $W_{-1}$. In addition, 
\begin{eqnarray*}
&&E\left\{ \mathds{1}\left\{ Y\leq y_{\tau }\right\} |P(W,\alpha _{\theta
_{0}}),\partial P(W,\alpha _{\theta _{0}})/\partial \alpha _{\theta
_{0}},W_{-1}\right\} \\
&=&E\left\{ \mathds{1}\left\{ Y\leq y_{\tau }\right\} |P(W,\alpha _{\theta
_{0}}),W_{-1}\right\} =m\left( y_{\tau },P(W,\alpha _{\theta
_{0}}),W_{-1}\right)
\end{eqnarray*}%
because (i) given $W_{-1},$ $\partial P(W,\alpha _{\theta _{0}})/\partial
\alpha _{\theta _{0}}$ contains only information on $Z_{1}$, and (ii) given $%
W_{-1},$ $Z_{1}$ is independent of $\left( U_{D},U_{1},U_{2}\right) .$ These
imply that 
\begin{equation*}
\int_{\mathcal{W}}\frac{\partial }{\partial z_{1}}\left[ \tilde{m}_{0,\alpha
}^{\ast }\left( y_{\tau },P(w,\alpha _{\theta _{0}}),w_{-1};P\left( \cdot
,\alpha _{\theta _{0}}\right) \right) \right] dF_{W}\left( w\right) =0.
\end{equation*}%
Hence, 
\begin{equation*}
\left. \frac{\partial }{\partial \theta }E\left[ \frac{\partial m_{\theta
_{0}}(y_{\tau },P(W,\alpha _{\theta }),W_{-1})}{\partial z_{1}}\right]
\right\vert _{\theta =\theta _{0}}=0.
\end{equation*}
\end{proof}

\begin{lemma}
\label{param_ave_estimation} Under Assumption \ref{Assumption_T2} given in
the supplementary appendix, we have 
\begin{equation}
T_{2n}(\hat{y}_{\tau },\hat{m},\hat{\alpha})-T_{2}=\mathbb{P}_{n}\psi
_{\partial m_{0}}\left( W,y_{\tau }\right) +\mathbb{P}_{n}\psi
_{m_{0}}\left( Y,W,y_{\tau }\right) +\mathbb{P}_{n}\tilde{\psi}%
_{Q}(Y,y_{\tau })+o_{p}(n^{-1/2}),  \notag
\end{equation}%
where%
\begin{align*}
\psi _{\partial m_{0}}\left( W,y_{\tau }\right) & :=\frac{\partial
m_{0}(y_{\tau },\tilde{W})}{\partial z_{1}}-T_{2}, \\
\psi _{m_{0}}\left( Y,W,y_{\tau }\right) & :=-\left[ \mathds{1}\left\{ Y\leq
y_{\tau }\right\} -m_{0}(y_{\tau },\tilde{W})\right] \times E\left[ \frac{%
\partial \log f_{W}(W)}{\partial z_{1}}\bigg|\tilde{W}\right] ,
\end{align*}%
and 
\begin{equation*}
\tilde{\psi}_{Q}(Y,y_{\tau }):=\left[ \frac{\tau -\mathds{1}\left\{ Y\leq
y_{\tau }\right\} }{f_{Y}(y_{\tau })}\right] \times E\left[ \frac{\partial
f_{Y|\tilde{W}}(y_{\tau }|\tilde{W})}{\partial z_{1}}\right] .
\end{equation*}
\end{lemma}

\begin{proof}[Proof of Lemma \protect\ref{param_ave_estimation}]
We follow \citesupp{newey1994} to first show that 
\begin{equation*}
T_{2,\theta }=E_{\theta }\left[ \frac{\partial m_{\theta }(y_{\tau ,\theta },%
\tilde{W}\left( \alpha _{\theta }\right) )}{\partial z_{1}}\right]
\end{equation*}%
is differentiable at $\theta _{0}$ and then represent the derivative in the
form of (\ref{path_derivative}). For this, it suffices to show that each of
the four derivatives below exists at $\theta =\theta _{0}:$ 
\begin{eqnarray}
&&\frac{\partial }{\partial \theta }E_{\theta }\left[ \frac{\partial
m_{0}(y_{\tau },\tilde{W}\left( \alpha _{0}\right) )}{\partial z_{1}}\right]
;\frac{\partial }{\partial \theta }E\left[ \frac{\partial m_{\theta
}(y_{\tau },\tilde{W}\left( \alpha _{0}\right) )}{\partial z_{1}}\right] ; 
\notag \\
&&\frac{\partial }{\partial \theta }E\left[ \frac{\partial m_{0}(y_{\tau
,\theta },\tilde{W}\left( \alpha _{0}\right) )}{\partial z_{1}}\right] ;%
\frac{\partial }{\partial \theta }E\left[ \frac{\partial m_{0}(y_{\tau },%
\tilde{W}\left( \alpha _{\theta }\right) )}{\partial z_{1}}\right]
\label{pop_moment_decomposition_app}
\end{eqnarray}%
and can be represented in the form of (\ref{path_derivative}).

By Lemma \ref{lemma_hahn_ridder}, the last derivative exists and is equal to
zero at $\theta =\theta _{0}$. We deal with the rest three derivatives in (%
\ref{pop_moment_decomposition_app}) one at a time. Consider the first
derivative.$\ $Under Assumption \ref{Assumption_T2}(a.iii) and (b) with $%
\ell =1$, we have 
\begin{equation*}
\int_{\mathcal{W}}\left[ \left\vert \frac{\partial m_{0}(y_{\tau },\tilde{w}%
\left( \alpha _{0}\right) )}{\partial z_{1}}\right\vert \sup_{\theta \in
\Theta _{0}}\left\vert \frac{\partial f_{W}(w;\theta )}{\partial \theta }%
\right\vert \right] dw<\infty .
\end{equation*}%
So $E_{\theta }\left[ \frac{\partial m_{0}(y_{\tau },\tilde{W}\left( \alpha
_{0}\right) )}{\partial z_{1}}\right] $ is differentiable in $\theta $ over $%
\Theta _{0}$ and 
\begin{equation*}
\left. \frac{\partial }{\partial \theta }E_{\theta }\left[ \frac{\partial
m_{0}(y_{\tau },\tilde{W}\left( \alpha _{0}\right) )}{\partial z_{1}}\right]
\right\vert _{\theta =\theta _{0}}=E\left[ \frac{\partial m_{0}(y_{\tau },%
\tilde{W}\left( \alpha _{0}\right) )}{\partial z_{1}}S(O)\right] =E\left[
\psi _{\partial m_{0}}(W,y_{\tau })S(O)\right] .
\end{equation*}

Now, for the second derivative in (\ref{pop_moment_decomposition_app}),
Theorem 7.2 in \citesupp{newey1994} shows that Assumption \ref{Assumption_T2}
implies the following:

\begin{enumerate}
\item There exists a function $\gamma _{m_{0}}(o)$ and a measure $\hat{F}%
_{m_{0}}$ such that $E[\gamma _{m_{0}}(O)]=0$, $E[\gamma
_{m_{0}}(O)^{2}]<\infty $, and for all $\hat{m}$ such that $\Vert \hat{m}%
-m_{0}\Vert _{\infty }$ is small enough, 
\begin{equation*}
E\left[ \frac{\partial \hat{m}(y_{\tau },P(Z,X,\alpha _{0}),W_{-1})}{%
\partial z_{1}}-\frac{\partial m_{0}(y_{\tau },P(Z,X,\alpha _{0}),W_{-1})}{%
\partial z_{1}}\right] =\int_{\mathcal{O}}\gamma _{m_{0}}(o)d\hat{F}%
_{m_{0}}(o)
\end{equation*}%
where $\mathcal{O}$ is the support of $O.$

\item The approximation below holds: 
\begin{equation*}
\int_{\mathcal{O}}\gamma _{m_{0}}(o)d\hat{F}_{m_{0}}(o)=\frac{1}{n}%
\sum_{i=1}^{n}\gamma _{m_{0}}(O_{i})+o_{p}(n^{-1/2}).
\end{equation*}
\end{enumerate}

\noindent This shows that $\gamma _{m_{0}}(o)$ is the influence function of $%
E\left[ \frac{\partial \hat{m}(y_{\tau },\tilde{W})}{\partial z_{1}}\right] $%
. It remains to show that $\gamma _{m_{0}}\left( \cdot \right) $ equals $%
\psi _{m_{0}}\left( \cdot \right) $ defined in the lemma. Let $f_{Y|\tilde{W}%
}(y|\tilde{w};\tilde{\theta})$ be a parametrization of the conditional
density of $Y$ given $\tilde{W}$ with the true density corresponding to $%
\tilde{\theta}_{0}.$ With some abuse of notation, we let%
\begin{equation*}
m_{\tilde{\theta}}\left( y_{\tau },\tilde{w}\right) :=E_{\tilde{\theta}}%
\left[ \mathds{1}\left\{ Y\leq y_{\tau }\right\} |\tilde{W}=\tilde{w}\right]
:=\int_{-\infty }^{y_{\tau }}f_{Y|\tilde{W}}\left( y|\tilde{w};\tilde{\theta}%
\right) dy.
\end{equation*}

We need to show that $E\left[ \frac{m_{\tilde{\theta}}\left( y_{\tau },%
\tilde{w}\right) }{\partial z_{1}}\right] $ is differentiable in $\tilde{%
\theta}$ over a neighborhood of $\tilde{\theta}_{0}$ and 
\begin{equation}
\left. \frac{\partial }{\partial \tilde{\theta}}E\frac{m_{\tilde{\theta}%
}\left( y_{\tau },\tilde{w}\right) }{\partial z_{1}}\right\vert _{\tilde{%
\theta}=\tilde{\theta}_{0}}=E\left[ \psi _{m_{0}}\left( Y,W,y_{\tau }\right)
S(Y,\tilde{W};\tilde{\theta}_{0})\right] ,  \label{IF for m_theta}
\end{equation}%
where $S(Y,\tilde{W};\tilde{\theta}_{0})$ is the score of the
parametrization $f_{Y|\tilde{W}}(y|\tilde{w};\tilde{\theta})f_{\tilde{W}%
}\left( \tilde{w}\right) $ at $\tilde{\theta}=\tilde{\theta}_{0}.$ The
differentiability of $E\frac{\partial m_{\tilde{\theta}}(y_{\tau },\tilde{W})%
}{\partial z_{1}}$ in $\tilde{\theta}$ holds under Assumption \ref%
{Assumption_T2}(d.i).

To show (\ref{IF for m_theta}), we denote $\left( P\left(
z_{1},w_{-1}\right) ,w_{-1}^{\prime }\right) ^{\prime }$ by $\tilde{w}%
_{-1}^{\prime }=\left( \tilde{w}_{1},\tilde{w}_{-1}^{\prime }\right)
^{\prime }$ so that $\tilde{w}_{1}=P\left( z_{1},w_{-1}\right) $ and $\tilde{%
w}_{-1}\equiv w_{-1}.$ Similarly, we denote $\left( P\left(
Z_{1},W_{-1}\right) ,W_{-1}^{\prime }\right) ^{\prime }$ by $\tilde{W}%
^{\prime }=\left( \tilde{W}_{1},\tilde{W}_{-1}^{\prime }\right) ^{\prime }$
so that $\tilde{W}_{1}=P\left( Z_{1},W_{-1}\right) $ and $\tilde{W}%
_{-1}\equiv W_{-1}.$ Then, using integration by parts and Assumption \ref%
{Assumption_ignore_P_error}(a), we have 
\begin{eqnarray*}
E\frac{\partial m_{\tilde{\theta}}(y_{\tau },\tilde{W})}{\partial z_{1}} &=&E%
\frac{\partial m_{\tilde{\theta}}(y_{\tau },P\left( Z_{1},W_{-1}\right)
,W_{-1})}{\partial z_{1}} \\
&=&\int_{\mathcal{W}}\frac{\partial E_{\tilde{\theta}}\left[ \mathds{1}%
\left\{ Y\leq y_{\tau }\right\} |\tilde{W}=\tilde{w}\right] }{\partial z_{1}}%
dF_{W}\left( w\right) \\
&=&-\int_{\mathcal{W}}E_{\tilde{\theta}}\left[ \mathds{1}\left\{ Y\leq
y_{\tau }\right\} |\tilde{W}=\tilde{w}\right] \frac{\partial \log f_{W}(W)}{%
\partial z_{1}}dF_{W}\left( w\right) .
\end{eqnarray*}

Therefore, 
\begin{eqnarray*}
&&-\frac{\partial }{\partial \tilde{\theta}}E\frac{m_{\tilde{\theta}}\left(
y_{\tau },\tilde{w}\right) }{\partial z_{1}} \\
&=&\frac{\partial }{\partial \tilde{\theta}}E\left[ E\left[ \mathds{1}%
\left\{ Y\leq y_{\tau }\right\} |\tilde{W};\tilde{\theta})\right] \frac{%
\partial \log f_{W}(W)}{\partial z_{1}}\right] \\
&=&\frac{\partial }{\partial \tilde{\theta}}E\left\{ E\left[ \mathds{1}%
\left\{ Y\leq y_{\tau }\right\} |\tilde{W};\tilde{\theta})\right] \cdot E%
\left[ \left. \frac{\partial \log f_{W}(W)}{\partial z_{1}}\right\vert 
\tilde{W}\right] \right\} \\
&=&\frac{\partial }{\partial \tilde{\theta}}\int_{\mathcal{W}}E\left[ %
\mathds{1}\left\{ Y\leq y_{\tau }\right\} |\tilde{W}=\tilde{w};\tilde{\theta}%
)\right] \cdot E\left[ \left. \frac{\partial \log f_{W}(W)}{\partial z_{1}}%
\right\vert \tilde{W}=\tilde{w}\right] dF_{\tilde{W}}\left( \tilde{w}\right)
\\
&=&\int_{\mathcal{W}}\int_{\mathcal{-\infty }}^{y_{\tau }}\frac{\partial
\log f_{Y|\tilde{W}}\left( y|\tilde{w},\tilde{\theta}\right) }{\partial 
\tilde{\theta}}f_{Y|\tilde{W}}\left( y|\tilde{w},\tilde{\theta}\right)
dy\cdot E\left[ \left. \frac{\partial \log f_{W}(W)}{\partial z_{1}}%
\right\vert \tilde{W}=\tilde{w}\right] dF_{\tilde{W}}\left( \tilde{w}\right)
\\
&=&\int_{\mathcal{W}}\int_{\mathcal{-\infty }}^{y_{\tau }}E\left[ \left. 
\frac{\partial \log f_{W}(W)}{\partial z_{1}}\right\vert \tilde{W}=\tilde{w}%
\right] \frac{\partial \log f_{Y|\tilde{W}}\left( y|\tilde{w},\tilde{\theta}%
\right) }{\partial \tilde{\theta}}f_{Y|\tilde{W}}\left( y|\tilde{w},\tilde{%
\theta}\right) dydF_{\tilde{W}}\left( \tilde{w}\right) \\
&=&E\left[ \left\{ \mathds{1}\left\{ Y\leq y_{\tau }\right\} \cdot E\left[
\left. \frac{\partial \log f_{W}(W)}{\partial z_{1}}\right\vert \tilde{W}=%
\tilde{w}\right] \right\} S\left( Y,\tilde{W};\tilde{\theta}_{0}\right) %
\right]
\end{eqnarray*}%
where 
\begin{equation*}
S\left( Y,\tilde{W}\right) =\left. \frac{\partial \log f_{Y|\tilde{W}}\left(
Y|\tilde{W};\tilde{\theta}\right) }{\partial \tilde{\theta}}\right\vert _{%
\tilde{\theta}=\tilde{\theta}_{0}}=\left. \frac{\partial \log f_{Y,\tilde{W}%
}\left( Y,\tilde{W};\tilde{\theta}\right) }{\partial \tilde{\theta}}%
\right\vert _{\tilde{\theta}=\tilde{\theta}_{0}}
\end{equation*}%
satisfying $E\left[ S\left( Y,\tilde{W};\tilde{\theta}_{0}\right) |\tilde{W}%
\right] =0$. Hence, 
\begin{eqnarray*}
&&\left. \frac{\partial }{\partial \tilde{\theta}}E\frac{m_{\tilde{\theta}%
}\left( y_{\tau },\tilde{w}\right) }{\partial z_{1}}\right\vert _{\tilde{%
\theta}=\tilde{\theta}_{0}} \\
&=&-E\left[ \left\{ \left[ \mathds{1}\left\{ Y\leq y_{\tau }\right\}
-m_{0}\left( y_{\tau },\tilde{W}\right) \right] \cdot E\left[ \left. \frac{%
\partial \log f_{W}(W)}{\partial z_{1}}\right\vert \tilde{W}\right] \right\}
S\left( Y,\tilde{W}\right) \right] \\
&=&E\left[ \psi _{m_{0}}\left( Y,W,y_{\tau }\right) S\left( Y,\tilde{W}%
\right) \right]
\end{eqnarray*}%
as desired.

Next, the dominating condition in Assumption \ref{Assumption_T2}(d.ii)
ensures that the third derivative\ in (\ref{pop_moment_decomposition_app})
exists and 
\begin{equation*}
\left. \frac{\partial }{\partial \theta }E\left[ \frac{\partial
m_{0}(y_{\tau ,\theta },\tilde{W}\left( \alpha _{0}\right) )}{\partial z_{1}}%
\right] \right\vert _{\theta =\theta _{0}}=E\left[ \frac{\partial
^{2}m_{0}(y_{\tau },\tilde{W}\left( \alpha _{0}\right) )}{\partial y_{\tau
}\partial z_{1}}\right] \left. \frac{\partial y_{\tau ,\theta }}{\partial
\theta }\right\vert _{\theta =\theta _{0}}.
\end{equation*}%
Given the approximation 
\begin{equation*}
\hat{y}_{\tau }-y_{\tau }=\mathbb{P}_{n}\psi _{Q}(Y,y_{\tau
})+o_{p}(n^{-1/2}),
\end{equation*}%
from Lemma \ref{quantile_an}, we have 
\begin{equation*}
\left. \frac{\partial y_{\tau ,\theta }}{\partial \theta }\right\vert
_{\theta =\theta _{0}}=E\left[ \psi _{Q}(Y,y_{\tau })S(O)\right] .
\end{equation*}%
Hence, 
\begin{eqnarray*}
\left. \frac{\partial }{\partial \theta }E\left[ \frac{\partial
m_{0}(y_{\tau ,\theta },\tilde{W}\left( \alpha _{0}\right) )}{\partial z_{1}}%
\right] \right\vert _{\theta =\theta _{0}} &=&E\left[ \frac{\partial
^{2}m_{0}(y_{\tau },\tilde{W}\left( \alpha _{0}\right) )}{\partial y_{\tau
}\partial z_{1}}\right] E\left[ \psi _{Q}(Y,y_{\tau })S(O)\right] \\
&=&\left\{ E\left[ \frac{\partial f_{Y|\tilde{W}}(y_{\tau }|\tilde{W})}{%
\partial z_{1}}\right] \right\} E\left[ \psi _{Q}(Y,y_{\tau })S(O)\right] \\
&=&E\left[ \tilde{\psi}_{Q}(Y,y_{\tau })S(O)\right] .
\end{eqnarray*}%
To sum up, we have shown that%
\begin{equation*}
\left. E_{\theta }\left[ \frac{\partial m_{\theta }(y_{\tau ,\theta },\tilde{%
W}\left( \alpha _{\theta }\right) )}{\partial z_{1}}\right] \right\vert
_{\theta =\theta _{0}}=E\left\{ \left[ \psi _{\partial m_{0}}\left(
W,y_{\tau }\right) +\psi _{m_{0}}\left( Y,W,y_{\tau }\right) +\tilde{\psi}%
_{Q}(Y,y_{\tau })\right] S(O)\right\} .
\end{equation*}

This, combined with the arguments in \citesupp{newey1994}, yields 
\begin{equation*}
T_{2n}(\hat{y}_{\tau },\hat{m},\hat{\alpha})-T_{2}=\mathbb{P}_{n}\psi
_{\partial m_{0}}\left( W,y_{\tau }\right) +\mathbb{P}_{n}\psi
_{m_{0}}\left( Y,W,y_{\tau }\right) +\mathbb{P}_{n}\tilde{\psi}%
_{Q}(Y,y_{\tau })+o_{p}(n^{-1/2}).
\end{equation*}
\end{proof}

\subsection{Estimation of the Asymptotic Variance of the UNIQUE}

\label{estimation_variance_appendix}

The asymptotic variance $V_{\tau }$ in (\ref{variance_param}) can be
estimated by the plug-in estimator 
\begin{equation}
\hat{V}_{\tau }=\frac{h}{n}\sum_{i=1}^{n}\hat{\psi}_{\Pi _{\tau },i}^{2},
\label{var_param_if}
\end{equation}%
where, by Theorem \ref{uqr_if_param}, 
\begin{eqnarray*}
\hat{\psi}_{\Pi _{\tau },i} &=&\frac{\hat{T}_{2n}}{\hat{f}_{Y}(\hat{y}_{\tau
})^{2}\hat{T}_{1n}}\hat{\psi}_{f,i}(\hat{y}_{\tau })+\frac{\hat{T}_{2n}}{%
\hat{f}_{Y}(\hat{y}_{\tau })^{2}\hat{T}_{1n}}\hat{f}_{Y}^{\prime }(\hat{y}%
_{\tau })\hat{\psi}_{Q,i}(\hat{y}_{\tau }) \\
&+&\frac{\hat{T}_{2n}}{\hat{f}_{Y}(\hat{y}_{\tau })\hat{T}_{1n}^{2}}\hat{\psi%
}_{\partial P,i}+\frac{\hat{T}_{2n}}{\hat{f}_{Y}(\hat{y}_{\tau })\hat{T}%
_{1n}^{2}}\left[ \frac{1}{n}\sum_{j=1}^{n}\frac{\partial ^{2}P(W_{j},\hat{%
\alpha})}{\partial z_{1}\partial \hat{\alpha}^{\prime }}\right] \hat{\psi}%
_{\alpha ,i} \\
&-&\frac{1}{\hat{f}_{Y}(\hat{y}_{\tau })\hat{T}_{1n}}\hat{\psi}_{\partial
m,i}-\frac{1}{\hat{f}_{Y}(\hat{y}_{\tau })\hat{T}_{1n}}\hat{\psi}_{m,i}-%
\frac{1}{\hat{f}_{Y}(\hat{y}_{\tau })\hat{T}_{1n}}\hat{E}\left[ \frac{%
\partial f_{Y|\tilde{W}}(\hat{y}_{\tau }|\tilde{W}\left( \hat{\alpha}\right)
)}{\partial z_{1}}\right] \hat{\psi}_{Q,i}(\hat{y}_{\tau }).
\end{eqnarray*}%
In this equation, $\hat{T}_{2n}=T_{2n}(\hat{y}_{\tau },\hat{m},\hat{\alpha}%
), $ $\hat{T}_{1n}=T_{1n}(\hat{\alpha}),$ 
\begin{eqnarray*}
\hat{\psi}_{f,i}(\hat{y}_{\tau }) &=&K_{h}\left( \left[ Y_{i}-\hat{y}_{\tau }%
\right] \right) -\frac{1}{n}\sum_{j=1}^{n}K_{h}\left( \left[ Y_{j}-\hat{y}%
_{\tau }\right] \right) \\
\hat{\psi}_{Q,i}(\hat{y}_{\tau }) &=&\frac{\tau -\mathds{1}\left\{ Y_{i}\leq 
\hat{y}_{\tau }\right\} }{\hat{f}_{Y}(\hat{y}_{\tau })}, \\
\hat{\psi}_{\partial P,i} &=&\frac{\partial P(W_{i},\hat{\alpha})}{\partial
z_{1}}-\frac{1}{n}\sum_{j=1}^{n}\frac{\partial P(W_{j},\hat{\alpha})}{%
\partial z_{1}}, \\
\hat{\psi}_{\alpha ,i} &=&\left( -\frac{1}{n}\sum_{j=1}^{n}\frac{\left[
P_{\partial }(W_{j},\hat{\alpha})\right] ^{2}W_{j}W_{j}^{\prime }}{P(W_{j},%
\hat{\alpha})\left[ 1-P(W_{j},\hat{\alpha})\right] }\right) ^{-1}\frac{%
P_{\partial }(W_{i},\hat{\alpha})W_{i}\left[ D_{i}-P(W_{i},\hat{\alpha})%
\right] }{P(W_{i},\hat{\alpha})\left[ 1-P(W_{i},\hat{\alpha})\right] }, \\
\hat{\psi}_{\partial m,i} &=&\frac{\partial \hat{m}(\hat{y}_{\tau },\tilde{W}%
_{i}\left( \hat{\alpha}\right) )}{\partial z_{1}}-\hat{T}_{2n} \\
\hat{\psi}_{m,i} &=&-\left( \mathds{1}\left\{ Y_{i}\leq \hat{y}_{\tau
}\right\} -\hat{m}(\hat{y}_{\tau },\tilde{W}_{i}\left( \hat{\alpha}\right)
)\right) \times \hat{E}\left[ \left. \frac{\partial \log f_{W}(W_{i})}{%
\partial z_{1}}\right\vert \tilde{W}_{i}\left( \hat{\alpha}\right) \right] ,
\end{eqnarray*}%
and 
\begin{equation*}
\hat{E}\left[ \frac{\partial f_{Y|\tilde{W}}(\hat{y}_{\tau }|\tilde{W}\left( 
\hat{\alpha}\right) )}{\partial z_{1}}\right] =\frac{1}{n}\sum_{i=1}^{n}%
\frac{\partial \hat{f}_{Y|\tilde{W}\left( \hat{\alpha}\right) }(\hat{y}%
_{\tau }|p,W_{-1,i})}{\partial p}\bigg|_{p=P(W_{i},\hat{\alpha})}\frac{%
\partial P(z,X_{i},\hat{\alpha})}{\partial z_{1}}\bigg|_{z=Z_{i}}.
\end{equation*}

Most of these plug-in estimates are self-explanatory. For example, $\hat{\psi%
}_{\alpha ,i}$ is the estimated influence function for the MLE when $P(W_{i},%
\hat{\alpha})=P(W_{i}\hat{\alpha})$ and $P_{\partial }\left( a\right)
=\partial P\left( a\right) /\partial a.$ If the propensity score function
does not take a linear index form, then we need to make some adjustment to $%
\hat{\psi}_{\alpha ,i}$. We only need to find the influence function for the
MLE, which is an easy task, and then plug $\hat{\alpha}$ into the influence
function.

The only remaining quantity that needs some explanation is $\hat{\psi}%
_{m,i}, $ which involves a nonparametric regression of $\frac{\partial \log
f_{W}(W_{i})}{\partial z_{1}}$ on $\tilde{W}_{i}\left( \hat{\alpha}\right) :=%
\left[ P(W_{i},\hat{\alpha}),W_{-1,i}\right] .$ We let%
\begin{equation*}
\hat{E}\left[ \left. \frac{\partial \log f_{W}(W_{i})}{\partial z_{1}}%
\right\vert \tilde{W}_{i}\left( \hat{\alpha}\right) \right] =-\phi ^{J}(%
\tilde{W}_{i}(\hat{\alpha}))^{\prime }\left( \sum_{\ell =1}^{n}\phi ^{J}(%
\tilde{W}_{\ell }(\hat{\alpha}))\phi ^{J}(\tilde{W}_{\ell }(\hat{\alpha}%
))^{\prime }\right) ^{-1}\sum_{\ell =1}^{n}\frac{\partial \phi ^{J}(\tilde{W}%
_{\ell }(\hat{\alpha}))}{\partial z_{1}}.
\end{equation*}%
To see why this may be consistent for $E\left[ \left. \frac{\partial \log
f(W_{i})}{\partial z_{1}}\right\vert \tilde{W}_{i}\left( \alpha _{0}\right) %
\right] ,$ we note that using integration by parts, the above is just a
series approximation to $E\left[ \left. \frac{\partial \log f_{W}(W_{i})}{%
\partial z_{1}}\right\vert \tilde{W}_{i}\left( \alpha _{0}\right) \right] $.

The consistency of $\hat{V}_{\tau }$ can be established by using the uniform
law of large numbers. The arguments are standard but tedious. We omit the
details here.$\ $

\subsection{Unconditional Instrumental Quantile Estimation under
Nonparametric Propensity Score}

\label{non_parametric_ps_appendix}

We drop the parametric specification of the propensity score in Assumption %
\ref{Assumption_parametric_ps} and estimate the propensity score
non-parametrically using the series method. With respect to the results in
Section\ \ref{estimation}, we only need to modify Lemma \ref%
{ps_estimation_param}, since Lemma \ref{param_ave_estimation} shows that we
do not need to account for the error from estimating the propensity score.

Let $\hat{P}(w)$ denote the nonparametric series estimator of $P(w)$. The
estimator of $T_{1}:=E\left[ \frac{\partial P(W)}{\partial z_{1}}\right] $
is now 
\begin{equation*}
T_{1n}(\hat{P}):=\frac{1}{n}\sum_{i=1}^{n}\frac{\partial \hat{P}(w)}{%
\partial z_{1}}\bigg |_{w=W_{i}}.
\end{equation*}%
The estimator of $T_{2}$ is the same as in (\ref{t2n}) but with $P(W_{i},%
\hat{\alpha})$ replaced by $\hat{P}(W_{i}):$ 
\begin{equation}
T_{2n}(\hat{y}_{\tau },\hat{m},\hat{P}):=\frac{1}{n}\sum_{i=1}^{n}\frac{%
\partial \hat{m}(\hat{y}_{\tau },\hat{P}(W_{i}),W_{-1,i})}{\partial z_{1}},
\end{equation}%
where, as in (\ref{t2n}), $\hat{m}$ is the series estimator of $m.$ The
formula is the same as before, and we only need to replace $P(W_{i},\hat{%
\alpha})$ by $\hat{P}(W_{i})$. The nonparametric UNIQUE becomes 
\begin{eqnarray}
\hat{\Pi}_{\tau }(\hat{y}_{\tau },\hat{f}_{Y},\hat{m},\hat{P}):= &&-\frac{1}{%
\hat{f}_{Y}(\hat{y}_{\tau })}\left[ \frac{1}{n}\sum_{i=1}^{n}\frac{\partial 
\hat{P}(W_{i})}{\partial z_{1}}\right] ^{-1}\frac{1}{n}\sum_{i=1}^{n}\frac{%
\partial \hat{m}(\hat{y}_{\tau },\hat{P}(W_{i}),W_{-1,i})}{\partial z_{1}} 
\notag \\
&=&-\frac{1}{\hat{f}_{Y}(\hat{y}_{\tau })}\frac{T_{2n}(\hat{y}_{\tau },\hat{m%
},\hat{P})}{T_{1n}(\hat{P})}.
\end{eqnarray}

The following lemma follows directly from Theorem 7.2 of \citesupp{newey1994}%
.

\begin{lemma}
\label{nonparam_ps_estimation} Let Assumption \ref{Assumption_ignore_P_error}%
$\left( a\right) $ and Assumption \ref{Assumption_T2}$(a,c)$ hold. Assume
further that $P(z,x)$ is continuously differentiable with respect to $z_{1}$
for all orders, and that there is a constant $C$ such that $\left\vert
\partial ^{\ell }P(z,x)/\partial z_{1}^{\ell }\right\vert \leq C^{\ell }$
for all $\ell \in \mathbb{N}$.$\ $Then 
\begin{equation*}
T_{1n}(\hat{P})-T_{1}=\mathbb{P}_{n}\psi _{\partial P}\left( W\right) +%
\mathbb{P}_{n}\psi _{P}\left( D,W\right) +o_{p}(n^{-1/2}),
\end{equation*}%
where we define 
\begin{equation*}
\psi _{\partial P}\left( W\right) :=\frac{\partial P(W)}{\partial z_{1}}%
-T_{1}
\end{equation*}%
and 
\begin{equation*}
\psi _{P}\left( D,W\right) :=-\left( D-P(W)\right) \times \frac{\partial
\log f_{W}(W)}{\partial z_{1}}.
\end{equation*}
\end{lemma}

Using a proof similar to that of Lemma \ref{param_ave_estimation}, we can
show that $T_{2n}(\hat{y}_{\tau },\hat{m},\hat{P})$ and $T_{2n}(\hat{y}%
_{\tau },\hat{m},P)$ have the same influence function. That is, we have 
\begin{equation*}
T_{2n}(\hat{y}_{\tau },\hat{m},\hat{P})-T_{2}=\mathbb{P}_{n}\psi _{\partial
m_{0}}(W,y_{\tau })+\mathbb{P}_{n}\psi _{m_{0}}\left( Y,W,y_{\tau }\right) +%
\mathbb{P}_{n}\tilde{\psi}_{Q}(Y,y_{\tau })+o_{p}(n^{-1/2}),
\end{equation*}%
where 
\begin{eqnarray*}
\psi _{\partial m_{0}}(W,y_{\tau }):= &&\frac{\partial m_{0}(y_{\tau
},P(W),W_{-1})}{\partial z_{1}}-T_{2}, \\
\psi _{m_{0}}\left( Y,W,y_{\tau }\right) := &&-\left( \mathds{1}\left\{
Y\leq y_{\tau }\right\} -m_{0}(y_{\tau },P(W),W_{-1})\right) \times E\left[ 
\frac{\partial \log f_{W}(W)}{\partial z_{1}}\bigg|P(W),W_{-1}\right] ,
\end{eqnarray*}%
and 
\begin{equation*}
\tilde{\psi}_{Q}(Y,y_{\tau })=\psi _{Q}(Y,y_{\tau })\times E\left[ \frac{%
\partial f_{Y|P(W),W_{-1}}(y_{\tau }|P(W),W_{-1})}{\partial z_{1}}\right] .
\end{equation*}

Given the asymptotic linear representations of $T_{1n}(\hat{P})-T_{1}$ and $%
T_{2n}(\hat{y}_{\tau },\hat{m},\hat{P})-T_{2},$ we can directly use Lemma %
\ref{param_ave_estimation}, together with Lemma \ref{two_step_density}, to
obtain an asymptotic linear representation of $\hat{\Pi}_{\tau }(\hat{y}%
_{\tau },\hat{f}_{Y},\hat{m},\hat{P})$.

\begin{theorem}
\label{uqr_if_nonparam} Under the assumptions of Lemmas \ref%
{two_step_density}, \ref{param_ave_estimation}, and \ref%
{nonparam_ps_estimation}, we have 
\begin{eqnarray}
\hat{\Pi}_{\tau }-\Pi _{\tau } &=&\frac{T_{2}}{f_{Y}(y_{\tau })^{2}T_{1}}%
\left[ \mathbb{P}_{n}\psi _{f_{Y}}(Y,y_{\tau })+B_{f_{Y}}(y_{\tau })\right] +%
\frac{T_{2}}{f_{Y}(y_{\tau })^{2}T_{1}}f_{Y}^{\prime }(y_{\tau })\mathbb{P}%
_{n}\psi _{Q}(Y,y_{\tau })  \notag  \label{est_param_pi_decom_6} \\
&+&\frac{T_{2}}{f_{Y}(y_{\tau })T_{1}^{2}}\mathbb{P}_{n}\psi _{\partial
P}\left( W\right) +\frac{T_{2}}{f_{Y}(y_{\tau })T_{1}^{2}}\mathbb{P}_{n}\psi
_{P}\left( D,W\right)  \notag \\
&-&\frac{1}{f_{Y}(y_{\tau })T_{1}}\mathbb{P}_{n}\psi _{\partial
m_{0}}(W,y_{\tau })-\frac{1}{f_{Y}(y_{\tau })T_{1}}\mathbb{P}_{n}\psi
_{m_{0}}\left( Y,W,y_{\tau }\right)  \notag \\
&-&\frac{1}{f_{Y}(y_{\tau })T_{1}}\mathbb{P}_{n}\tilde{\psi}_{Q}(Y,y_{\tau
})+R_{\Pi },
\end{eqnarray}%
where 
\begin{eqnarray}
R_{\Pi } &=&O_{p}\left( |\hat{f}_{Y}(\hat{y}_{\tau })-f_{Y}(y_{\tau
})|^{2}\right) +O_{p}(n^{-1})+O_{p}\left( n^{-1/2}|\hat{f}_{Y}(\hat{y}_{\tau
})-f_{Y}(y_{\tau })|\right)  \notag  \label{est_param_pi_decom_remainder} \\
&&+o_{p}(n^{-1/2}h^{-1/2})+o_{p}(h^{2}).
\end{eqnarray}

Furthermore, under Assumption \ref{Assumption_rate}, $\sqrt{nh}R_{\Pi
}=o_{p}(1).$
\end{theorem}

We summarize the results of Theorem \ref{uqr_if_nonparam} in a single
equation: 
\begin{equation*}
\hat{\Pi}_{\tau }-\Pi _{\tau }=\mathbb{P}_{n}\psi _{\Pi _{\tau }}\left(
O\right) +\tilde{B}_{f_{Y}}(y_{\tau })+o_{p}(n^{-1/2}h^{-1/2}),
\end{equation*}%
where $\psi _{\Pi _{\tau }}$ collects all the influence functions in (\ref%
{est_param_pi_decom_6}) except for the bias, $R_{\Pi }$ is absorbed in the $%
o_{p}(n^{-1/2}h^{-1/2})$ term, and 
\begin{equation*}
\tilde{B}_{f_{Y}}(y_{\tau }):=\frac{T_{2}}{f_{Y}(y_{\tau })^{2}T_{1}}%
B_{f_{Y}}(y_{\tau }).
\end{equation*}%
The bias term is $o_{p}(n^{-1/2}h^{-1/2})$ by Assumption \ref%
{Assumption_rate}. The following corollary provides the asymptotic
distribution of $\hat{\Pi}_{\tau }$.

\begin{corollary}
\label{corollary_nonparam}Under the assumptions of Theorem \ref%
{uqr_if_nonparam}, 
\begin{equation*}
\sqrt{nh}\left( \hat{\Pi}_{\tau }-\Pi _{\tau }\right) =\sqrt{n}\mathbb{P}_{n}%
\sqrt{h}\psi _{\Pi _{\tau }}+o_{p}(1)\Rightarrow \mathcal{N}(0,V_{\tau }),
\end{equation*}%
where 
\begin{equation*}
V_{\tau }=\lim_{h\downarrow 0}E\left[ h\psi _{\Pi _{\tau }}^{2}\right] .
\end{equation*}
\end{corollary}

The asymptotic variance takes the same form as the asymptotic variance in
Corollary \ref{corollary_param}. Estimating the asymptotic variance and
testing for a zero unconditional effect are entirely similar to the case
with a parametric propensity score. We omit the details to avoid repetition
and redundancy. From the perspective of implementation, there is no
substantive difference between a parametric approach and a nonparametric
approach to propensity score estimation.

\bibliographystylesupp{aea} \bibliographysupp{references}

\end{document}